\documentclass[12pt]{article}
\usepackage{amsmath, amssymb}
\usepackage[margin=0.8in]{geometry}
\usepackage{setspace}
\usepackage[round]{natbib}
\usepackage{subcaption}
\usepackage{rotating}
\usepackage[usenames,dvipsnames]{xcolor}
\usepackage{cancel}
\usepackage[font={small,it}]{caption}
\usepackage{enumitem}
\usepackage{ragged2e}
\usepackage{collcell}
\usepackage{datatool}
\usepackage{longtable}
\usepackage[export]{adjustbox}
\usepackage{graphicx}
\usepackage{epsfig}
\usepackage{amsthm}
\usepackage{amsmath}
\usepackage{amsmath,graphicx,multirow,tikz}
\usepackage{eulervm}
\usepackage[usenames,dvipsnames]{xcolor}
\usepackage{array}
\newcolumntype{C}[1]{>{\centering\arraybackslash}p{#1}}
\usepackage{appendix}
\usepackage{amssymb}
\usepackage{amscd}
\usepackage{latexsym}
\usepackage[all]{xy}
\usepackage{mathtools}
\usepackage{multirow} 
\usepackage{setspace}

\usepackage[space]{grffile} 

\usepackage{rotating}
\usepackage{color}

\usepackage{url}
\usepackage{lipsum}
\newtheorem{theorem}{Theorem}
\newtheorem{definition}{Definition}

\newtheorem{lemma}{Lemma}
\newtheorem{claim}[theorem]{Claim}

\DeclarePairedDelimiter\ceil{\lceil}{\rceil}
\DeclarePairedDelimiter\floor{\lfloor}{\rfloor}

\newcommand{\BZ}{{\mathbb Z}}

\usepackage[nokeyprefix]{refstyle}
\usepackage{varioref}
\usepackage{xr-hyper}
\usepackage{hyperref}
\begin{document}

\begin{center}
\Large Design of the Millennium Villages Project Sampling Plan: \\
a simulation study for a multi-module survey \normalsize \\
\bigskip
Shira Mitchell, Rebecca Ross, Susanna Makela, Elizabeth A. Stuart, Avi Feller, Alan M. Zaslavsky, Andrew Gelman \\ 
\end{center}

\begin{abstract}
The Millennium Villages Project (MVP) is a ten-year integrated rural development project implemented in ten sub-Saharan African sites. At its conclusion we will conduct an evaluation of its causal effect on a variety of development outcomes, measured via household surveys in treatment and comparison areas. Outcomes are measured by six survey modules, with sample sizes for each demographic group determined by budget, logistics, and the group's vulnerability. We design a sampling plan that aims to reduce effort for survey enumerators and maximize precision for all outcomes. We propose two-stage sampling designs, sampling households at the first stage, followed by a second stage sample that differs across demographic groups. Two-stage designs are usually constructed by simple random sampling (SRS) of households and proportional within-household sampling, or probability proportional to size sampling (PPS) of households with fixed sampling within each. No measure of household size is proportional for all demographic groups, putting PPS schemes at a disadvantage. The SRS schemes have the disadvantage that multiple individuals sampled per household decreases efficiency due to intra-household correlation. We conduct a simulation study (using both design- and model-based survey inference) to understand these tradeoffs and recommend a sampling plan for the Millennium Villages Project. Similar design issues arise in other studies with surveys that target different demographic groups.
\end{abstract}

\section{Background}\label{background}

The Millennium Villages Project (MVP) is an economic development project that targets rural populations across ten countries in sub-Saharan Africa, implementing a multi-sector package of interventions at a village level \citep{Sachs_McArthur2005, Sanchez2007}. See \citet{MVP_protocol_Lancet} for background on the project, study site selection, outcomes of interest, and a comprehensive description of the plan to evaluate its effectiveness. \citet{MitchellGelman2015} describe our plan for causal inference about the MVP's effect on a variety of development outcomes measured in different demographic groups. These outcomes will be measured via survey modules administered in both treatment and comparison villages. 

A design analysis described in \citet{MitchellGelman2015} was used to recommend the number of control villages and magnitude of sampling in each. Next, we must determine how to select households and individuals within households. We propose a two-stage sample: households will be sampled in stage I, followed by individuals within households in stage II \citep{Lohr2010, SarndalBook}. In the first stage, we must decide between simple random sampling and probability proportional to size sampling of households. Because the project operates at the village level, a sampling plan that efficiently estimates outcome means per village is an efficient sampling plan for the overall causal evaluation. In this paper we conduct a simulation study to decide on a sampling plan for estimating finite population village means. 

We aim to minimize the design effect, the ratio between the actual and effective sample sizes. One factor in determining the efficiency of a sampling design is the intraclass correlation, i.e. the correlation among individuals within a household. If more than one individual is sampled per household, the intraclass correlation increases the design effect, reducing the effective sample size relative to the actual sample size. 

Another factor in the efficiency of a sampling design is the distribution of individuals' sampling probabilities. Sampling probabilities can be optimized for a specific outcome, e.g. by sampling with probability approximately proportional to the outcome \citep[p.88]{SarndalBook}. However, with many outcomes of interest, such tailored optimization is difficult or impossible. Therefore, a {\it self-weighted} sample design is preferred, such that all individuals are sampled with equal probability \citep[p.287]{Kish92,Lohr2010}. Such samples are representative without weighting adjustments, and unbiased point estimates can be obtained from standard statistical procedures. 

Given a fixed precision, we aim to minimize time and resources for the survey enumerator teams. This includes minimizing the numbers of people surveyed (i.e. the actual sample size), but also considering the number of households visited, and the effort required to prepare a sampling frame. To conduct the first stage of sampling, a scheme that samples households with equal probability only requires a list of all households with GPS coordinates identifying their locations. However, a scheme which samples households with probability proportional to size requires some measure of household size (e.g. the total number of household members). This additional piece of information requires more effort for enumerators, especially for larger villages with many households. After either method of first stage sampling, we will conduct a demographic census in the sampled households to create the sampling frame for the second stage. 

In this paper we conduct a simulation study to understand the tradeoffs between simple random sampling and probability proportional to size sampling of households in the context of the MVP evaluation. Additionally, our simulations explore design-based versus model-based inference, a dichotomy which has implications for our the analysis of our outcome data.

\section{Outcomes and survey modules}
The Millennium Villages Project (MVP) defines 51 outcomes of interest, including measures of poverty alleviation, agriculture, education, gender equality, health, environmental sustainability, and infrastructure \citep{MVP_protocol_Lancet}. These outcomes are measured in six different survey modules, whose content is discussed in Section \ref{surveys} of \citet{MVP_protocol_Lancet}.  These modules include:
\begin{itemize}
\item a {\bf household survey}, administered to all household heads (or other knowledgeable household members) within the sampled households;
\item a sex-specific {\bf adult survey}, administered to men and women of reproductive age (15 to 49 years) within the sampled households;
\item within the adult-female survey, a {\bf birth history} section, administered to women of reproductive age (15 to 49 years) both in the sampled households and in additional sampled households to reach sample size sufficient for estimating child mortality;
\item a {\bf nutrition survey}, administered to men and women age 15 to 49 years in sampled households;
\item {\bf blood (malaria and anemia)} testing, administered to four age-sex groups in sampled households: children age 6 to 59 months, school-aged children (5 to 14 years old), men age 15 to 49 years, women age 15 to 49 years; and 
\item {\bf anthropometry} measurements, administered among children age 6 to 59 months in sampled households.
\end{itemize}
For each module and age-sex group combination, the project has budgeted a target sample size based on a combination of budget, logistics, and relative importance of different vulnerable populations and intervention beneficiaries.

\section{Sampling plans considered}

For the purpose of our simulation study, we consider all survey modules except for birth history and the nutrition survey. Our sampling will be performed in two phases. First, we will sample households using either simple random sampling (SRS, without replacement) or probability proportional to size sampling (PPS, with replacement), with household size defined as $N_{h,\textrm{total}}$, the number of household members under 50 years old in household $h$. Let $s_I$ be the set of (unique) sampled households. In the PPS scheme, we use $r_I$ to denote the set of sampled households with repeats. Let $n_I = |s_I|$ be the number of households sampled without replacement in the SRS scheme and let $m_I = |r_I|$ be the number of households sampled with replacement in the PPS scheme. We let $n_I = m_I = 300$ based on the project's previous survey rounds and budget for the final survey round. 

To describe the within-household sampling plans for each survey module, we use the following notation. Let $N_h$ be the total number of people in household $h$ that are in the target age-sex group for a particular module. Let $n_h \leq N_h$ be the number of people in household $h$ that we sample and survey. For example, if considering the anthropometry module, then $N_h$ is the number of children under five years of age in household $h$ and $n_h$ is the number of those sampled for the anthropometry module. Let $N = \sum_{h} N_h$ be the total number of people in the sampling frame (an MV1 or a control village) that are in the module's target age-sex group. 

We now outline the within-household sampling plans considered in our simulation study.

\subsubsection*{Adult, anthropometry, and blood modules}
For each module and age-sex group combination, the project has budgeted a target sample size, $n_{\textrm{target}}$:
\begin{itemize}
\item {\bf adult survey} - $n_{\textrm{target}} = 400$ men and $n_{\textrm{target}} = 400$ women of reproductive age (15 to 49 years);
\item {\bf blood (malaria and anemia)} - $n_{\textrm{target}} = 300$ children age 6 to 59 months, $n_{\textrm{target}} = 100$ school-aged children (5 to 14 years old), $n_{\textrm{target}} = 100$ men age 15 to 49 years, $n_{\textrm{target}} = 100$ women age 15 to 49 years; and 
\item {\bf anthropometry} - $n_{\textrm{target}} = 300$ children age 6 to 59 months.
\end{itemize}
\noindent The second stage sampling schemes we consider in this simulation study are, for a given module and age-sex group:
\begin{itemize}
\item For {\bf SRS} sampling of households - combine all $N_{s_I} = \sum_{h \in s_I} N_h$ people in the sampled households in the target age-sex group. If $N_{s_I} \leq n_{\textrm{target}}$, then survey all. Otherwise, we consider two options:
\begin{itemize}
\item {\bf stratify} by sampling $n_h$ individuals from each household, where $n_h$ is proportional (up to rounding) to $N_h$, and the constant of proportionality is determined by the total in the sampled households,
$$n_h = \textrm{round} \left ( N_h * \frac{n_{\textrm{target}}}{N_{s_I}} \right )\textrm{; or}$$
\item take an equal-probability {\bf systematic sample} of $n_{\textrm{target}}$ people. We order the households randomly, and people (in the module's target age-sex group) within households randomly, so that the people within a household are listed consecutively. We then take a sample using the fractional interval method described in \citet[p.77]{SarndalBook} and Appendix \ref{systematic_sampling_proofs}. This procedure enables us to control sample sizes and spread the sample across households such that the sample size in a household is always either the ceiling or the floor of the expected sample size in that household under simple random sampling (see Appendix \ref{systematic_sampling_proofs}). Conceptually, this is similar to stratifying on household, except that there is dependence of the samples between strata (i.e. households). 

\end{itemize}
\item For {\bf PPS} sampling of households - if $n_{\textrm{target}} \geq m_I$, sample a fixed number of people, $n_h = 1$, per household (regardless of household size) if available.\footnote{It is possible that a household is sampled without any members of the target age-sex group. Therefore, if $n_h = 1$, then the PPS scheme will result in a smaller sample size than the SRS scheme. Additionally, for the adult module (where $n_{\textrm{target}} = 400$), if $n_h = 1$ then the PPS scheme will at most sample only 300 adults, one per sampled household.} If $n_{\textrm{target}} < m_I$, take a simple random sample of $n_{\textrm{target}}$ households from $r_I$ to obtain a smaller PPS sample of households. Then sample $n_h = 1$ per household if available. 
\end{itemize}

\subsubsection*{Household survey}

For both the SRS and PPS schemes, the household survey module is administered to the head of household in each sampled household.

\subsection{Simulated data}

For each survey module, we simulate one outcome measured by that module. For the household survey we use the total household consumption; for the adult male survey we use the number of days after illness began when the man first sought advice or treatment; for the adult female survey we use the number of times a woman received antenatal care during her most recent pregnancy; for malaria and anemia testing we use hemoglobin blood concentration; for the anthropometry module we use the weight for age z-score. In generating simulated data, we make the simplifying assumption that all individuals in a target age-sex group have non-missing outcomes. For example, we generate antenatal care outcomes for all women of reproductive age.

To generate data, we use the multilevel model
\begin{align}\label{generate_data_continuous}
y_i &\sim \textrm{Normal}(\alpha_{h[i]}, \sigma_y) \textrm{ for individuals } i\\
\alpha_h &\sim \textrm{Normal}(\mu + \beta_1 N_{h,\textrm{total}} + \beta_2 N_{h,\textrm{total}}^2,\ \sigma_\alpha) \textrm{ for households }h, \nonumber
\end{align}
For total household consumption we use a model analogous to model \ref{generate_data_continuous}:\footnote{Model \ref{generate_data_unlogged_consumption} can be motivated by assuming that model \ref{generate_data_continuous} holds for individual-level consumption (this would assume that within a household consumption is identically distributed, not taking into account age-sex differences). This model implies that $t_h = \sum_{i | h[i] = h} y_i \sim \textrm{Normal}(N_{h,\textrm{total}} \alpha_h, N_{h, \textrm{total}} \sigma_y)$ and that the marginal variance of $t_h$ is $N_{h,\textrm{total}}^2 (\sigma_\alpha^2 + \sigma_y^2)$.}
\begin{align}\label{generate_data_unlogged_consumption}
t_h &\sim \textrm{Normal}(\mu + \beta_1 N_{h,\textrm{total}} + \beta_2 N_{h,\textrm{total}}^2,\ N_{h,\textrm{total}} \sigma_t) \textrm{ for households }h.
\end{align}
We also use a model for the log total consumption (which in our data is more Normally distributed than the total consumption),
\begin{align}\label{generate_data_consumption}
\log(t_h) &\sim \textrm{Normal}(\mu + \beta_1 N_{h,\textrm{total}} + \beta_2 N_{h,\textrm{total}}^2,\ \sigma_t) \textrm{ for households }h.
\end{align}
We use the demographic information from the census and the multilevel model with estimated parameter values (from the survey data) to generate simulated populations. If when models \ref{generate_data_continuous}, \ref{generate_data_unlogged_consumption} or \ref{generate_data_consumption} are fit to past survey data, the 50\% posterior interval of $\beta_1$ or $\beta_2$ contains 0, we set the parameter to 0 when simulating populations. This prevents us from using very noisy estimates of coefficients. Within each simulated population, we randomly sample according to the sampling plans described above, and estimate the finite population mean using either model-based or design-based inference.

\section{Bayesian model-based inference}

To generalize from the data to the population, both design-based and model-based inference must take into account how the data are collected. Let $y = (y_1,...,y_N)$ denote data for the population of interest, and $I = (I_1,...,I_N)$ indicators of the observation of $y$, where $I_i = 1$ if $y_i$ is sampled,\footnote{We assume that all units that are sampled are observed.} and $I_i = 0$ if $y_i$ is missing. Let `obs' = $\{ i : I_i = 1\}$ and `mis' = $\{ i : I_i = 0\}$. Thus, the information available is $y_{\textrm{obs}}, I$, and the likelihood is $p(y_{\textrm{obs}}, I | x, \theta, \phi) = \int p(y | x, \theta) p(I | x, y, \phi) d y_{\textrm{mis}}$, where $x$ are observed covariates. Bayesian inference computes the posterior distribution $p(\theta, \phi | x, y_{\textrm{obs}}, I)$ (superpopulation inference) and $p(y_{\textrm{mis}} | x, y_{\textrm{obs}}, I, \theta, \phi)$ (finite population inference). Under the {\it ignorability} condition, these inferences can be simplified to $p(\theta | x, y_{\textrm{obs}})$ and $p(y_{\textrm{mis}} | x, y_{\textrm{obs}}, \theta)$. Ignorability is satisfied if both the {\it missing at random} and {\it distinct parameters} conditions are satisfied \citep[p.202, 206-211]{BDA3}. Missing at random requires that the missingness be independent of the missing values conditional on observed variables and a parameter $\phi$: $p(I | x, y, \phi) = p(I | x, y_{\textrm{obs}}, \phi)$. The distinct parameters condition requires that the parameters of the missingness mechanism ($\phi$) be independent of the parameters of the data generating process ($\theta$), conditional on covariates: $p(\phi | x, \theta) = p(\phi | x)$.

We include design variables such that the data collection mechanism is ignorable with respect to this model. For example, in our SRS-stratified sampling plan, the data collection mechanism is:
$$p(I | x, y, \phi) = 1/\left [ \sum_{\substack{s_I \subseteq \{1,...,N_I \} \\ |s_I| = n_I}} \prod_{h \in s_I} {N_h \choose n_h} \right ] \textrm{ where }n_h = \textrm{round} \left ( N_h \frac{n}{N_{s_I}} \right )$$ if $\exists s_I \subseteq \{1,...,N_I \}$ s.t. $|s_I| = n_I$ and $\sum_{i : h[i] = h} I_i = \textrm{round} \left ( N_h \frac{n}{N_{s_I}} \right )$ for all $h \in s_I$. Otherwise, the probability of missingness pattern $I$ is zero. 

%

Thus, we include as design variables the household identifiers and the $N_h$ (e.g. the number of women per household, if the survey module targets women). Similar computations show that under the SRS-systematic sampling scheme these variables are also sufficient to satisfy missing at random. For the PPS scheme, we also will need the measure of household size used to select the households (e.g. the total number of household members under 50) \citep[p.211]{BDA3}. For simplicity, we fit the same ignorable model for both the SRS and PPS schemes. For the anthropometry, blood, and adult survey modules we fit
\begin{align}\label{ignorable_model}
y_i &\sim \textrm{Normal}(\alpha_{h[i]}, \sigma_y) \textrm{ for individuals } i\\
\alpha_h &\sim \textrm{Normal}(\mu + \beta_1 N_{h,\textrm{total}} + \beta_2 N_{h,\textrm{total}}^2 + \beta_3 N_h + \beta_4 N_h^2,\ \sigma_\alpha) \textrm{ for households }h, \nonumber
\end{align}
For the household survey, we fit models \ref{generate_data_unlogged_consumption} and \ref{generate_data_consumption}. 

Our parameter of interest is the finite population mean $\overline{Y} = \frac{1}{N} \sum_{h = 1}^{N_I} N_h \overline{y}_h$, where $\overline{y}_h = \frac{n_h}{N_h} \overline{y}_{h,obs} + \frac{N_h - n_h}{N_h} \overline{y}_{h,mis}$ \citep[p.205]{BDA3}. We obtain posterior simulations of $\overline{Y}$ as follows: if household $h$ is sampled, we use a simulation of $\alpha_h$ to generate $N_h - n_h$ simulated $y_i$'s. If household $h$ is not sampled, we use simulations of $\mu$ and $\sigma_\alpha$ to simulate a new $\alpha_h$, then generate $N_h$ simulated $y_i$'s. 

\section{Frequentist design-based inference}

We use the {\tt survey} package to compute design-based estimates and variances \citep{Lumley2004}. Though we perform our SRS schemes without replacement, we compute all variances without finite population corrections, using the Horvitz Thompson (Hajek) ratio estimator and its with-replacement variance \citep[p.247]{Lohr2010}.

Our SRS schemes are two-phase rather than two-stage designs, since the sampling within a household depends on which households were sampled in the first stage \citep[p.134-135]{SarndalBook}. This dependence is reflected in the design weights we compute, see below. For the SRS-systematic sampling scheme, the independence assumption of two-stage sampling is also violated, with the sampling in each household dependent on the sampling in other households. Our design-based analysis approximates these two-phase designs with a two-stage analysis. In contrast, in model-based inference the details of the design do not matter in the analysis once we include design variables in our model \citep[p.202, 206-211]{BDA3}.

\subsection{Design weights}

For the SRS-systematic design, the inclusion probabilities are:
\begin{align*}
\pi_{hi} &\equiv P[\textrm{person }i\textrm{ in household }h\textrm{ is sampled}] \\
&= P(h \in s_I) P(i \in s_h | h \in s_I) \\
&= \frac{n_I}{N_I} \sum_{s_I | h \in s_I} P(i \in s_h | h \in s_I, s_I) * P(s_I | h \in s_I) \\
&= \frac{n_I}{N_I} \sum_{s_I | h \in s_I} \underbrace{\min \left ( \frac{n}{N_{s_I}}, 1 \right )}_{(*)} * \frac{1}{{N_I - 1 \choose n_I - 1}}
\end{align*}
For SRS-stratified, we replace $(*)$ with $\min \left ( \frac{\textrm{round} \left ( \frac{n}{N_{s_I}} N_h \right )}{N_h}, 1 \right )$. In the simulations, instead of computing this precisely, we estimate it by randomly sampling $s_I$ such that $h \in s_I$. This avoids the computationally intensive loop over all ${N_I - 1 \choose n_I - 1}$ such sets. Although these weights are not equal for all individuals, because the distributions of household sizes (from the MVP demographic data) have no extreme outliers, in our simulations the weights are nearly equal. 

For the PPS scheme, the inclusion probabilities are:
\begin{align*}
\pi_{hi} &= E_k \left [ P \left ( \textrm{person }i\textrm{ in household }h\textrm{ is sampled } | \textrm{ household }i\textrm{ is chosen } k \textrm{ times}\right ) \right ] \\
&= E_k [ 1 - (1-n_h/N_h)^k ] \textrm{ since we independently subsample a household as many times as it is drawn.} \\
&\textrm{Since }k \sim \mbox{Bin}(m_I, p_h), \textrm{ by its probability generating function, we obtain} \\
&= 1 - (p_h (1-n_h/N_h) + (1-p_h))^{m_I} \\
&= 1 - \left ( 1 - p_h \frac{n_h}{N_h} \right )^{m_I} \\
&\textrm{if }p_h \frac{n_h}{N_h} \textrm{ is small, we can approximate this as: }\\
&= m_I p_h \frac{n_h}{N_h}
\end{align*}
In PPS sampling, $p_h \propto x_h$, where $x_h$ is a measure of household size \citep[p.97]{SarndalBook}. So the PPS weights are:
\begin{align*}
w_{hi} &= \frac{\sum_{h \in U_I} x_h}{m_I x_h} \frac{N_h}{n_h}.
\end{align*}
If $x_h \propto N_h$, and $n_h \propto 1$, then the design is self-weighted. We take $n_h = c$, a constant, but we cannot choose $x_h$ such that $x_h \propto N_h$ for all modules, since the target age-sex groups differ from module to module. We chose $x_h = N_{h,\textrm{total}}$, the number of household members under 50 years of age, because it represented a compromise between the different target age-sex groups. Thus, our weights are $w_{hi} \propto \frac{N_h}{N_{h,\textrm{total}}}$.

\section{Comparisons between sampling schemes: variances and design effects}

We want to compare the PPS and SRS designs (in either the Bayesian model-based or the design-based paradigms). In general, the two schemes will have slightly different sample sizes, making direct comparisons of variances less relevant. For the household survey module, we fix the sample sizes to be equal, and for the adult, anthropometry, and blood modules, we adjust for the differing sample sizes by computing a design effect, defined below.

The household survey module is administered to the heads of households only, not individual members. Therefore, the time cost of the household module is mostly determined by the number of households surveyed. We set up our simulations such that the number of household heads to be interviewed (i.e. sample size) is the same for the SRS and PPS sampling schemes. We first perform a PPS sampling of households. Then, we use the number of unique sampled households to obtain the number of households to sample for the SRS scheme. We then directly compare the variances in estimating $\overline{Y}$, the finite population mean consumption per person.

For the remaining modules, we compute design effects. To define the design effect (often abbreviated as ``deff"), we first introduce the following notation. Let $\widehat{\theta}_\pi = \widehat{\theta}_\pi(I, y_{\textrm{obs}})$ be the estimator of $\theta$ (in our case, $\theta = \overline{Y}$) where $\pi$ is the sampling distribution assumed to have been used in drawing sample $S$. Let $V_{\pi_1} (\widehat{\theta}_{\pi_2}; y)$ be the sampling variance of an estimator of $\theta$ that assumes sampling distribution $\pi_2$, and $\pi_1$ is the distribution with respect to which we want the variance. Let $\widehat{V}_{\pi_1} (\widehat{\theta}_{\pi_2}; \pi_3; I, y_{\textrm{obs}})$ be an estimator where $\pi_3$ is the sampling distribution assumed to have been used in drawing sample $S$. The population design effect is defined as $= V_p(\widehat{\theta}_p; y)/V_{SRS}(\widehat{\theta}_{SRS}; y)$. The estimated design effect is defined as $= \widehat{V}_p(\widehat{\theta}_p; p; y_{\textrm{obs}})/\widehat{V}_{SRS}(\widehat{\theta}_{SRS}; p; y_{\textrm{obs}})$. 

In the design-based setting, we compute design effects assuming sampling with-replacement in both numerator and denominator variances. This is done in the {\tt survey} package by specifying {\tt deff = `replace'}.

For the model-based simulations, we estimate the numerator of the deff with the posterior variance for $\overline{Y}$ from fitting a model that includes enough design variables such that the data collection mechanism is ignorable with respect to this model. This posterior variance includes an implicit finite population correction, so we compute a denominator variance that also includes such a correction:
\begin{align}\label{SRS_variance}
V_{SRS}(\widehat{\theta}_{SRS}; y) &= V_{SRS}(\overline{y}; y) \\
 &= \left ( 1 - \frac{n}{N} \right ) \frac{S^2}{n} \nonumber
\end{align}
where $S^2 = \frac{1}{N-1} \sum_{i = 1}^N (y_i - \overline{Y})^2$. 

To assess our estimated deff in the model-based setting, we compare the posterior variance $V_{p\textrm{-ignorable}}(\theta | y_{\textrm{obs}})$ from fitting an ignorable model with respect to a sampling distribution $p$ to the design-based sampling variance of the posterior means, $E_{p\textrm{-ignorable}}(\theta | y_{\textrm{obs}})$. The latter can be computed by simulation: we sample repeatedly from the full population using distribution $p$, fit the $p$-ignorable model, obtain a posterior mean of $\theta$, and compute the variance of these across the samples from $p$. Fixing one finite population, in Figure \ref{histogram_p} we create a histogram of posterior variances from fitting the $p$-ignorable model to each sample, and indicate with a vertical line the design-based variance of the posterior means, which is computed by simulation. We make the same comparison for $p = $ a simple random sample (and its ignorable model with flat priors and no design variables), and include the closed-form design-based estimate (\ref{SRS_variance}) as a vertical line, in addition to the simulation-computed design-based estimate. See Figure \ref{histogram_SRS}. We see that the posterior variances appear unbiased for the design-based variances.

\begin{figure}[h!]
    \centering
\subcaptionbox{Sampling distribution $p$ is SRS sampling of households followed by an equal-probability systematic sample within households. \label{histogram_p}}
 [0.49\textwidth]{\includegraphics[width=0.49\textwidth]{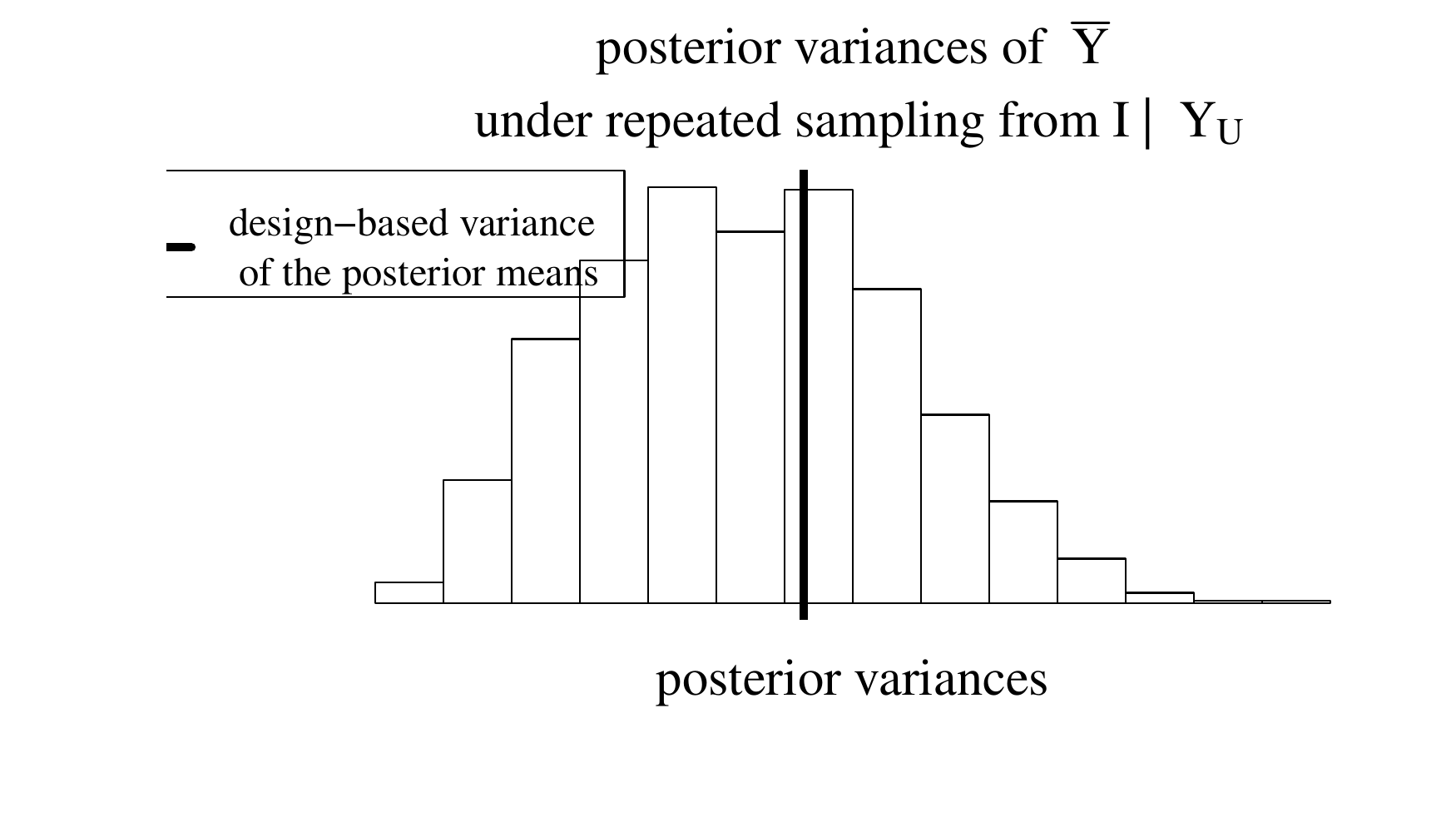}}
\subcaptionbox{Sampling distribution $p$ is SRS sampling of people. \label{histogram_SRS}}
 [0.49\textwidth]{\includegraphics[width=0.49\textwidth]{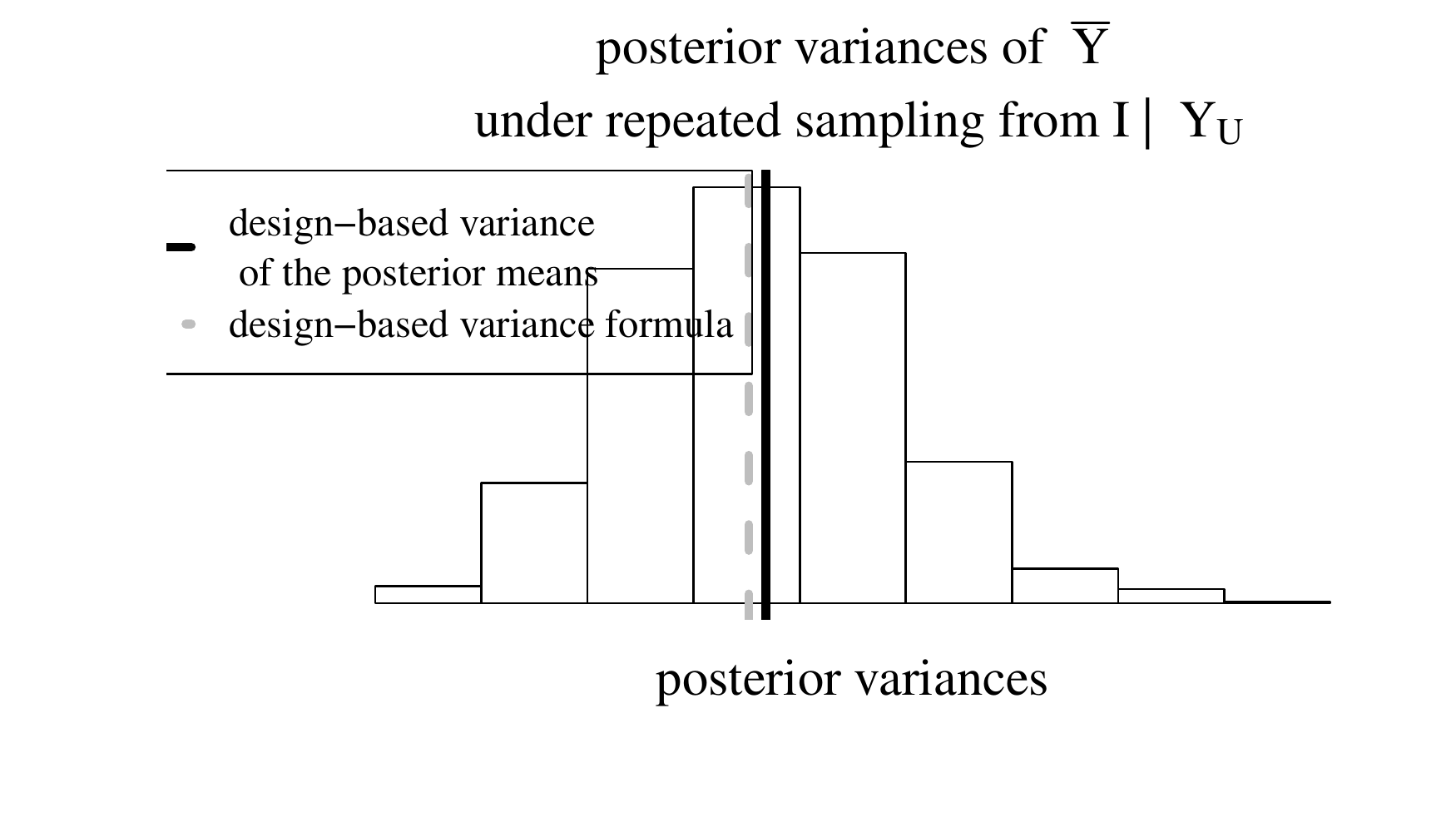}} 
  \caption[]{Fixing one finite population, we show a histogram of posterior variances from fitting a $p$-ignorable model to each sample using sampling distribution $p$, and indicate in a vertical line the design-based variance of the posterior means, which is computed by simulation. When $p$ is simple random sample of people, we also include the closed-form design-based estimate (\ref{SRS_variance}).}\label{Uganda_ACCESS_by_POPD}
\end{figure}

\section{Simulation results}

Our results are displayed in Appendix \ref{survey_sampling_results_appendix}, where we see that neither the SRS nor PPS sampling of households is more efficient (i.e. has a lower design effect) in general. 

We see that for modules with higher target sample sizes, SRS tends to be less efficient. For example, in the under-5 blood ($n_{\textrm{target}} = 300$) and adult ($n_{\textrm{target}} = 400$) modules the SRS scheme is less efficient. One explanation for this observation is the different numbers of people sampled per household in the SRS versus PPS schemes, which has efficiency implications due to the intra-house correlation. In the PPS scheme, the households sampled in the first stage are larger and therefore more likely to include people in the target demographics. In contrast, in the SRS scheme, the sample is often drawn from fewer households, with more people sampled per household. Moreover, the PPS scheme only samples one person per household draw (though this can result in more than one person being sampled per household due to the with-replacement sampling at the first stage). 

For modules where the target sample size is low, there are fewer people sampled per household in the SRS sampling scheme, and the intra-house correlation does not substantially impact the design efficiency. Therefore, because SRS has near-equal individual-level probability of sampling (see the design-weights computed above), its design effect in the absence of household clustering should be close to one. In contrast, the PPS scheme does not have near-equal individual-level sampling probabilities because the measure of household size is not proportional to the target demographic (see the design-weights computed above).

The relative efficiency of SRS versus PPS is similar between design-based and model-based simulations. In the few cases where they differ, design-based results show that SRS has higher design effects than PPS, relative to model-based results. In general, our model-based simulations show more variability across simulations than the design-based simulations. Comparing systematic to stratified sampling at the second stage of the SRS schemes, we see few differences except that stratified sampling tends to have higher variance across simulations. 

\section{Final sampling plan}

As described above, the PPS scheme requires a sampling frame that includes household sizes, whereas the SRS scheme only requires a list of households. Given our results, we cannot justify the additional resources required to collect the more detailed household list for the PPS scheme. Therefore, our sampling scheme will begin with an SRS sample of households. For the second stage of sampling for the adult, anthropometry, and blood modules, we prefer the control over sample size achieved by systematic sampling (as opposed to stratified sampling). 

The household and nutrition modules follow a different sampling scheme.  As mentioned above, the household module is administered to all household heads (or other knowledgeable household members) within the sampled households. The nutrition module consists of a food frequency questionnaire, which takes longer to administer than other modules. We suspect that the within-household correlation is very high for data on food frequency, because household members are likely to eat similar foods. (This intra-house correlation cannot be measured from project data, because the project has always limited this module to one member per household.)  For these reasons, we limit the nutrition module to one adult (age 15 to 49 years) per household.  

\section{Software}

For fitting multilevel models we use Stan in {\tt R}, \citep{Stan, R}.

\clearpage
\newpage

\appendix

\section*{Acknowledgements}

The authors would like to thank the following people for very valuable feedback and ideas: Jeffrey D. Sachs, Joseph K. Blitzstein, Qixuan Chen, Jennifer Hill, Macartan Humphreys, Michael Clemens, Alberto Abadie, Marc Levy, Linda Pistolesi, Keli Liu, Peng Ding, Natalie Exner, Abhishek Chakrabortty, Rachael Meager, and Natalie Bau.  

All mistakes are our own.

\newpage
\clearpage

\section{Properties of systematic sampling}\label{systematic_sampling_proofs}

\begin{definition}[The fractional interval method of systematic sampling]\label{systematic_def}
Consider a population of size $N^*$ consisting of people grouped into $n_I$ households indexed by $h$, with $N_h$ people within-household $h$.  Let $n_{target} < N^*$ be the desired sample size.  Set $a = \frac{N^*}{n_{target}}$.  Order the households randomly, and randomly order the people in the target group within the households.  Let $k = 1,...,N^*$ label the people in this order:
$$\underbrace{1......N_1}_{\textrm{household 1}}\underbrace{(N_1 + 1)......(N_1 + N_2)}_{\textrm{household 2}}............\underbrace{......N^*}_{\textrm{household }n_I}$$
Draw a random real number $\xi$ uniformly between 0 and $a$, $\xi \sim U(0,a)$, and sample all people with $k$ such that 
$$k - 1 < \xi + (j-1) a \leq k \textrm{ for }j = 1,...,n_{target}.$$
\citep[p.77]{SarndalBook}
\end{definition}

\begin{claim}\label{sample_size}
When performing the sampling scheme in Definition \ref{systematic_def}, the sample size will be $n_{target}$.
\end{claim}
\begin{proof}[Proof of Claim \ref{sample_size}]
Since $a \equiv \frac{N^*}{n_{target}}$ and $N^* > n_{target}$, $a > 1$.  Since $k - 1 < x \leq k \Leftrightarrow \ceil*{x} = k$, we can write $\ceil*{\xi + (j-1) a} = k$.  The ceiling function is monotone increasing and $\ceil*{x+1} = \ceil*{x} + 1$, so each time $j$ increases by 1, we get a different value of $k$.
Now we must show that the $k$'s stay in the set $\{1,...,N^*\}$, i.e. those from which we are sampling.  The first $k$ is such that $k - 1 < \xi \leq k$, where $\xi \in (0, a)$.  Since $\xi > 0$, we must have $k \geq 1$. The last $k$ is such that $k - 1 < \xi + (n_{target} - 1) a \leq k$, and we know $\xi + (n_{target} - 1) a \leq \xi + N^* - a < N^*$ because $\xi < a$.  Then $k \leq N^*$.  Thus, since each $j$ maps to a unique $k$, we've proven we get a sample size of exactly $n_{target}$.
%
%
\end{proof}

\bigskip

\begin{claim}\label{within_hh_sample_size}
When performing the sampling scheme in Definition \ref{systematic_def}, the sample size within each household $h$ is always the ceiling or the floor of the expected sample size in household $h$ under simple random sampling: $\frac{N_h}{a}$.
\end{claim}
We first prove the following lemma which is used to prove the above claim:
\begin{lemma}\label{mylemma}
Consider the set $A(x) \equiv \{ j \in \BZ^+ | \xi + (j-1) a \leq x\}$.  The maximum of $A(x)$ is $\frac{x}{a}$ if $\frac{x}{a} \in \BZ$, $\floor*{\frac{x}{a}}$ if $\xi > da$, or $\ceil*{\frac{x}{a}}$ if $\xi \leq da$, where $d \equiv \frac{x}{a} - \floor*{\frac{x}{a}}$, the ``decimal part."
\end{lemma}
\begin{proof}[Proof of Lemma \ref{mylemma}]
If $\frac{x}{a} \in \BZ$, let $j = \frac{x}{a}$, and we see that $\xi + \left ( \frac{x}{a} - 1 \right )a = \xi + x - a \leq x$ because $\xi < a$, so $\frac{x}{a} \in A(x)$.  Increasing $j$ by 1 increases the lefthand side of the inequality by $a$, and since $\xi > 0$, we see this $\frac{x}{a} + 1 \not\in A(x)$.  Therefore, $\frac{x}{a}$ is the maximum.

If $\frac{x}{a} \not\in \BZ$, let $d \equiv \frac{x}{a} - \floor*{\frac{x}{a}}$, the ``decimal part."  We see that 
$\xi + \left ( \floor*{\frac{x}{a}} - 1 \right )a \underset{+ da}{<} \xi + \left (\frac{x}{a} - 1 \right )a = \xi + x - a \underset{+a-\xi}{\leq} x$, so $\floor*{\frac{x}{a}} \in A(x)$.  Increasing $j$ by 1 (to $\ceil*{\frac{x}{a}}$) increases the leftmost side of the inequality by $a$.  If $a \leq da + (a - \xi)$, i.e. if $\xi \leq da$, then $\ceil*{\frac{x}{a}} \in A(x)$.
\end{proof}

\begin{proof}[Proof of Claim \ref{within_hh_sample_size}]

Consider household $h$, of size $N_h$.  Let $k^*$ be the last person before household $h$ in the ordering used by the systematic sampling.  Then $k^* + 1,...,k^* + N_h$ are the indices for all members of household $h$. In order to get the number sampled in household $h$, we consider the maximum of set $A(k^* + N_h)$ (the number of people sampled up through household $h$) and subtract from it the maximum of set $A(k^*)$ (the number of people sampled before household $h$).  
This gives, by Lemma \ref{mylemma}, $\frac{k^* + N_h}{a} - \frac{k^*}{a} = \frac{N_h}{a}$.
\end{proof}

\bigskip

\begin{claim}
The sampling scheme in definition \ref{systematic_def} is self-weighted.
\end{claim}

\begin{proof}
\begin{align*}
P(\textrm{person }k\textrm{ is sampled}) &= P \left ( \exists j \in \left \{1,...,n_{target} \right \} \ s.t. \ k - 1 < \xi + (j-1) a \leq k \right ) \\
&= P \left (\cup_{j=1}^{n_{target}} \big \{ k - 1 - (j-1) a < \xi \leq k - (j-1) a \big \} \right ) \\
&\textrm{each interval is length 1} = [k - (j-1)a] - [k - 1 - (j-1)a]  \\
&\textrm{space between intervals j and j+1 is }a - 1 = [k - 1 - j*a ] - [k - (j-1) a] \\
&\textrm{so by the picture below, we see that }(0,a)\textrm{ has overlaps with the intervals } \\
&\textrm{of length totaling 1.  So since }\xi \sim U(0,a), \\
&= \frac{1}{a}
\end{align*}
See below for a visual, in \textcolor{orange}{orange} is the interval $(0,a)$, which can overlap at most 2 intervals of length 1 (shown as over-braces, with overlaps totaling a length 1:
$$\underbrace{\overbrace{-----}^{j = 1, \textrm{ length}=1}--}_{a}\underbrace{\overbrace{-----}^{j = 2, \textrm{ length}=1}--}_{a}\underbrace{\overbrace{---\textcolor{orange}{ \bf --}}^{j = 3, \textrm{ length}=1}\textcolor{orange}{ \bf--}}_{a}\underbrace{\overbrace{\textcolor{orange}{ \bf---}--}^{j = 4, \textrm{ length}=1}--}_{a}...$$
\end{proof}

\clearpage
\newpage

\section{Survey Sampling Simulation Results}\label{survey_sampling_results_appendix}

\subsection{Design-based results}

\begin{figure}[h!]
    \centering
\subcaptionbox{\label{deff_adult_man_1_per_PPS_draw_design_based_systematic}}
 [0.49\textwidth]{\includegraphics[width=0.34\textwidth]{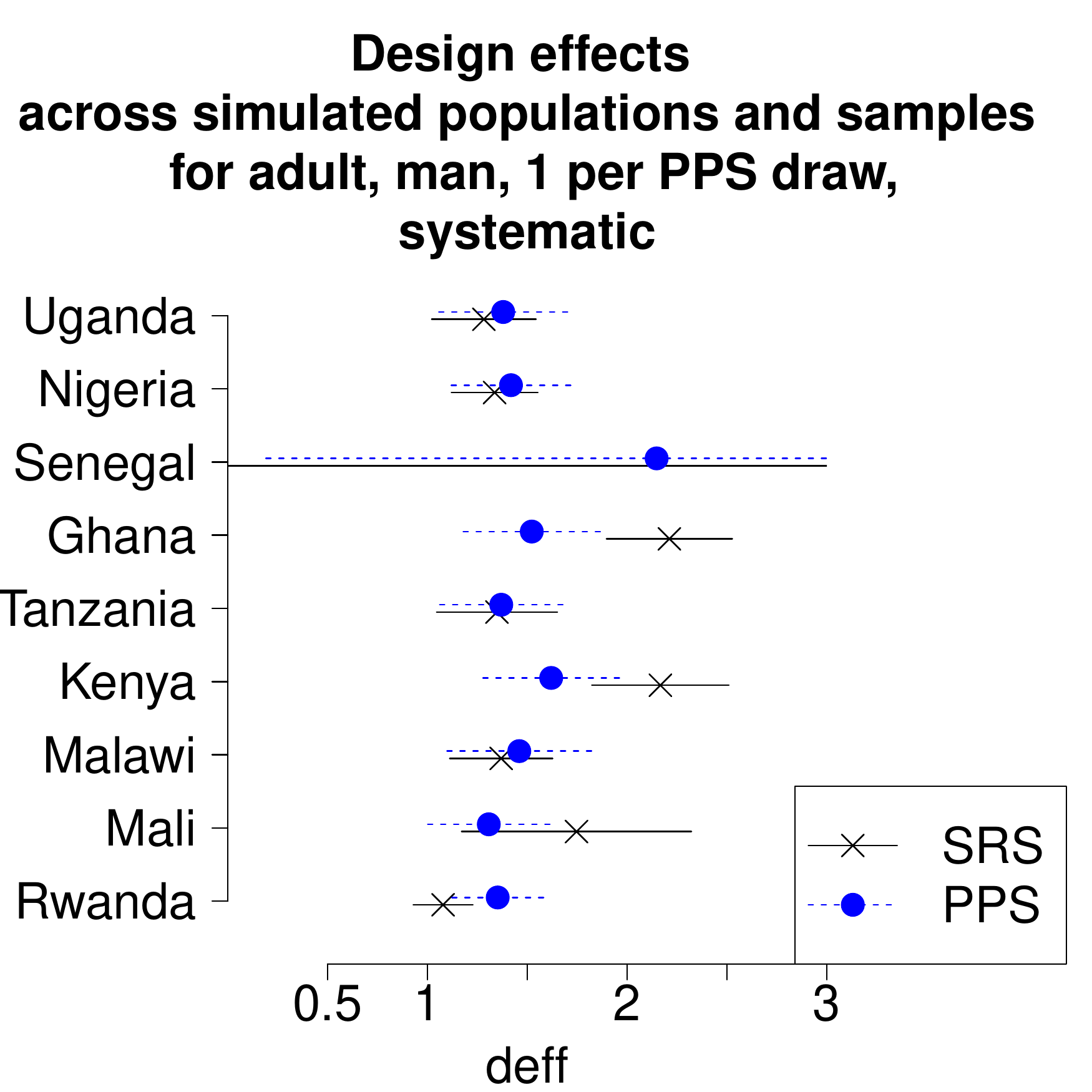}}
\subcaptionbox{\label{deff_adult_woman_1_per_PPS_draw_design_based_systematic}}
 [0.49\textwidth]{\includegraphics[width=0.34\textwidth]{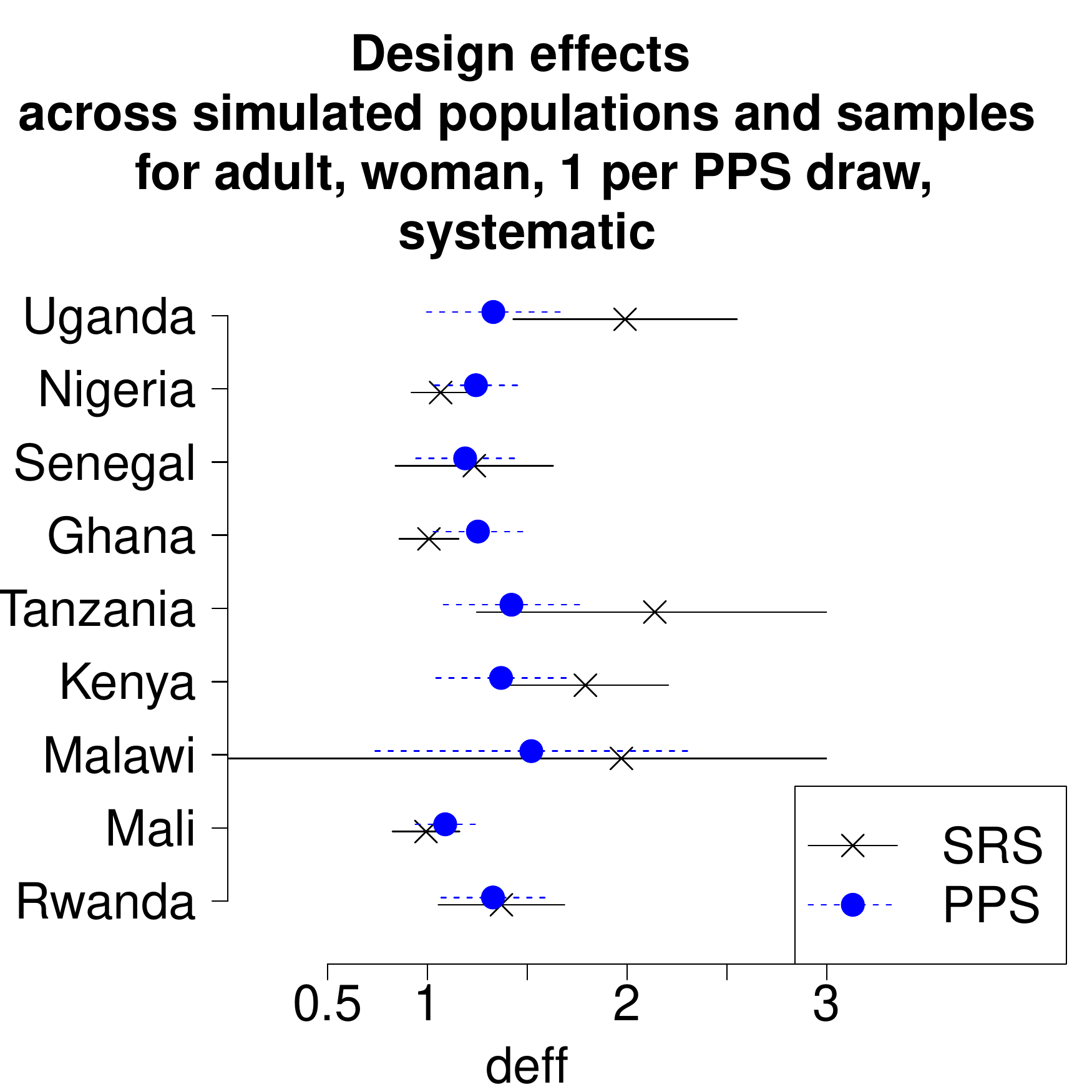}} \\
\subcaptionbox{\label{n_adult_man_1_per_PPS_draw_systematic}}
 [0.49\textwidth]{\includegraphics[width=0.34\textwidth]{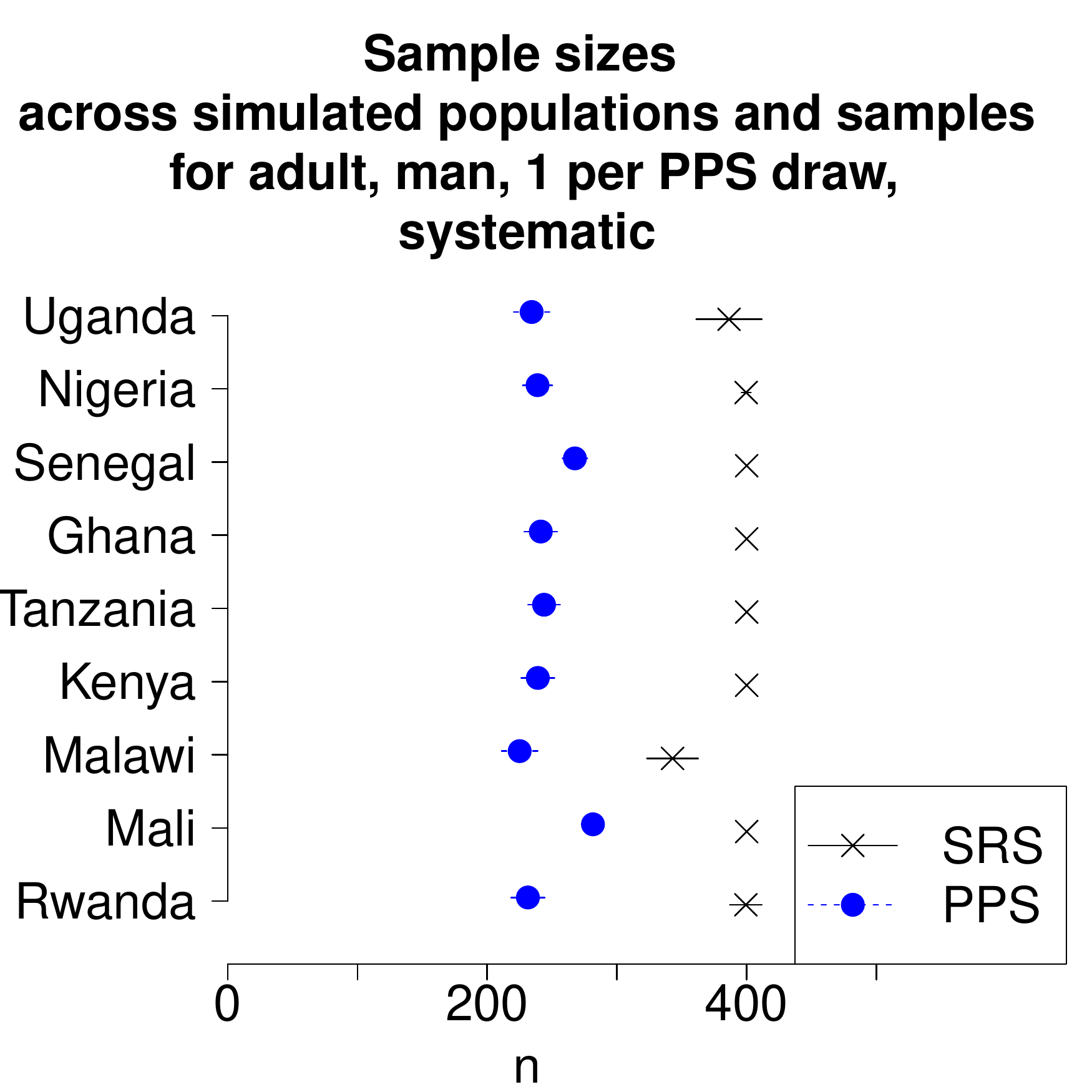}}
\subcaptionbox{\label{n_adult_woman_1_per_PPS_draw_systematic}}
 [0.49\textwidth]{\includegraphics[width=0.34\textwidth]{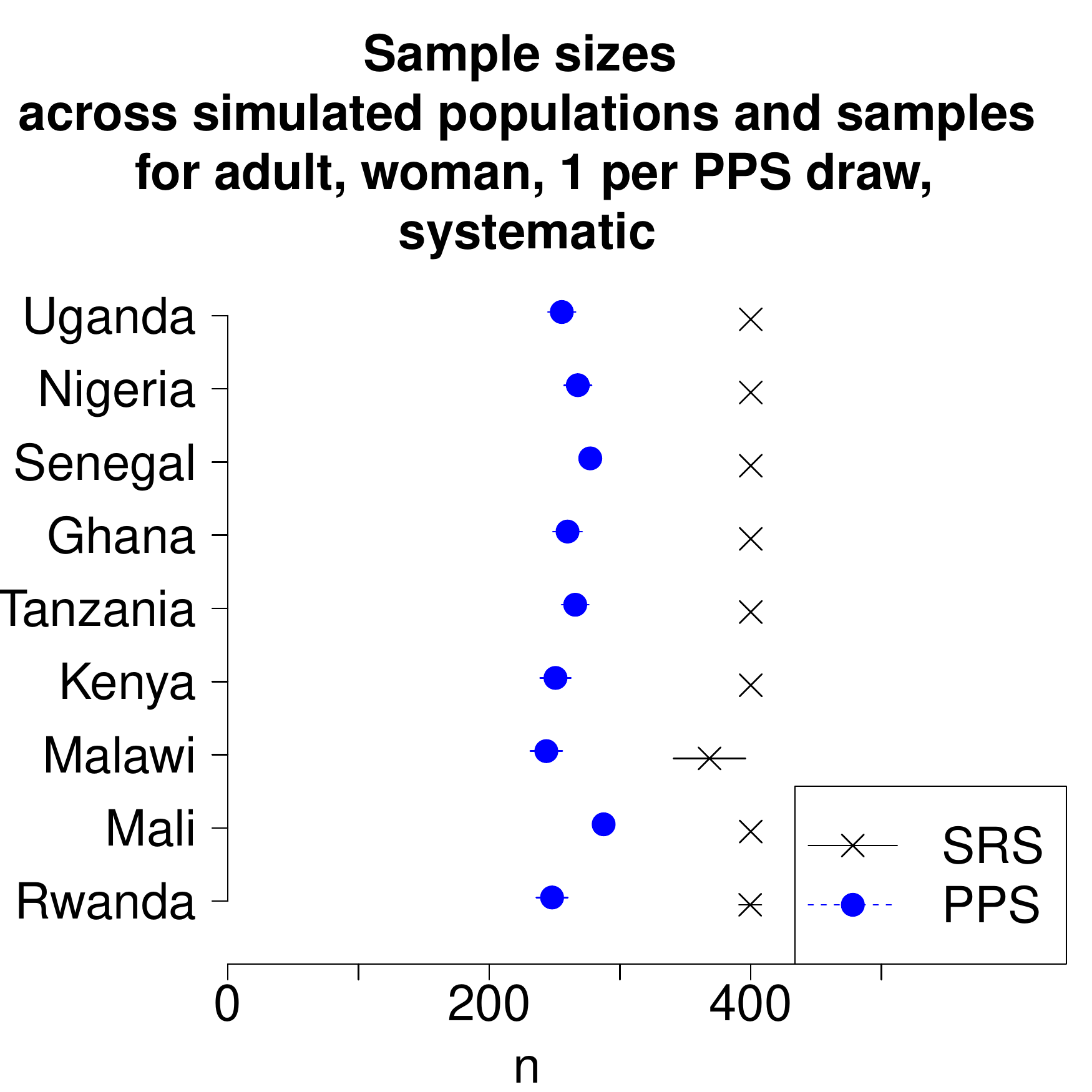}} \\
\subcaptionbox{\label{n_eff_adult_man_1_per_PPS_draw_design_based_systematic}}
 [0.49\textwidth]{\includegraphics[width=0.34\textwidth]{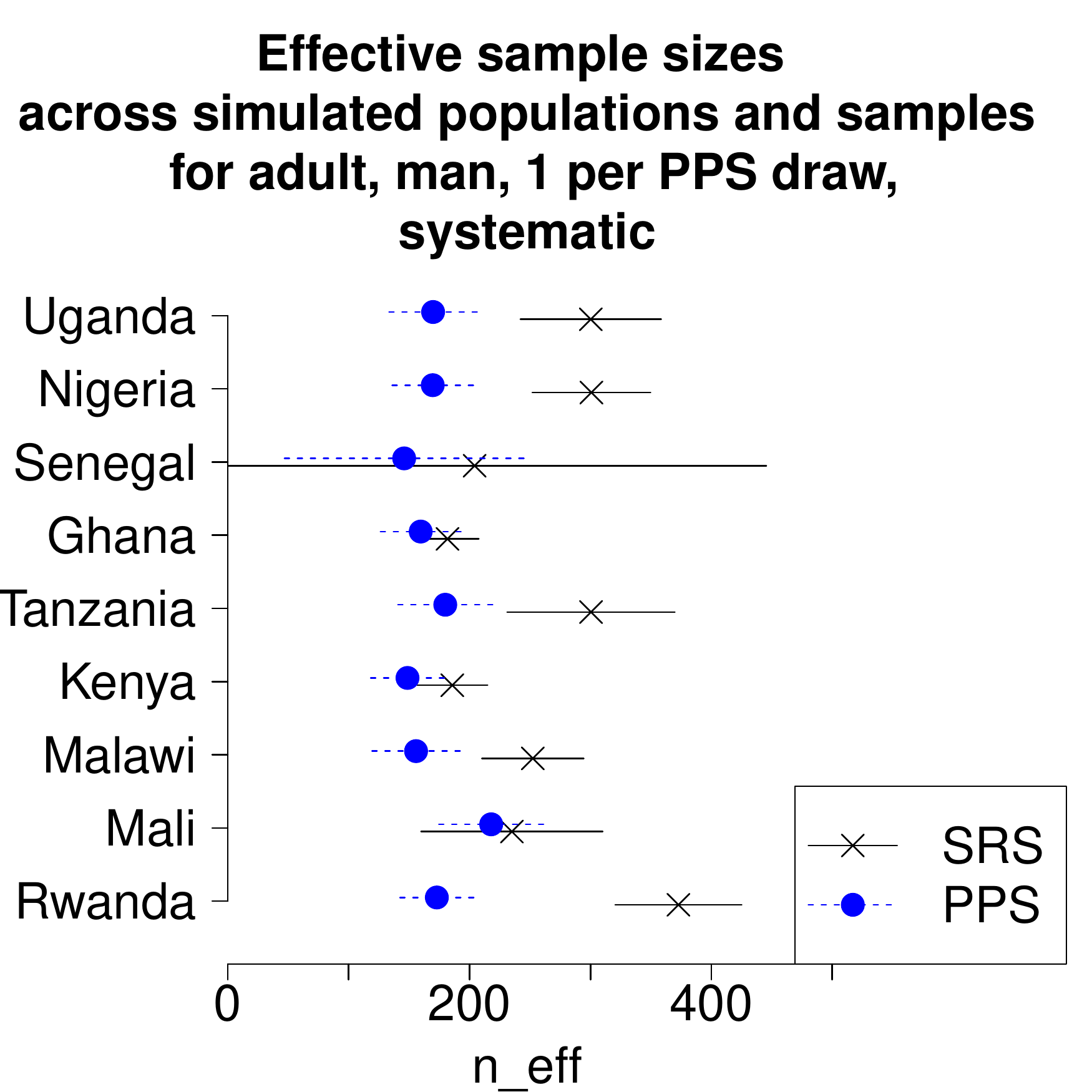}}
\subcaptionbox{\label{n_eff_adult_woman_1_per_PPS_draw_design_based_systematic}}
 [0.49\textwidth]{\includegraphics[width=0.34\textwidth]{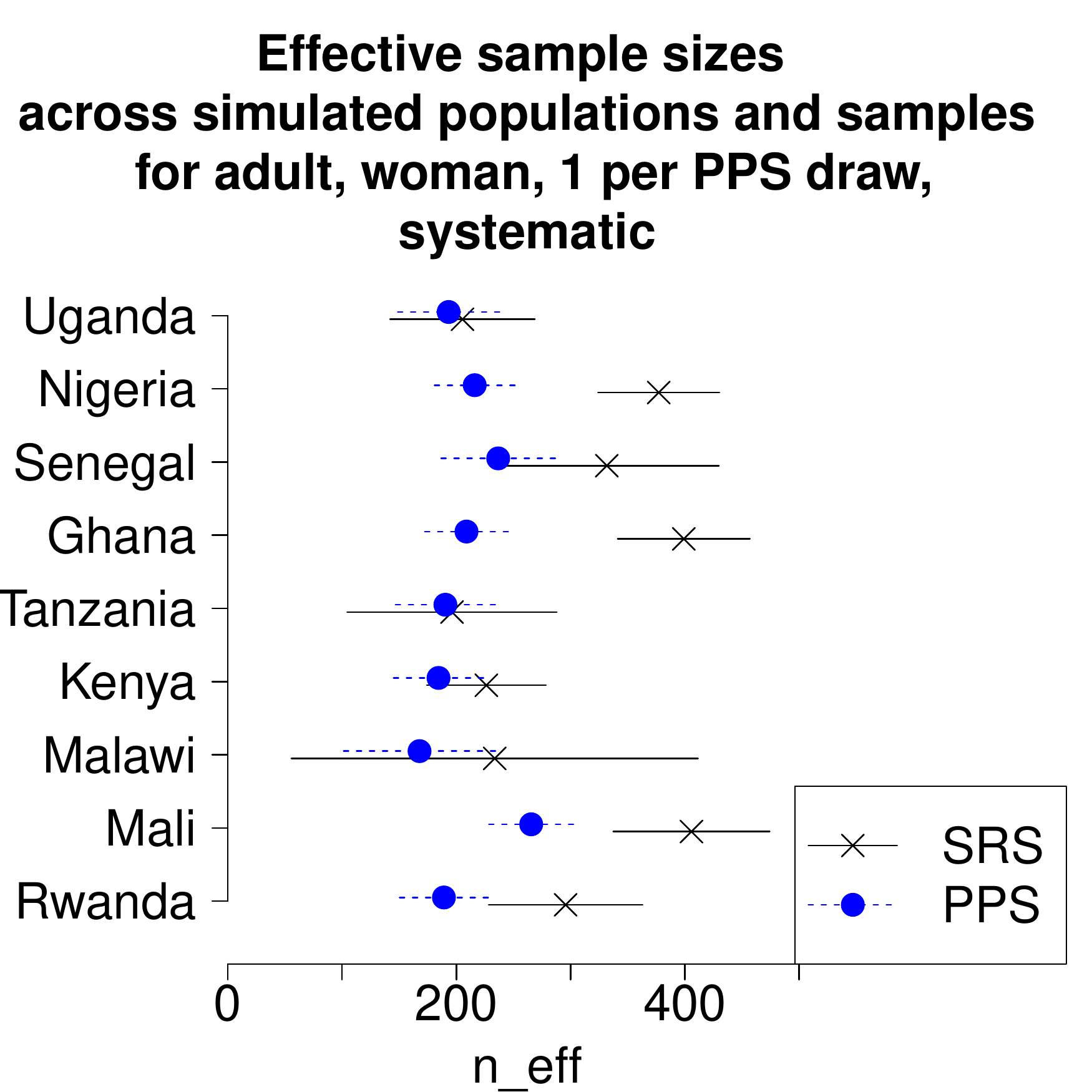}} 
\caption[]{Design-based adult module results}\label{adult_sampling_results_design_based_systematic}
\end{figure}

\begin{figure}[h!]
    \centering
\subcaptionbox{\label{deff_RDT_under_5_1_per_PPS_draw_design_based_systematic}}
 [0.49\textwidth]{\includegraphics[width=0.34\textwidth]{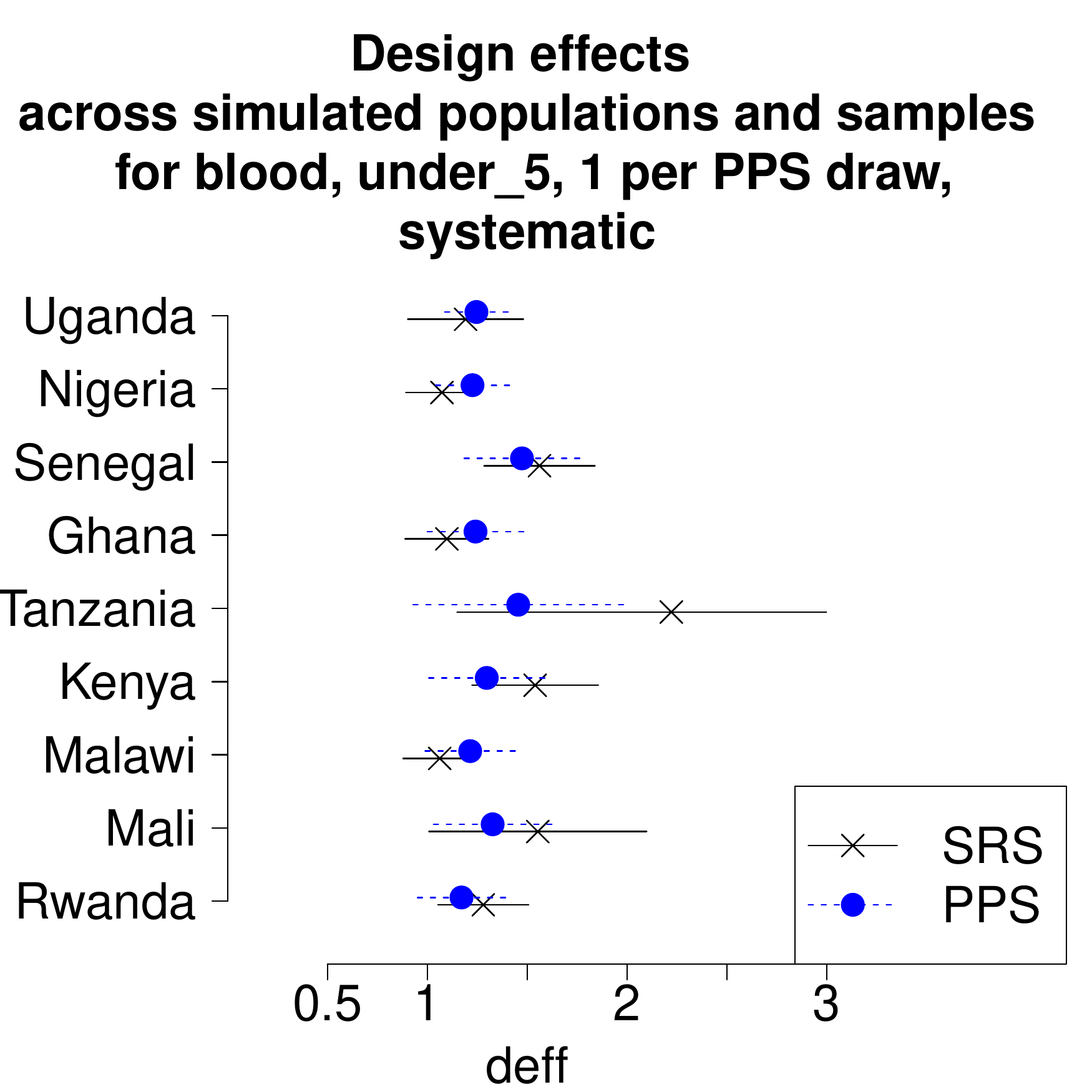}}
\subcaptionbox{\label{deff_RDT_school_age_1_per_PPS_draw_design_based_systematic}}
 [0.49\textwidth]{\includegraphics[width=0.34\textwidth]{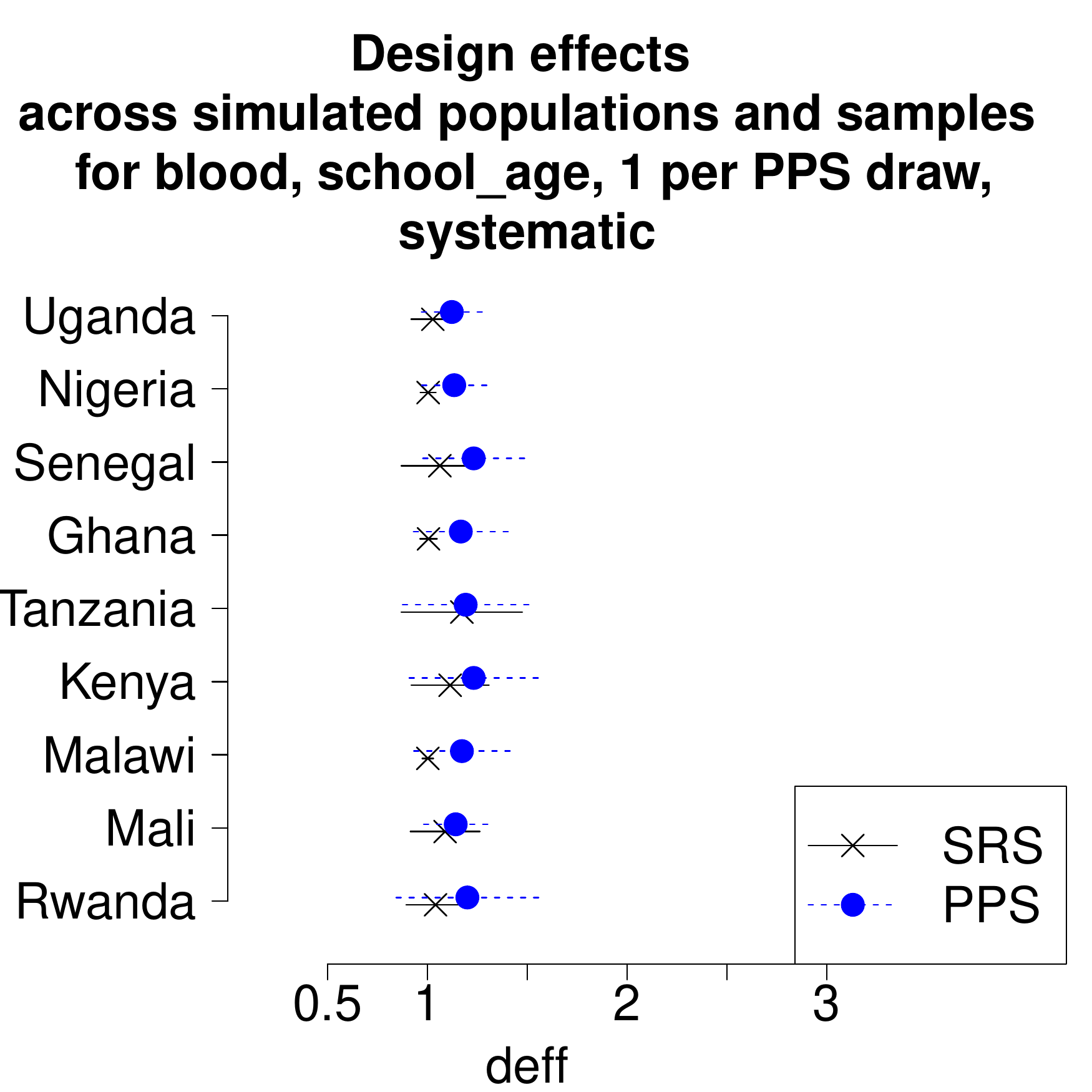}} \\
\subcaptionbox{\label{n_RDT_under_5_1_per_PPS_draw_systematic}}
 [0.49\textwidth]{\includegraphics[width=0.34\textwidth]{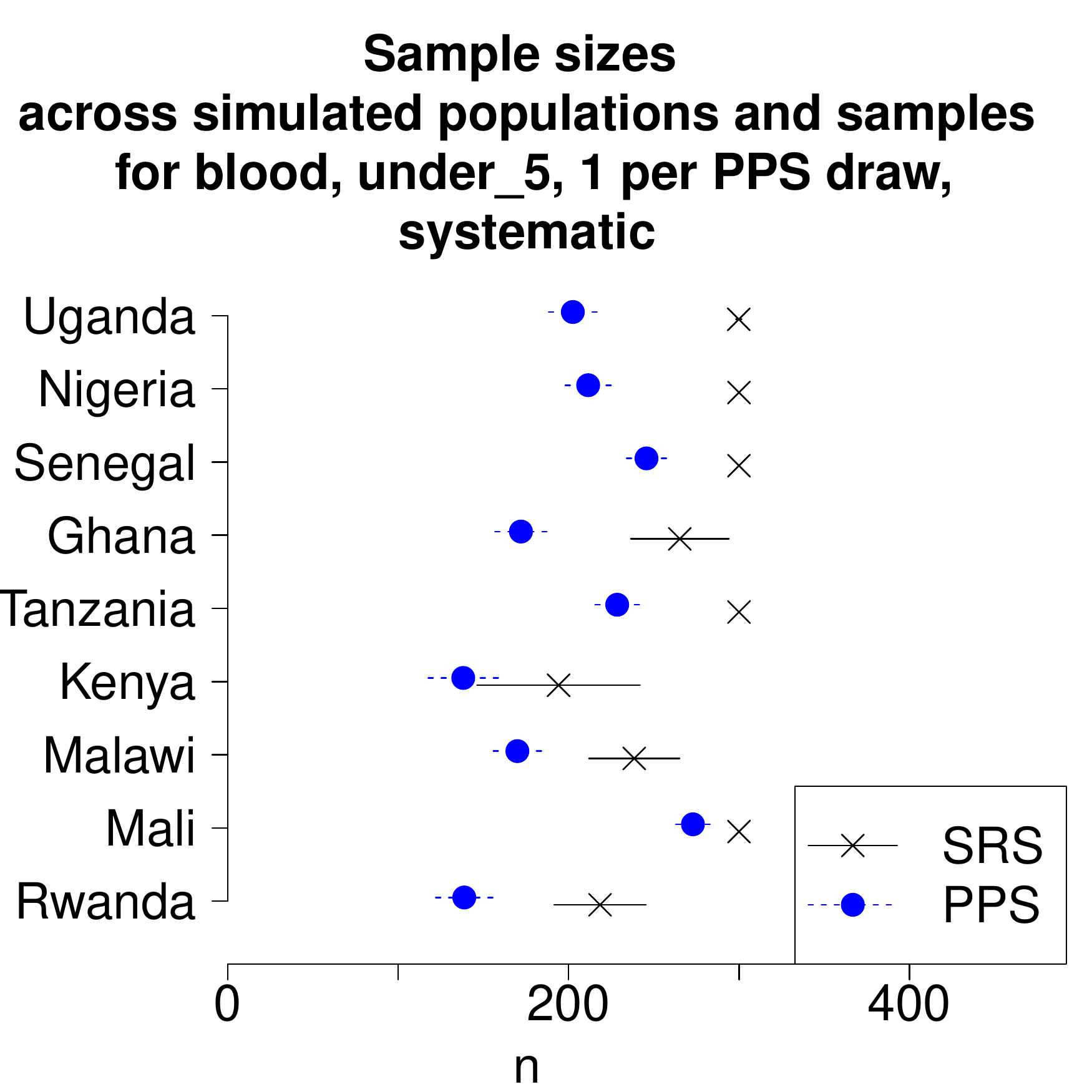}}
\subcaptionbox{\label{n_RDT_school_age_1_per_PPS_draw_systematic}}
 [0.49\textwidth]{\includegraphics[width=0.34\textwidth]{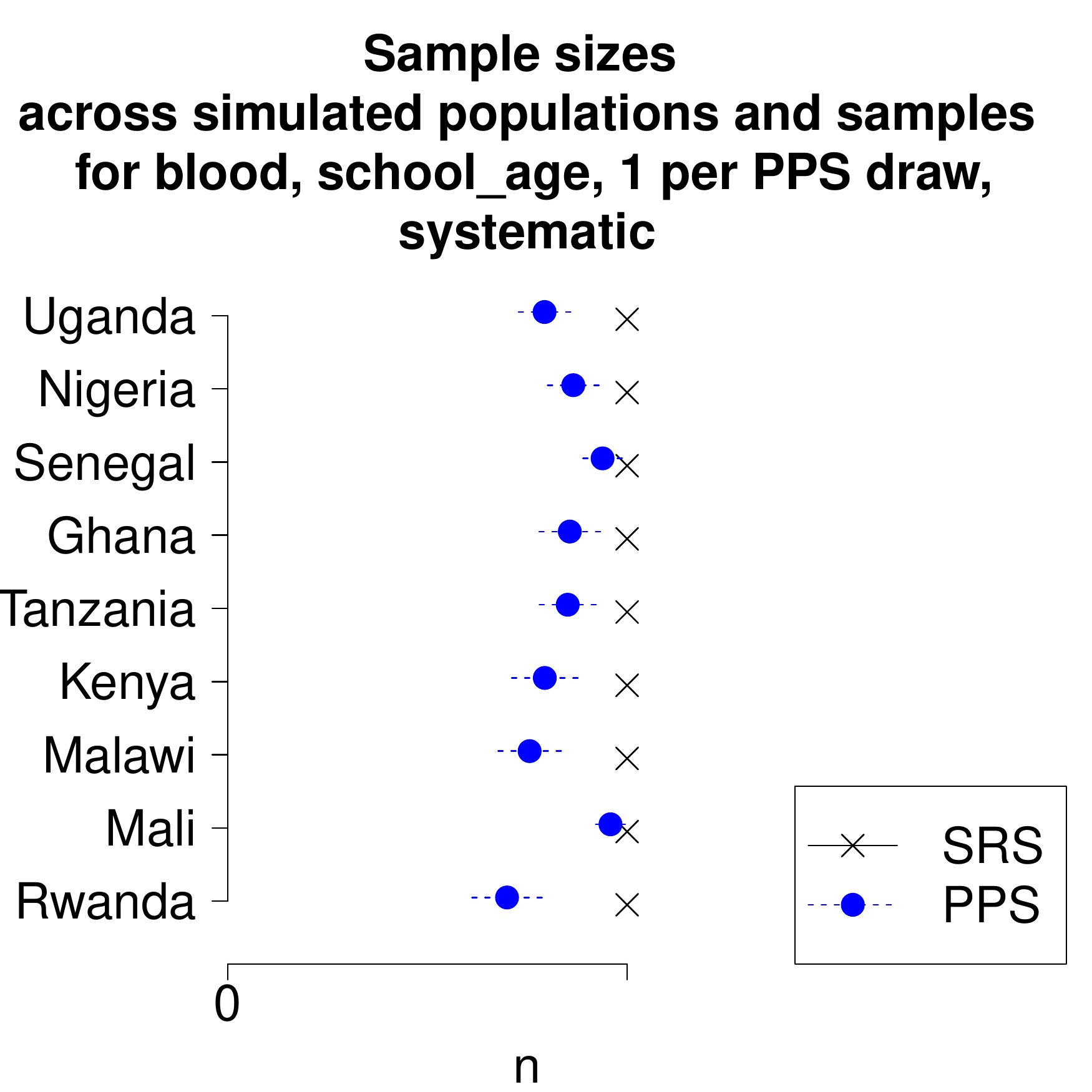}} \\
\subcaptionbox{\label{n_eff_RDT_under_5_1_per_PPS_draw_design_based_systematic}}
 [0.49\textwidth]{\includegraphics[width=0.34\textwidth]{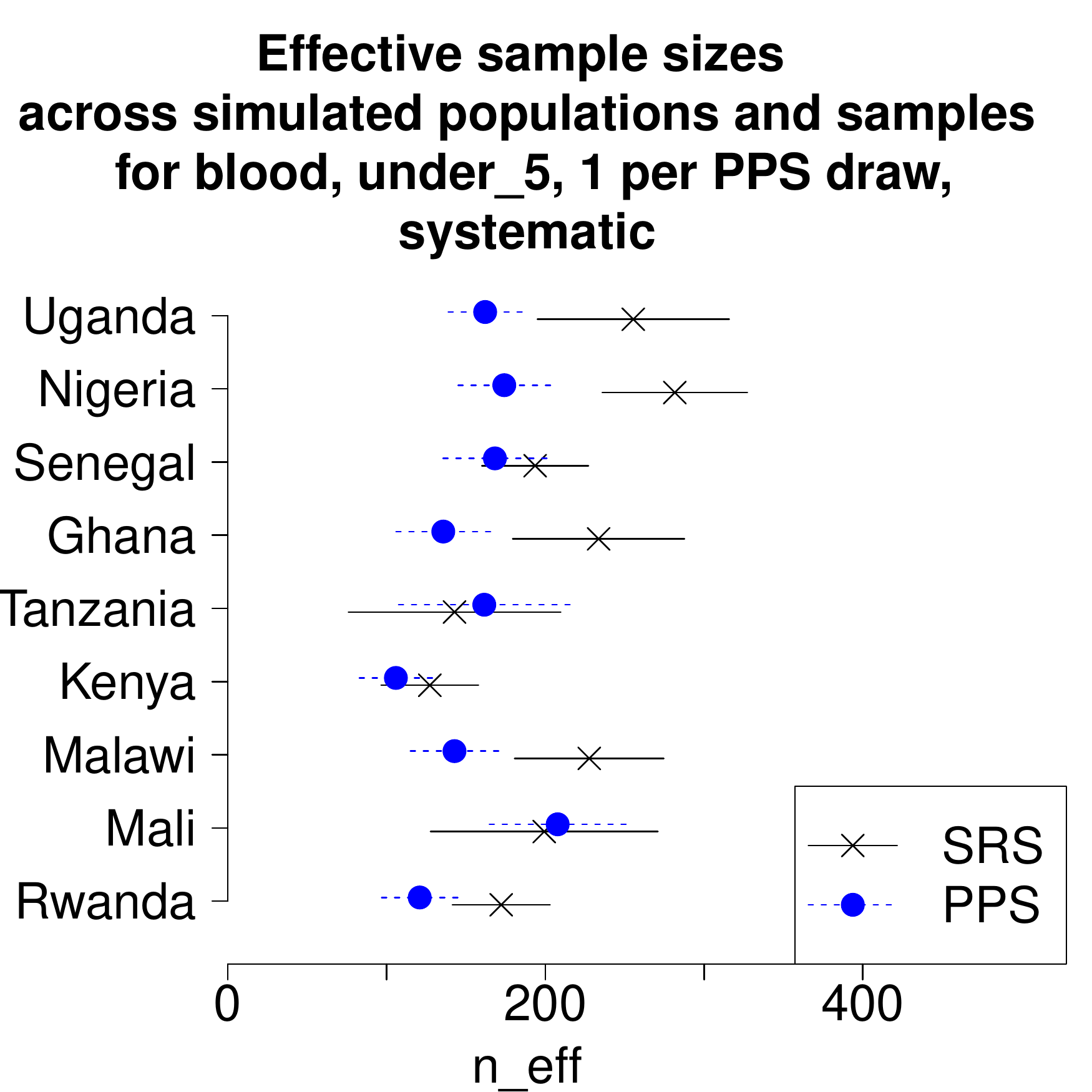}}
\subcaptionbox{\label{n_eff_RDT_school_age_1_per_PPS_draw_design_based_systematic}}
 [0.49\textwidth]{\includegraphics[width=0.34\textwidth]{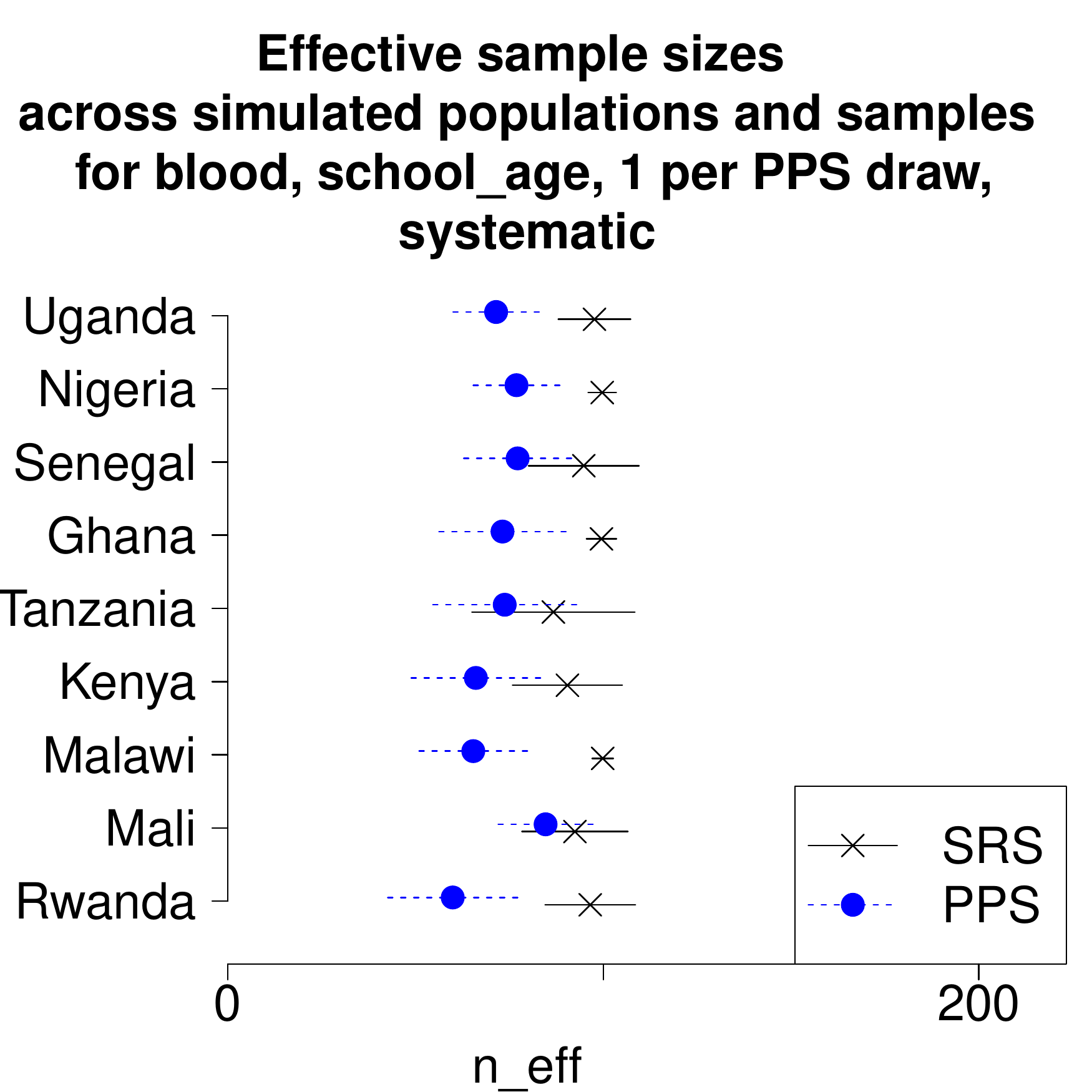}} 
\caption[]{Design-based blood (malaria and anemia) module results: under 5 and school-age children}\label{RDT_children_sampling_results_design_based_systematic}
\end{figure}

\begin{figure}[h!]
    \centering
\subcaptionbox{\label{deff_RDT_man_1_per_PPS_draw_design_based_systematic}}
 [0.49\textwidth]{\includegraphics[width=0.34\textwidth]{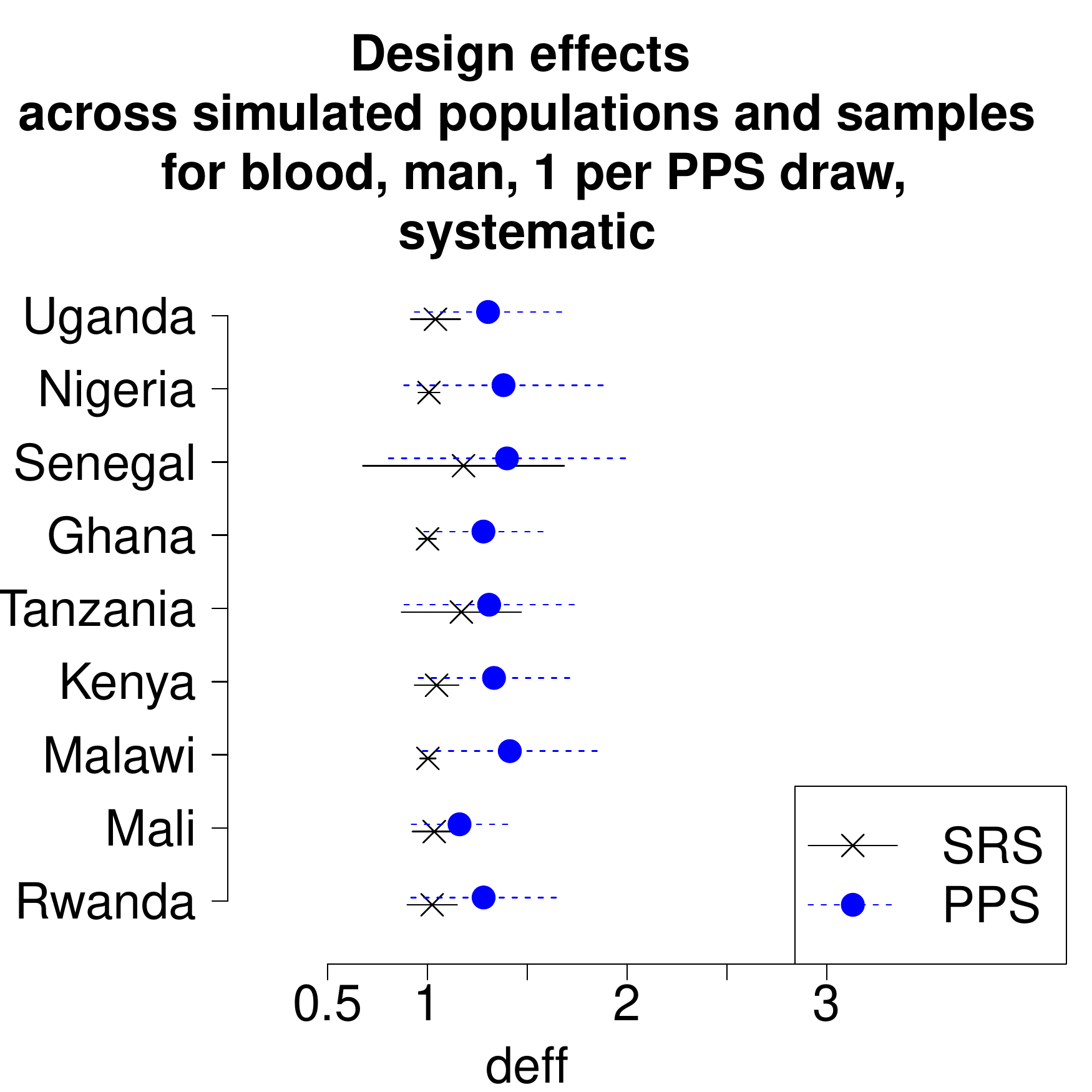}}
\subcaptionbox{\label{deff_RDT_woman_1_per_PPS_draw_design_based_systematic}}
 [0.49\textwidth]{\includegraphics[width=0.34\textwidth]{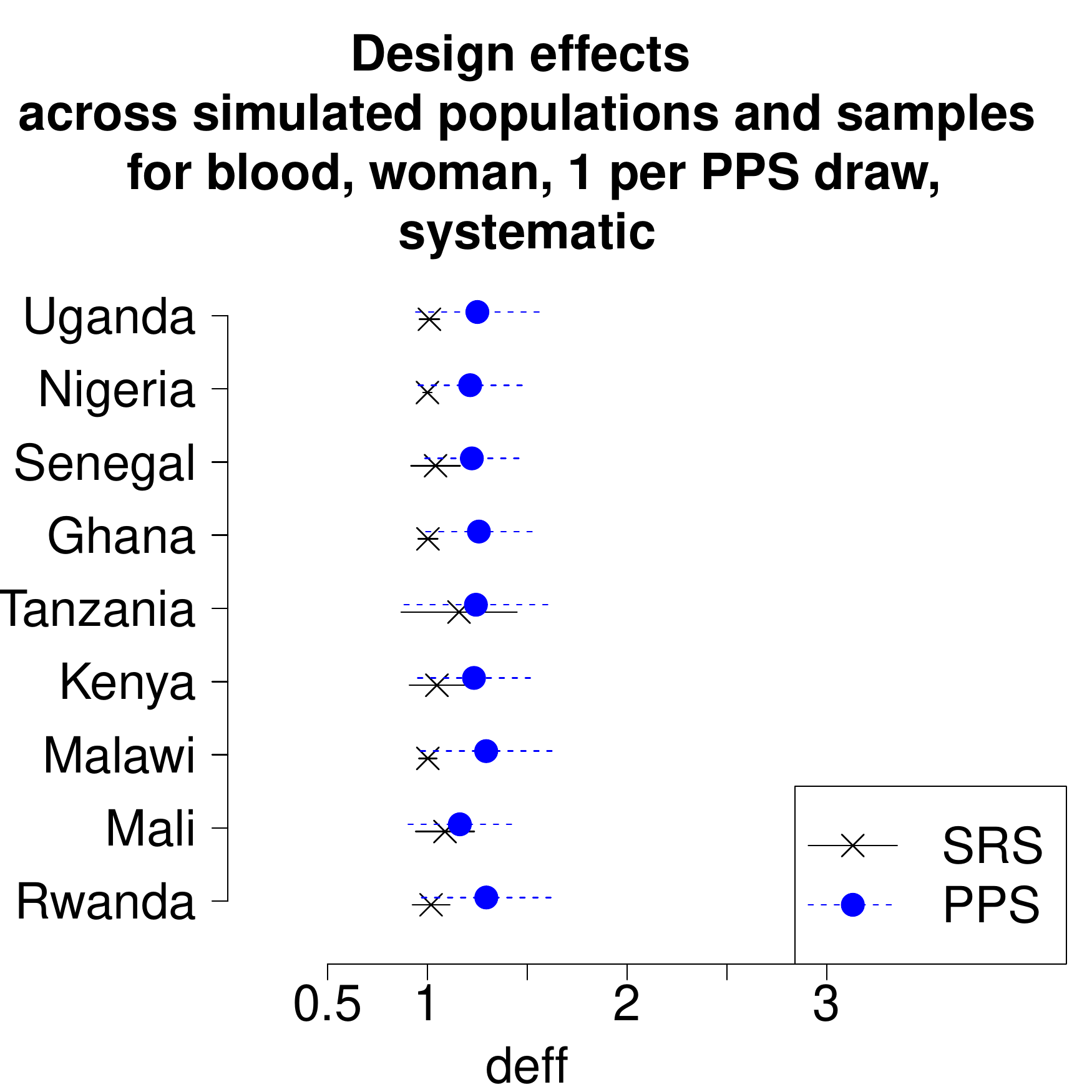}} \\
\subcaptionbox{\label{n_RDT_man_1_per_PPS_draw_systematic}}
 [0.49\textwidth]{\includegraphics[width=0.34\textwidth]{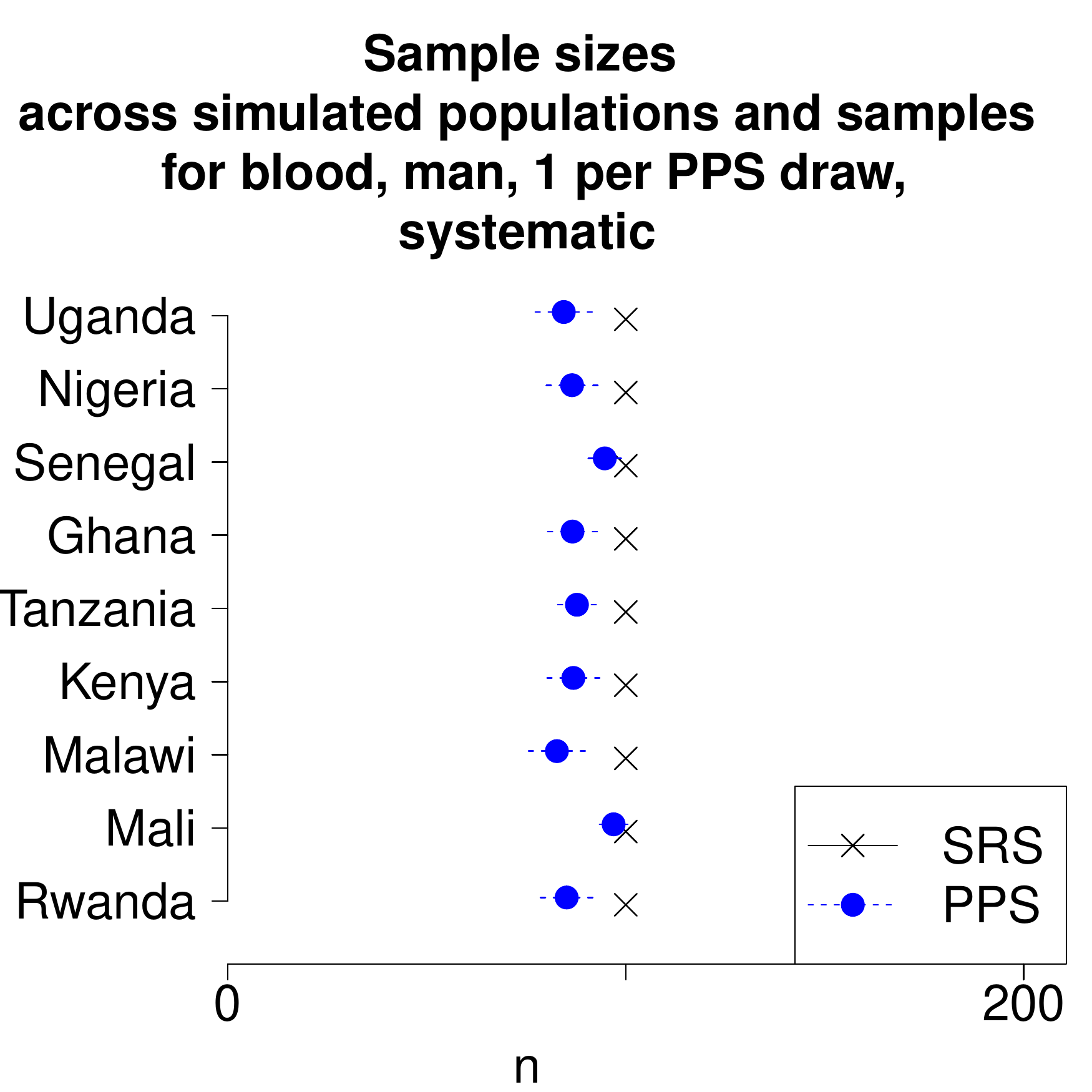}}
\subcaptionbox{\label{n_RDT_woman_1_per_PPS_draw_systematic}}
 [0.49\textwidth]{\includegraphics[width=0.34\textwidth]{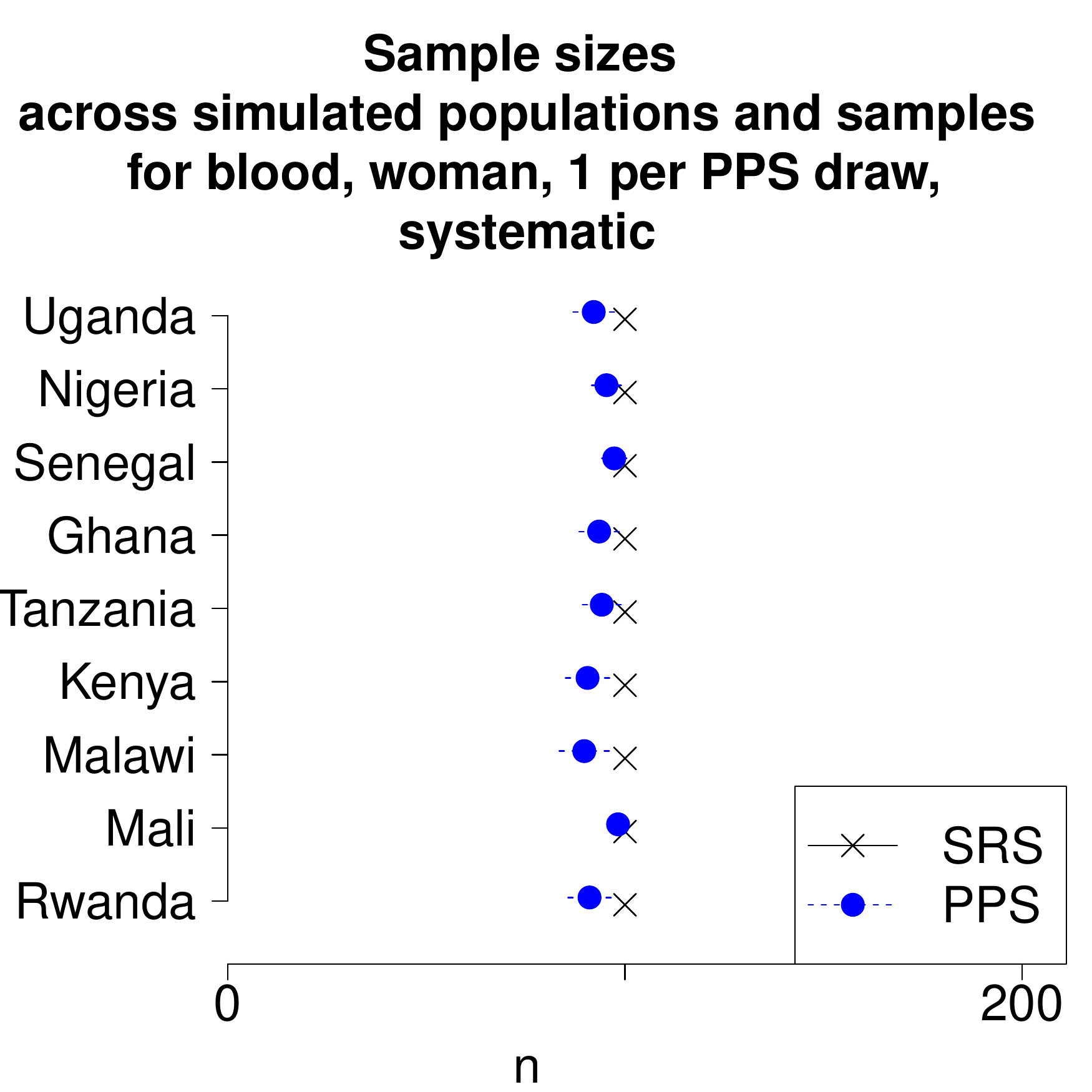}} \\
\subcaptionbox{\label{n_eff_RDT_man_1_per_PPS_draw_design_based_systematic}}
 [0.49\textwidth]{\includegraphics[width=0.34\textwidth]{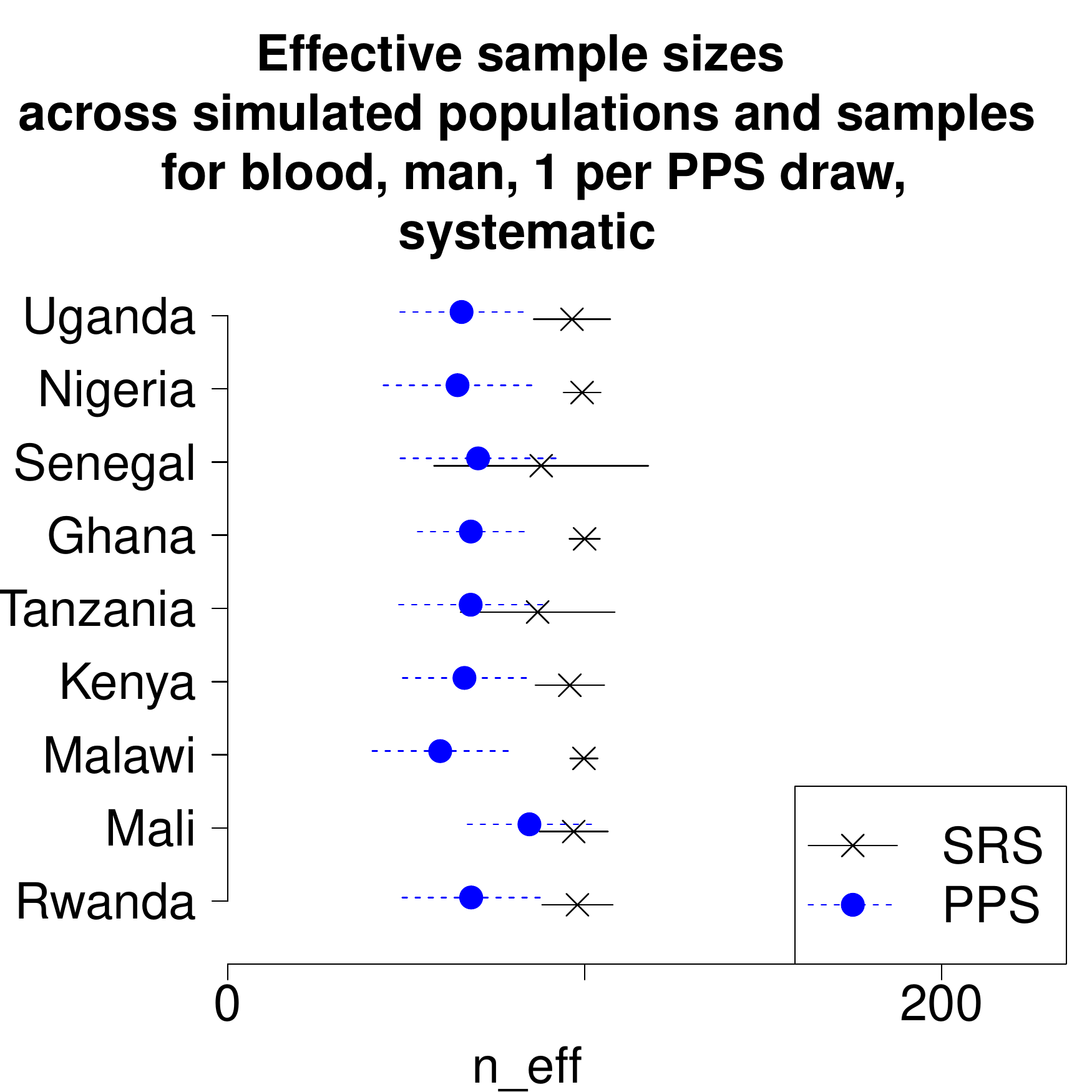}}
\subcaptionbox{\label{n_eff_RDT_woman_1_per_PPS_draw_design_based_systematic}}
 [0.49\textwidth]{\includegraphics[width=0.34\textwidth]{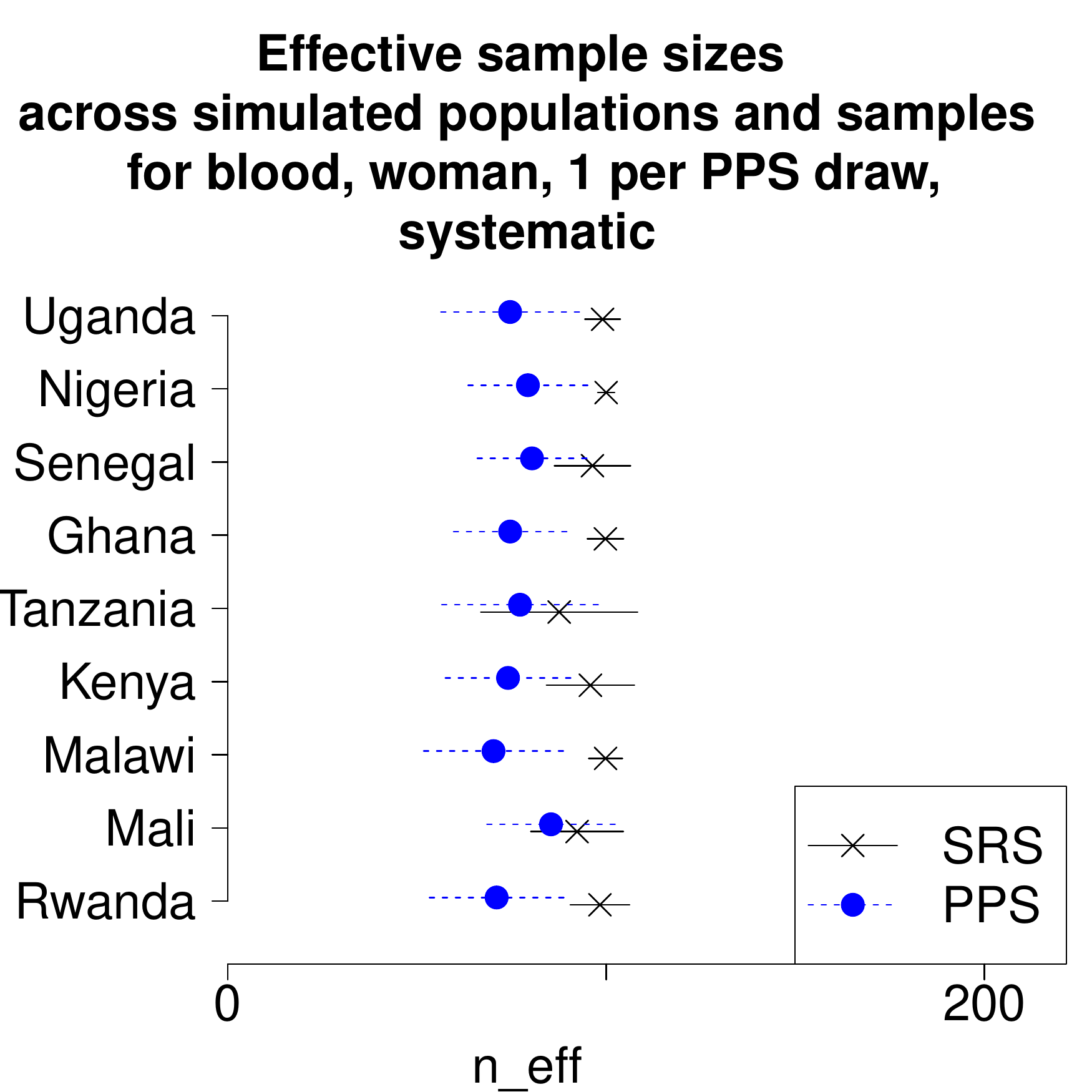}} 
\caption[]{Design-based blood (malaria and anemia) module results: men and women}\label{RDT_man_woman_sampling_results_design_based_systematic}
\end{figure}

\begin{figure}[h!]
    \centering
\subcaptionbox{\label{deff_anthro_under_5_1_per_PPS_draw_design_based_systematic}}
 [0.49\textwidth]{\includegraphics[width=0.34\textwidth]{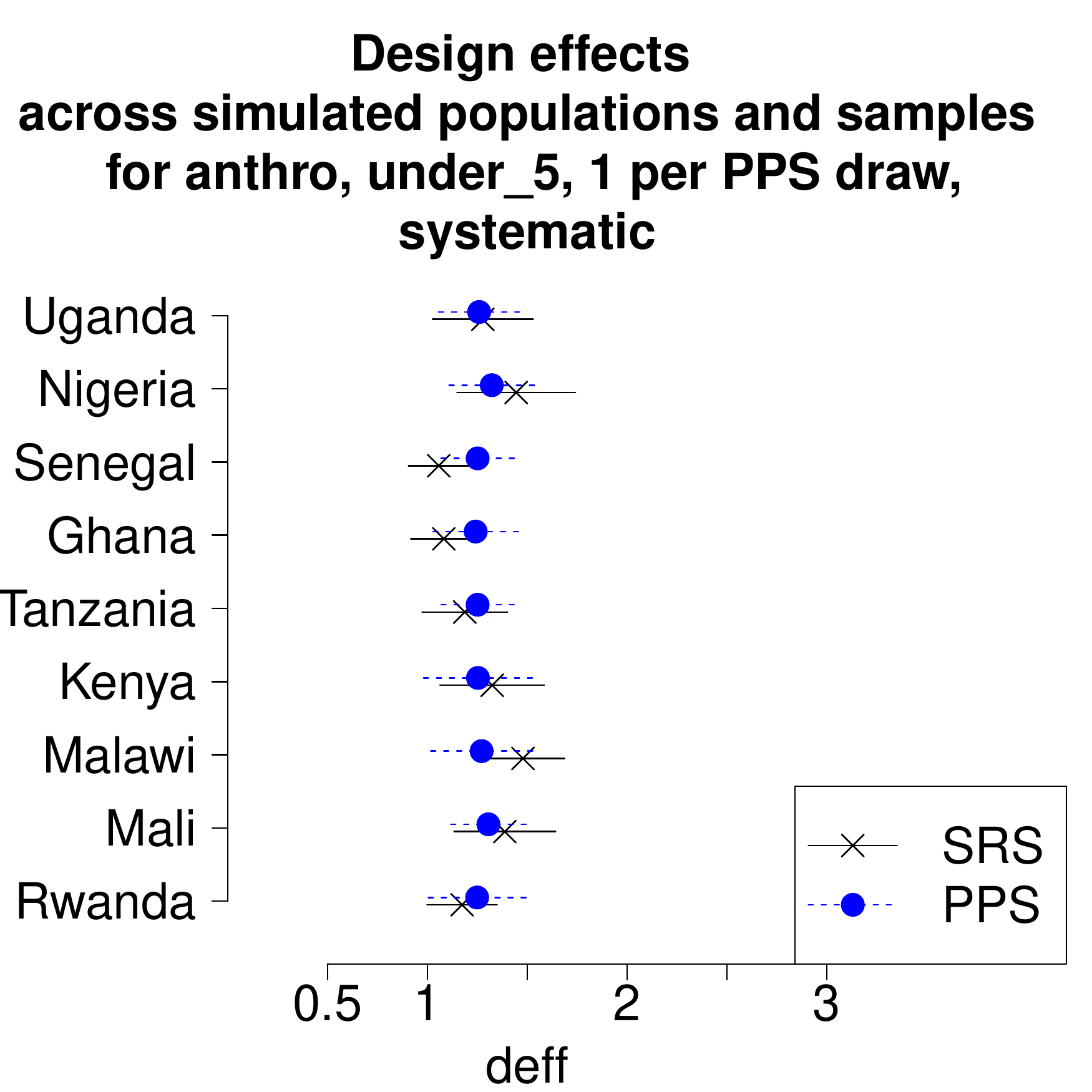}} \\
\subcaptionbox{\label{n_anthro_under_5_1_per_PPS_draw_systematic}}
 [0.49\textwidth]{\includegraphics[width=0.34\textwidth]{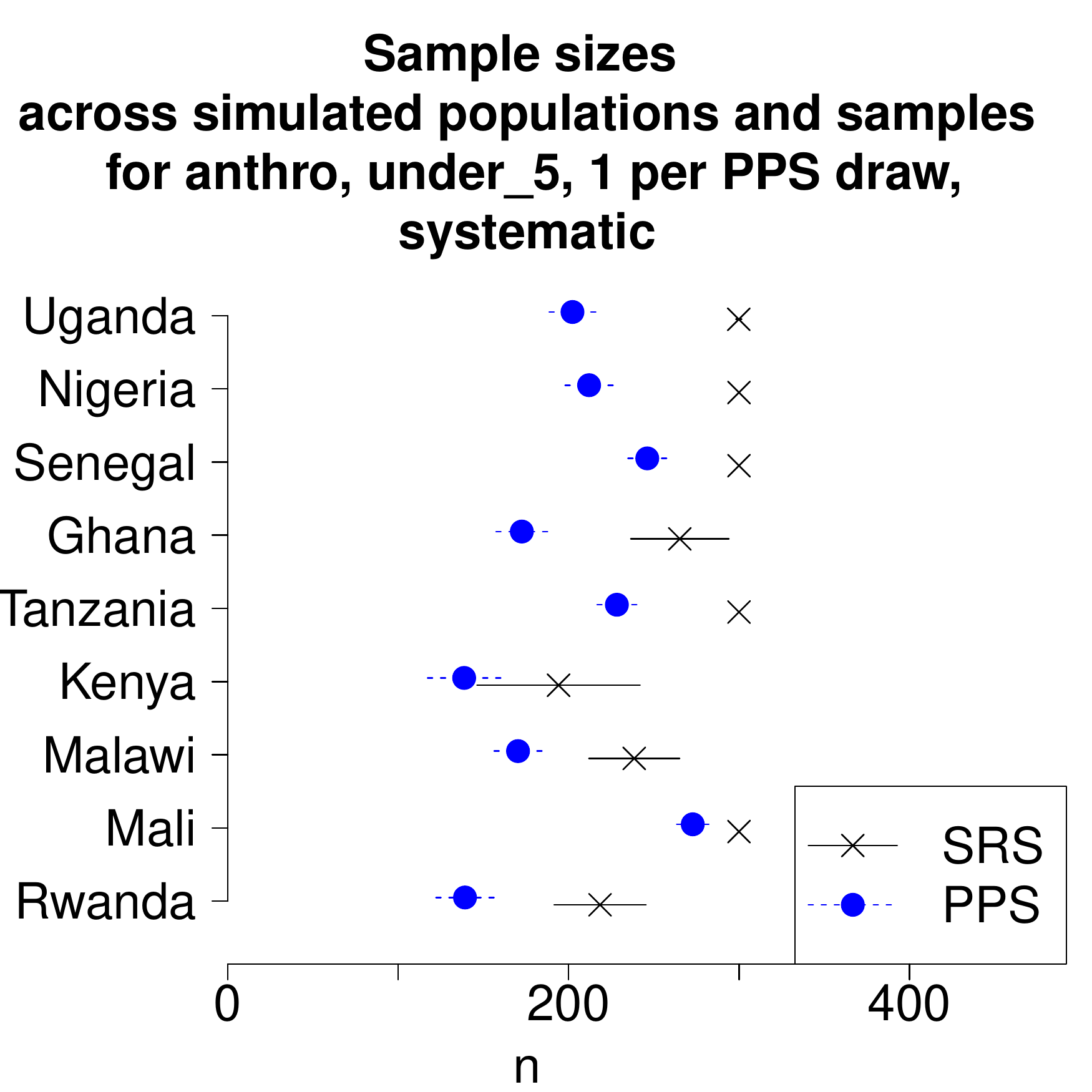}} \\
\subcaptionbox{\label{n_eff_anthro_under_5_1_per_PPS_draw_design_based_systematic}}
 [0.49\textwidth]{\includegraphics[width=0.34\textwidth]{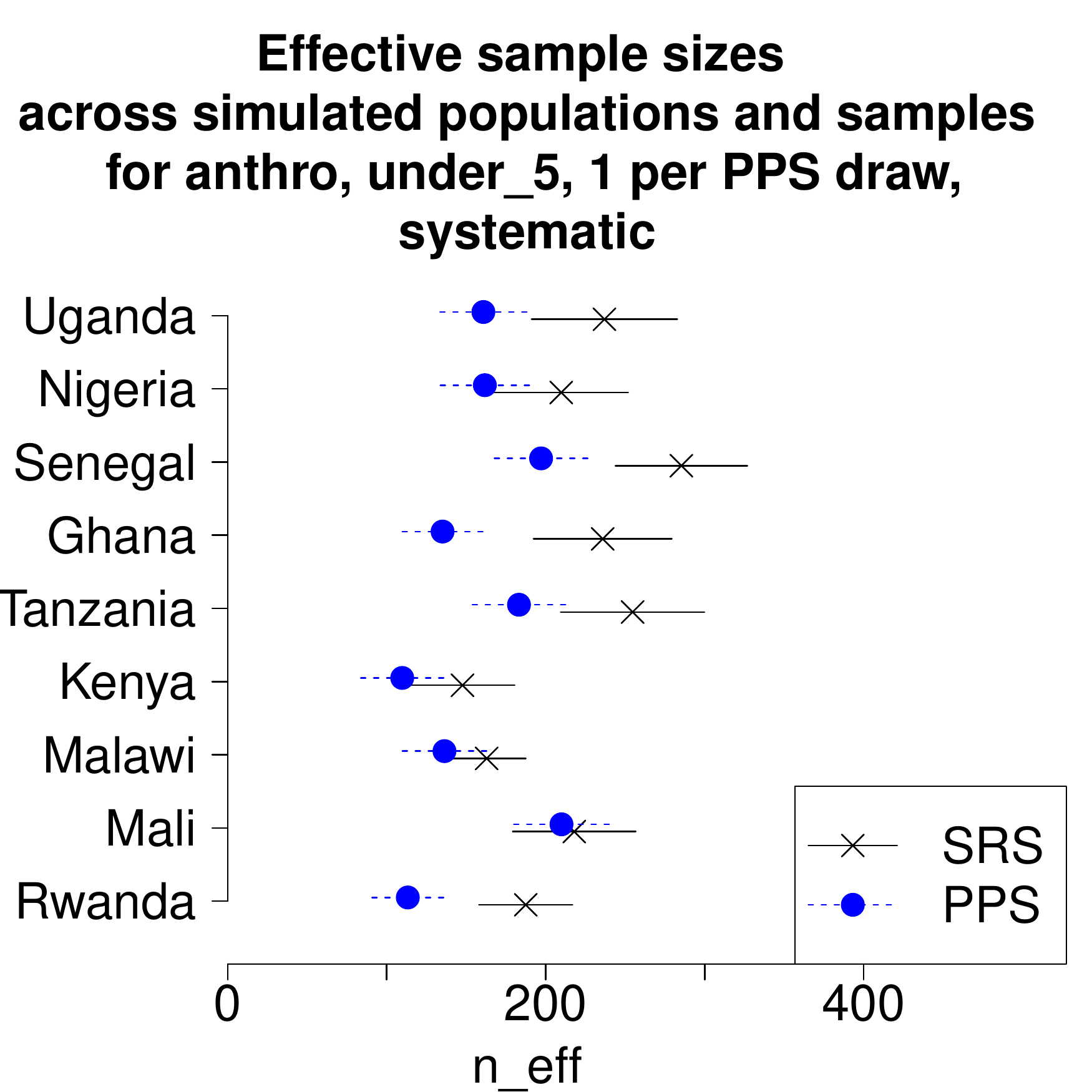}}
\caption[]{Design-based anthro module results}\label{anthro_sampling_results_design_based_systematic}
\end{figure}

\begin{figure}[h!]
    \centering
\subcaptionbox{\label{deff_adult_man_1_per_PPS_draw_design_based_stratified}}
 [0.49\textwidth]{\includegraphics[width=0.34\textwidth]{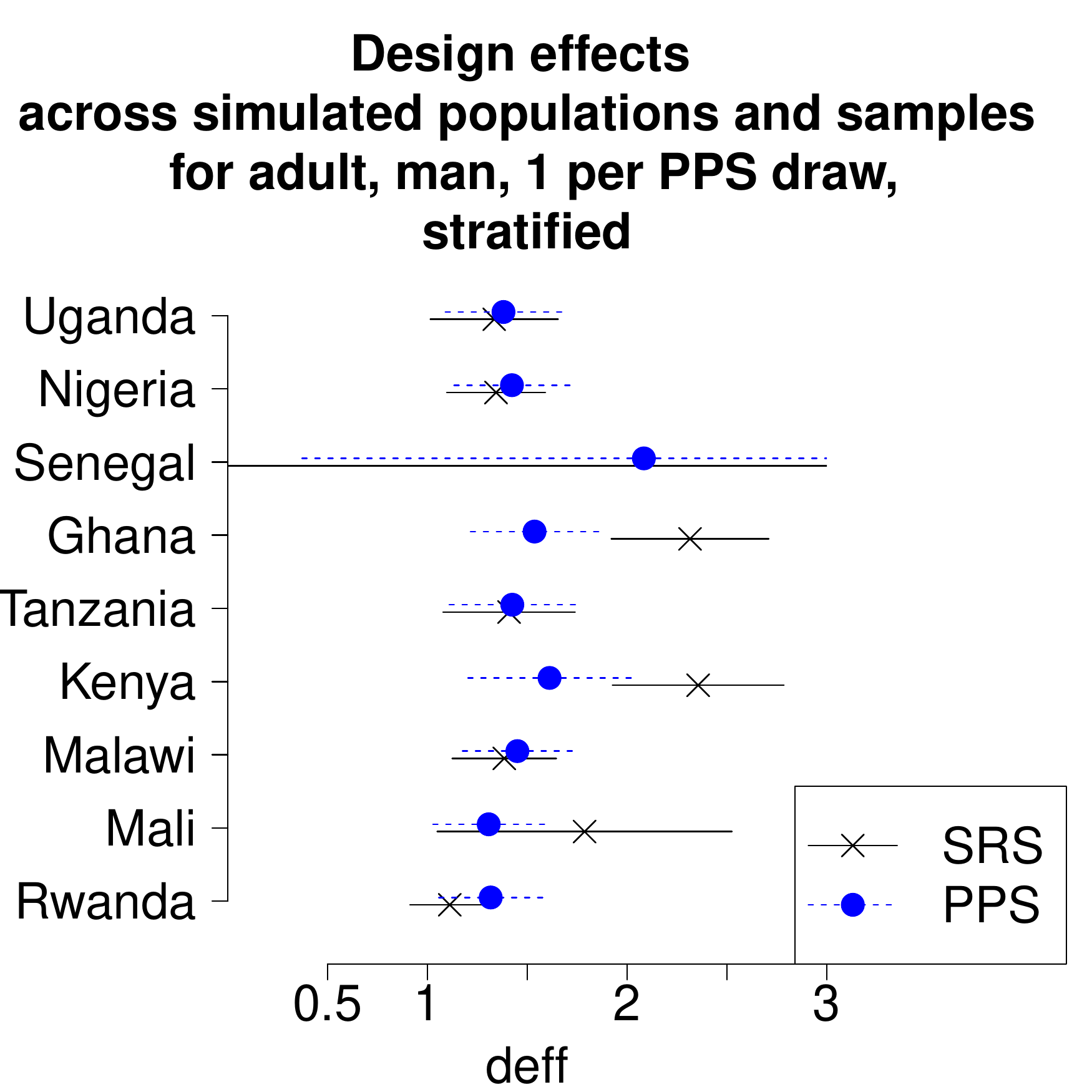}}
\subcaptionbox{\label{deff_adult_woman_1_per_PPS_draw_design_based_stratified}}
 [0.49\textwidth]{\includegraphics[width=0.34\textwidth]{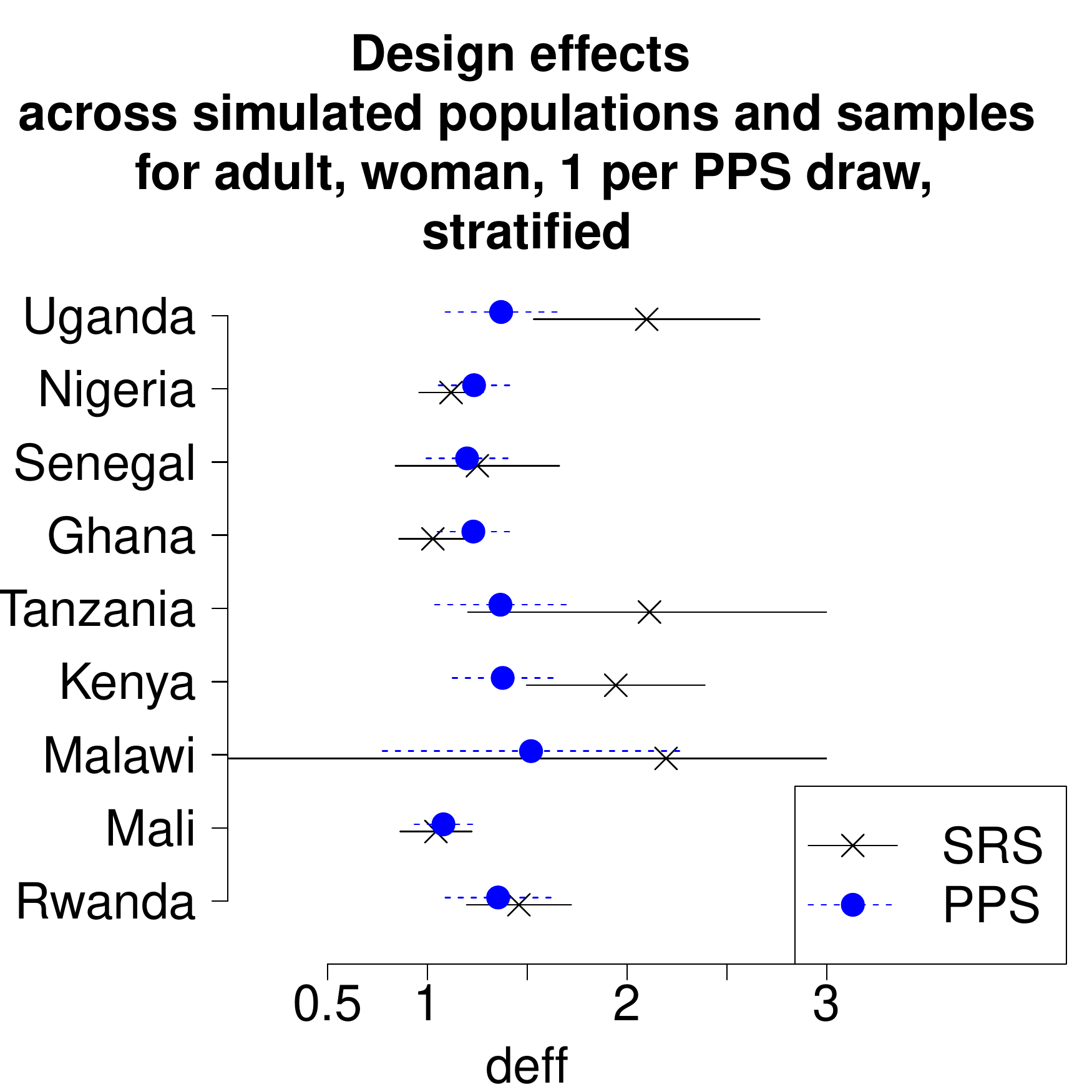}} \\
\subcaptionbox{\label{n_adult_man_1_per_PPS_draw_stratified}}
 [0.49\textwidth]{\includegraphics[width=0.34\textwidth]{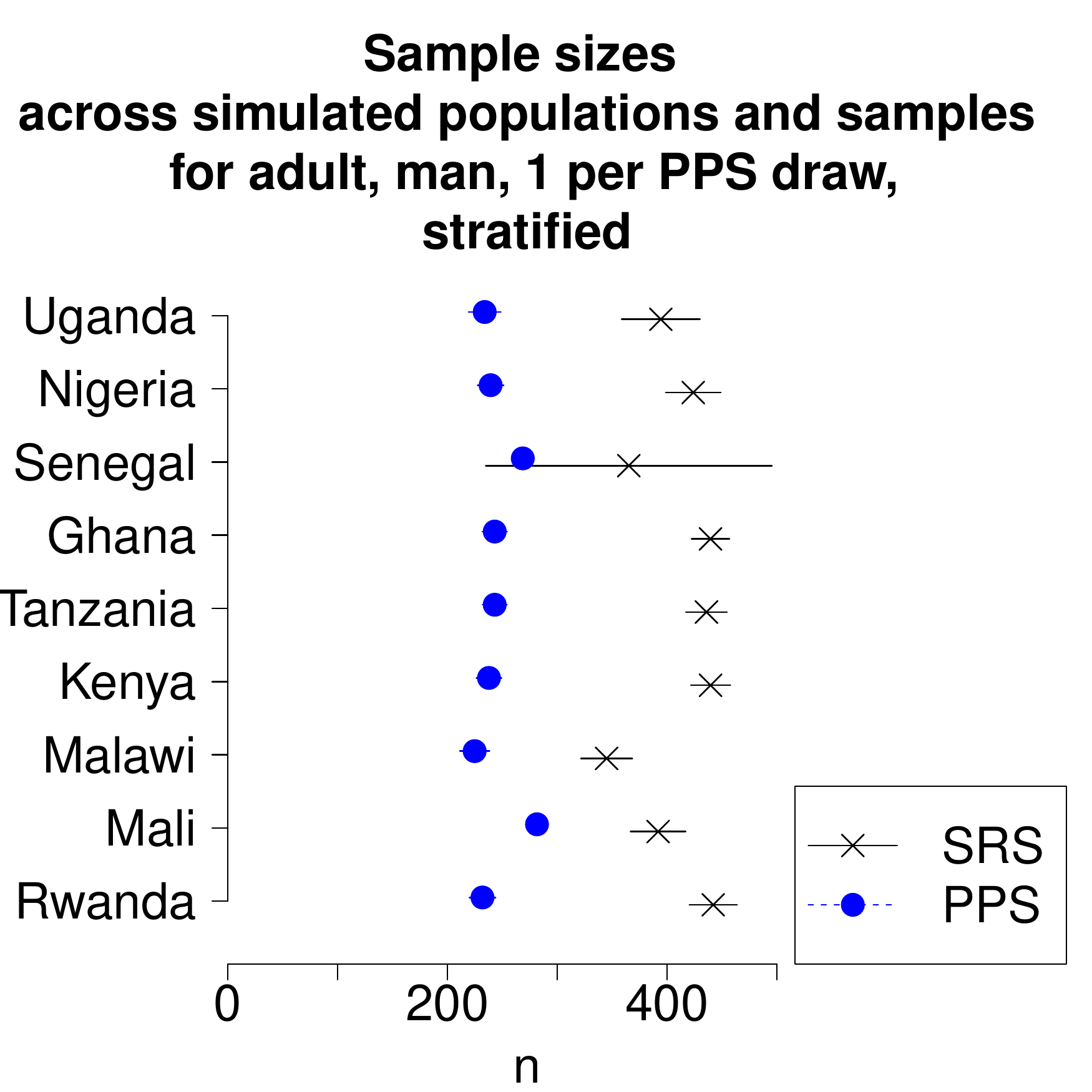}}
\subcaptionbox{\label{n_adult_woman_1_per_PPS_draw_stratified}}
 [0.49\textwidth]{\includegraphics[width=0.34\textwidth]{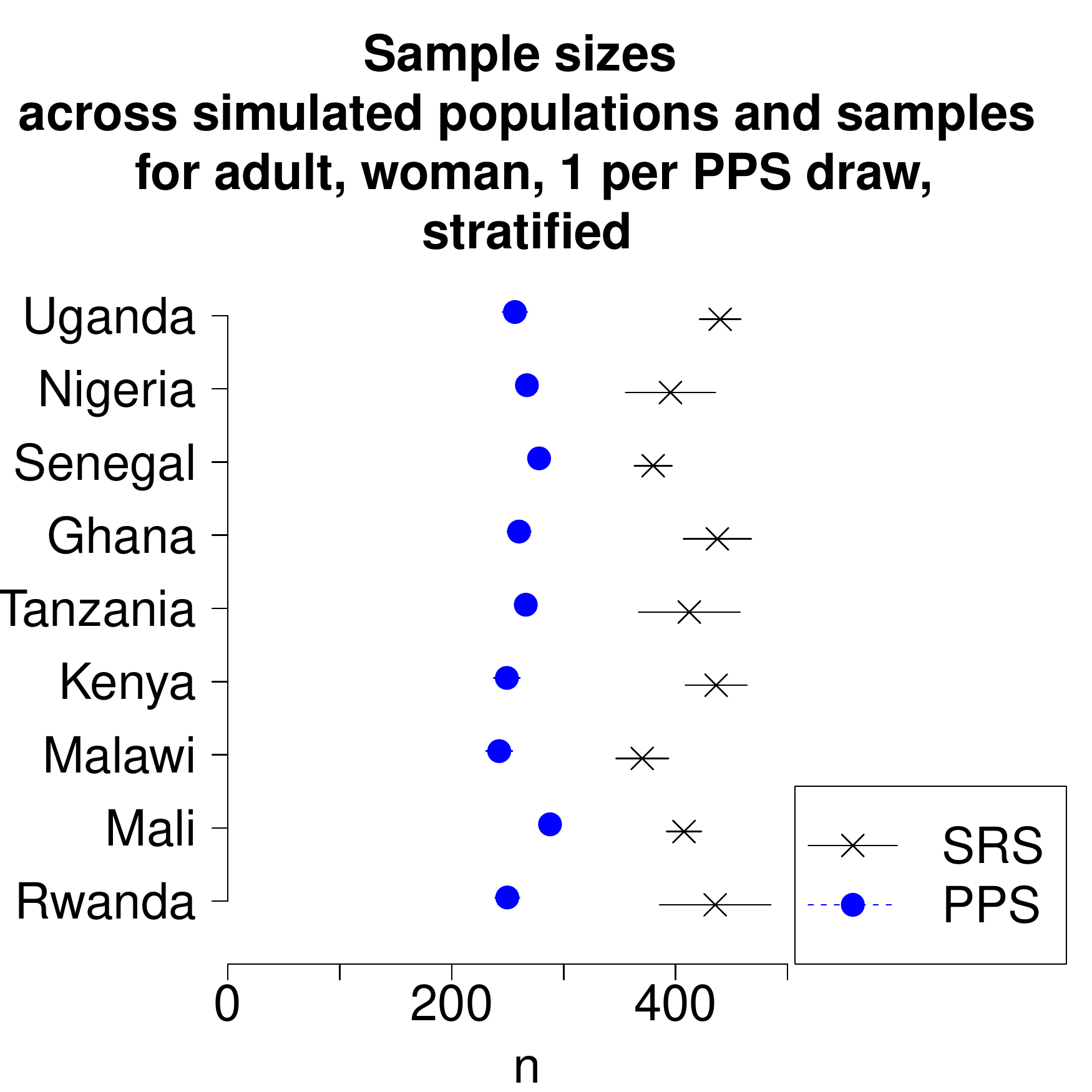}} \\
\subcaptionbox{\label{n_eff_adult_man_1_per_PPS_draw_design_based_stratified}}
 [0.49\textwidth]{\includegraphics[width=0.34\textwidth]{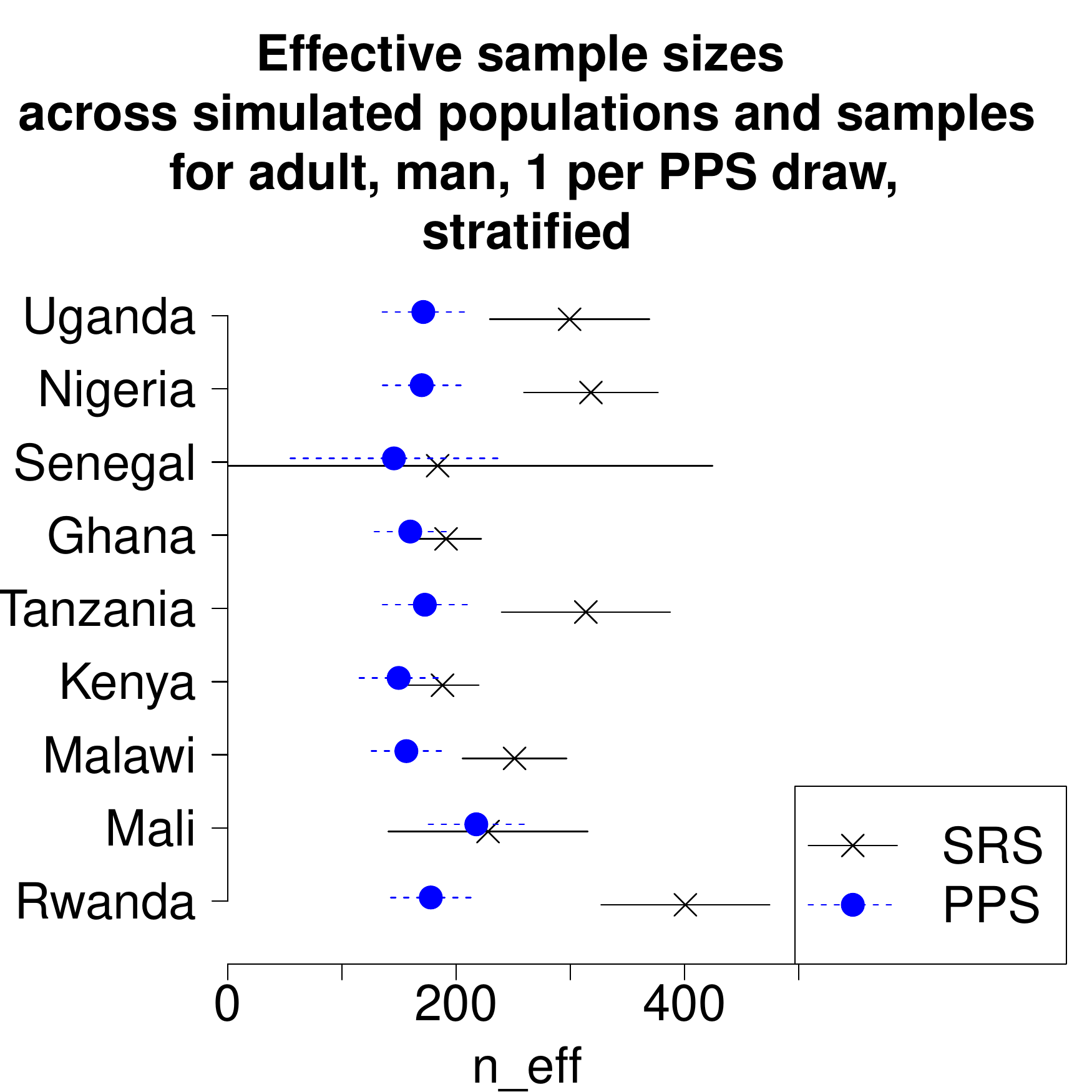}}
\subcaptionbox{\label{n_eff_adult_woman_1_per_PPS_draw_design_based_stratified}}
 [0.49\textwidth]{\includegraphics[width=0.34\textwidth]{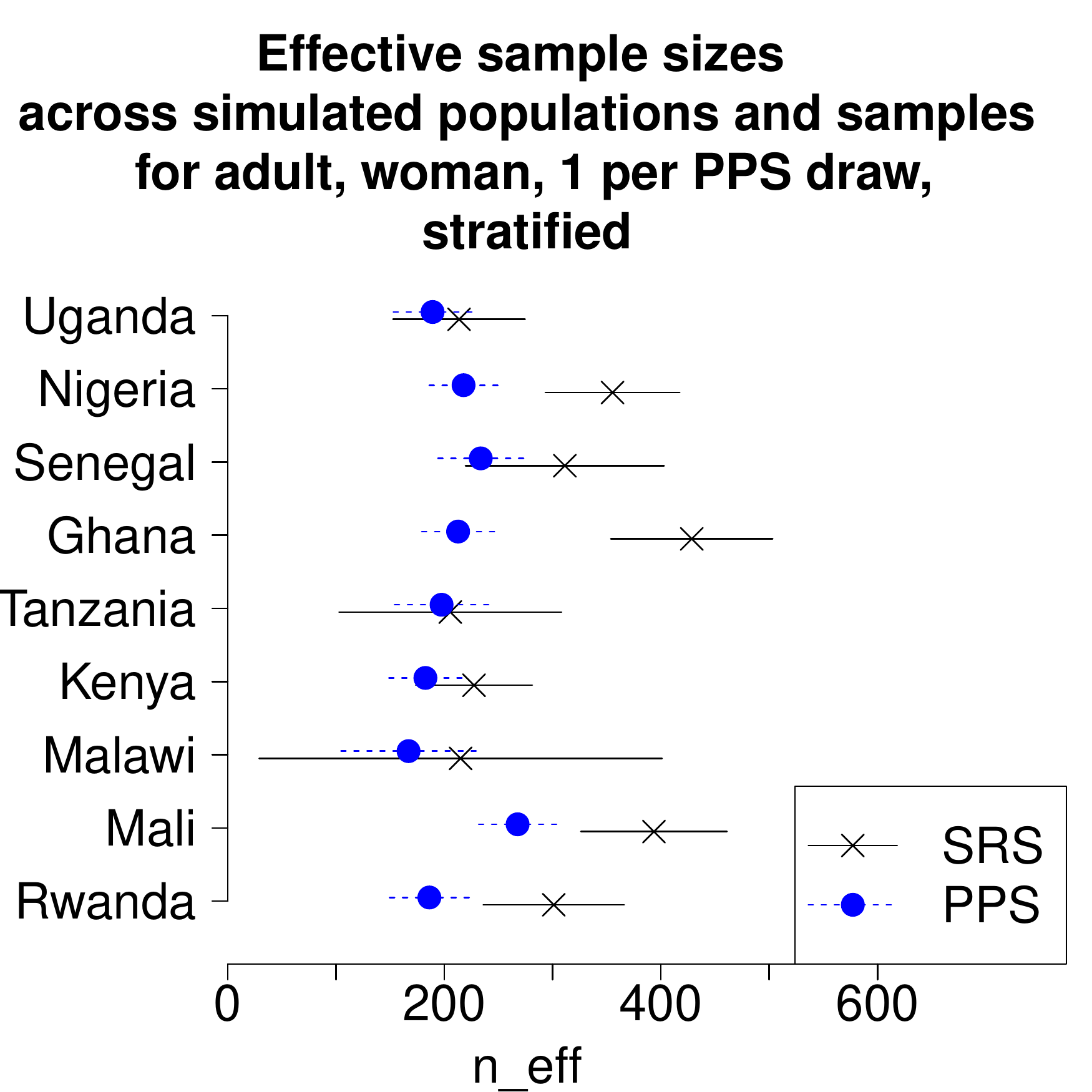}} 
\caption[]{Design-based adult module results}\label{adult_sampling_results_design_based_stratified}
\end{figure}

\begin{figure}[h!]
    \centering
\subcaptionbox{\label{deff_RDT_under_5_1_per_PPS_draw_design_based_stratified}}
 [0.49\textwidth]{\includegraphics[width=0.34\textwidth]{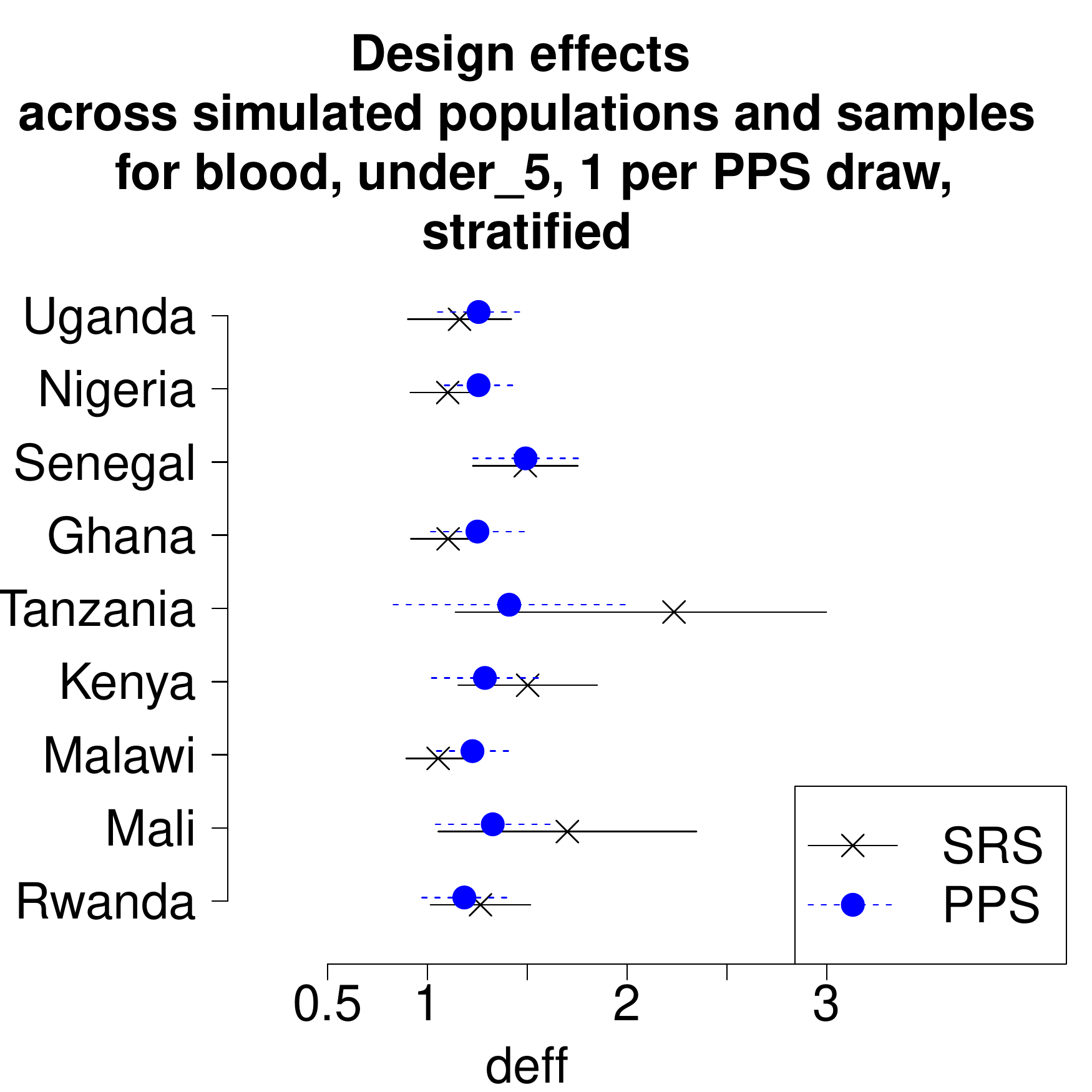}}
\subcaptionbox{\label{deff_RDT_school_age_1_per_PPS_draw_design_based_stratified}}
 [0.49\textwidth]{\includegraphics[width=0.34\textwidth]{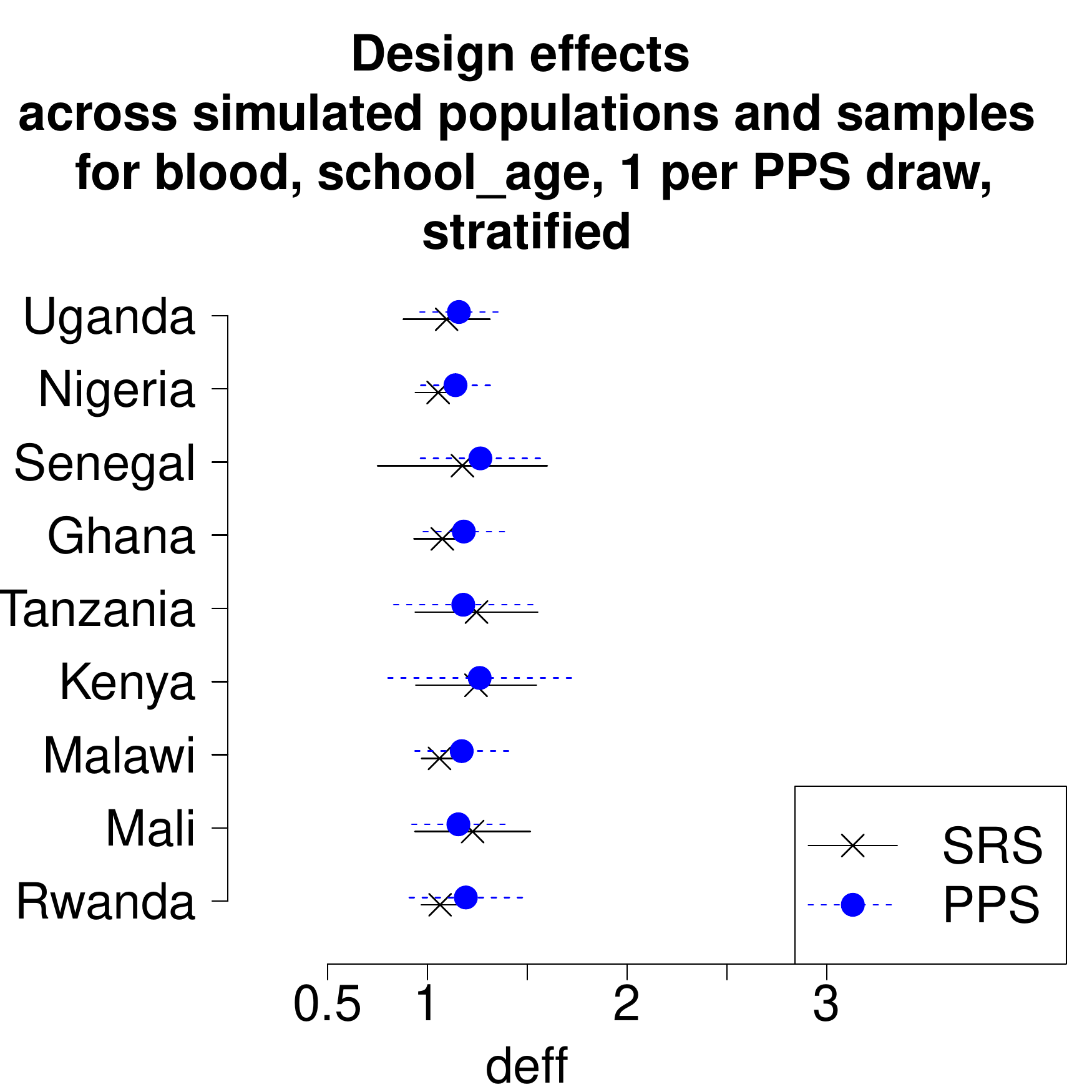}} \\
\subcaptionbox{\label{n_RDT_under_5_1_per_PPS_draw_stratified}}
 [0.49\textwidth]{\includegraphics[width=0.34\textwidth]{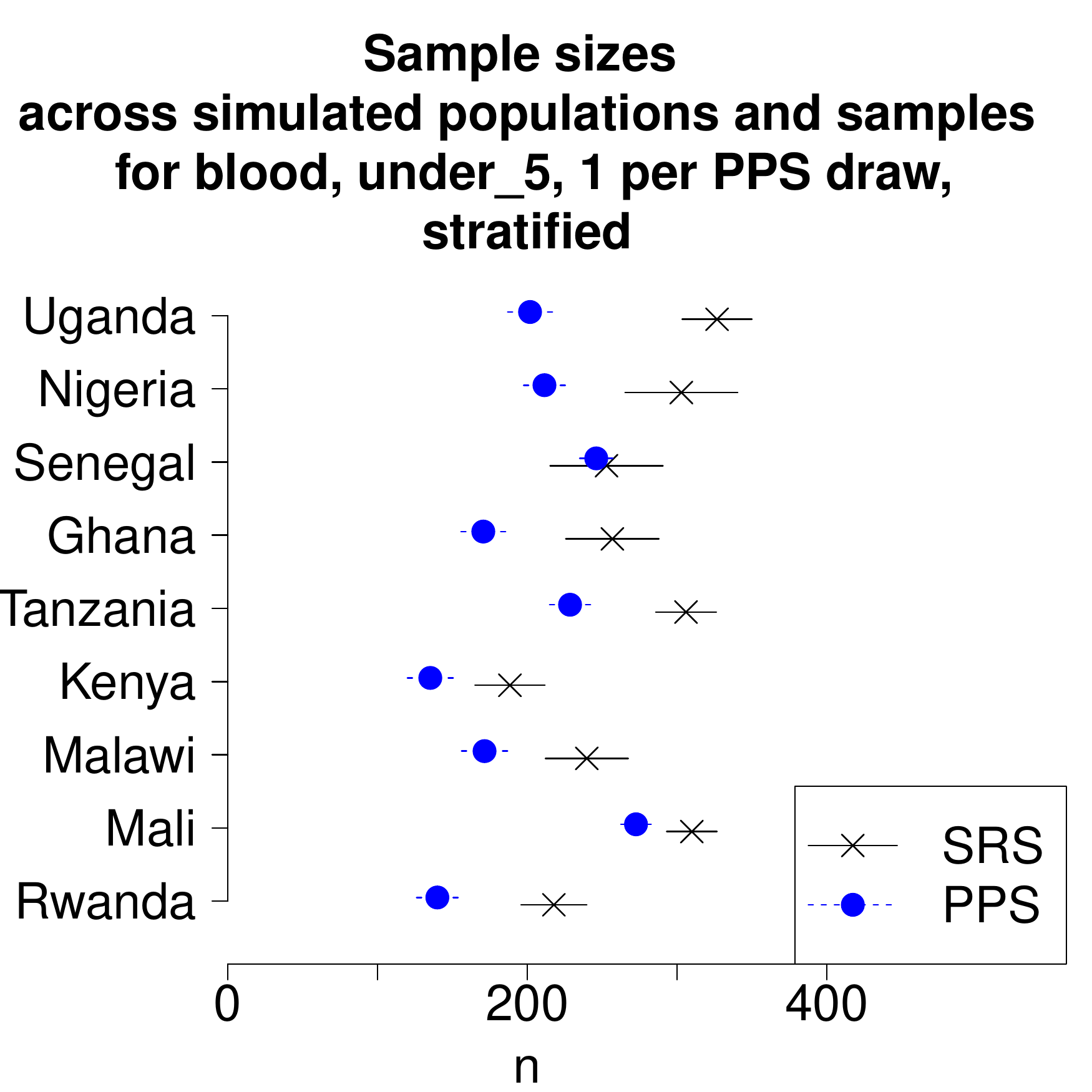}}
\subcaptionbox{\label{n_RDT_school_age_1_per_PPS_draw_stratified}}
 [0.49\textwidth]{\includegraphics[width=0.34\textwidth]{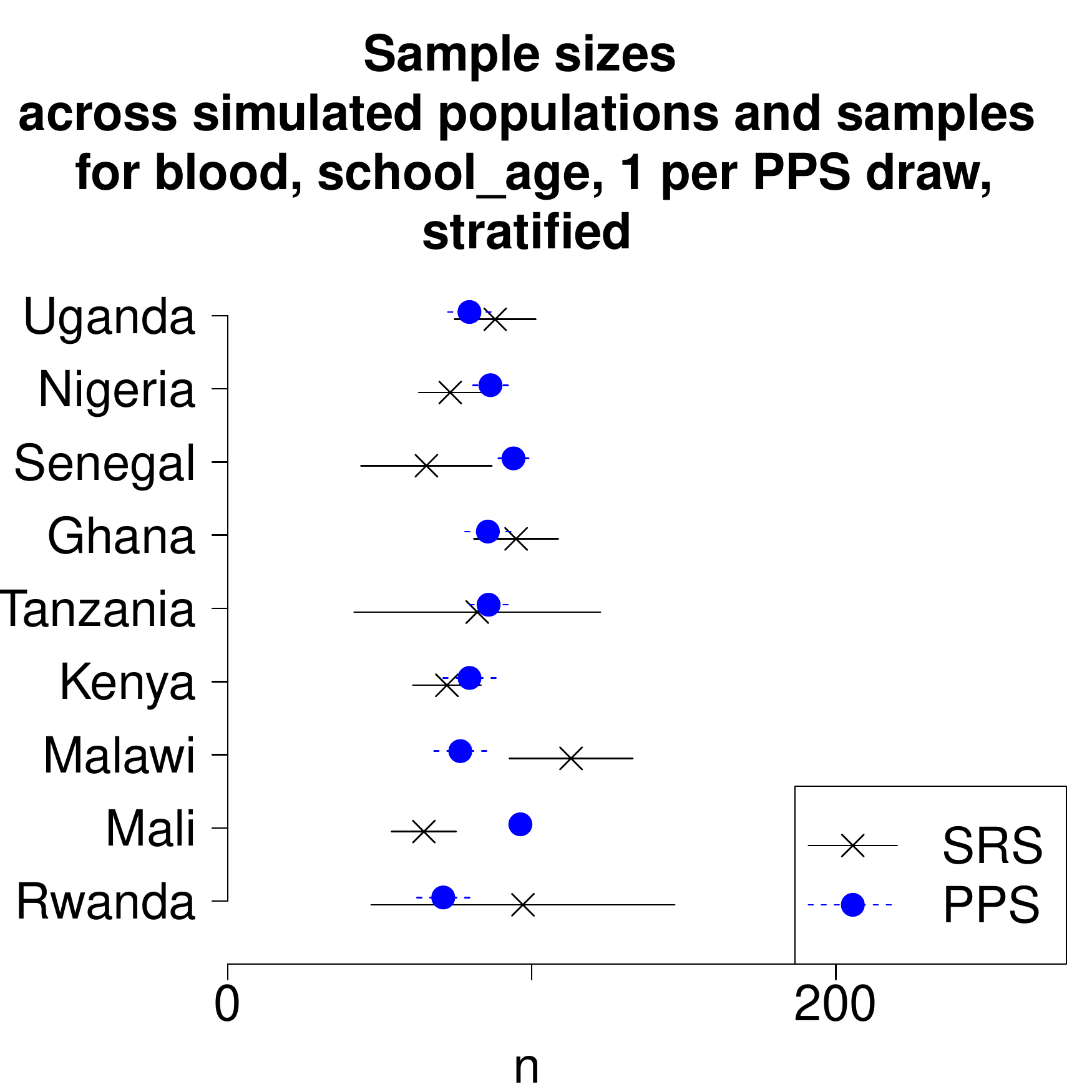}} \\
\subcaptionbox{\label{n_eff_RDT_under_5_1_per_PPS_draw_design_based_stratified}}
 [0.49\textwidth]{\includegraphics[width=0.34\textwidth]{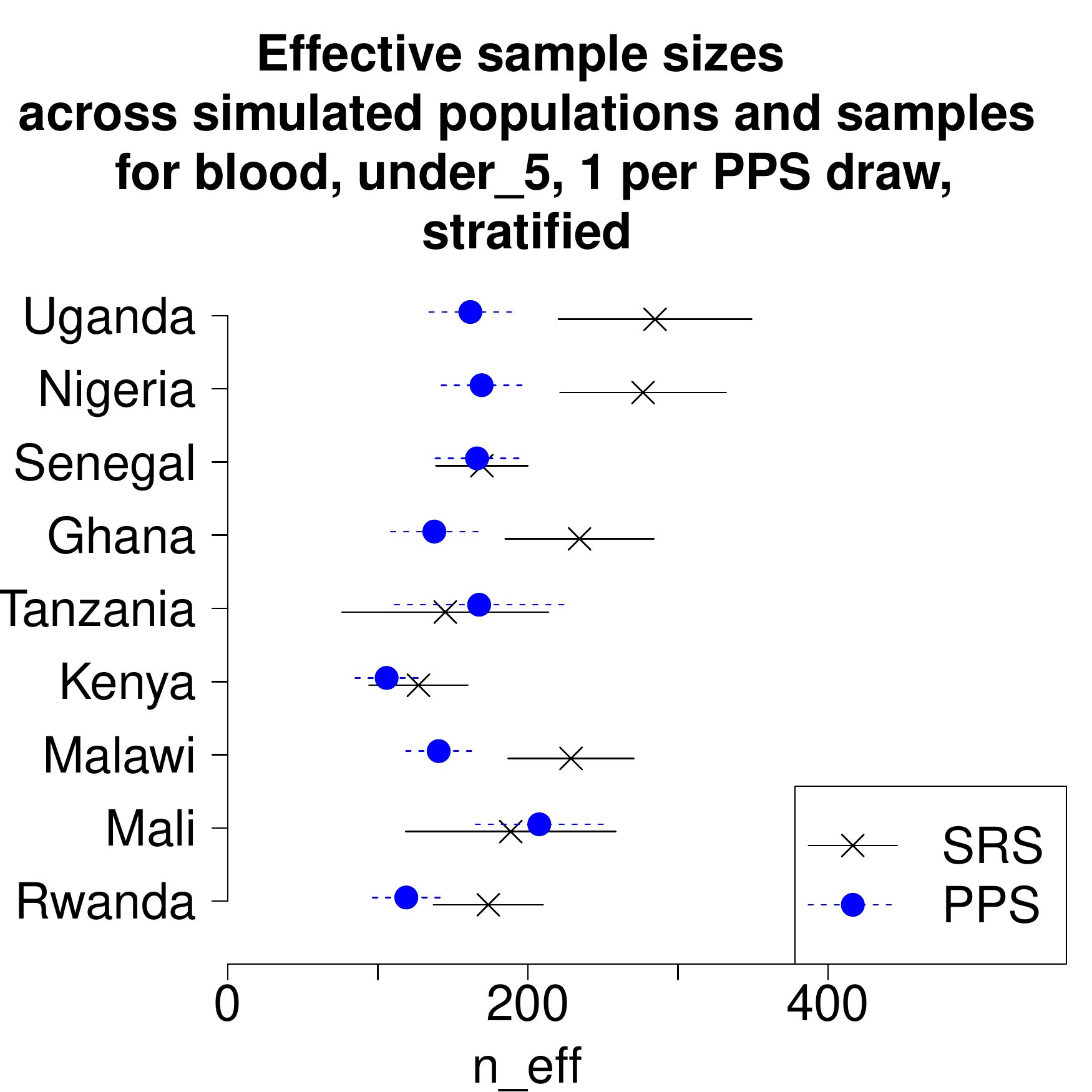}}
\subcaptionbox{\label{n_eff_RDT_school_age_1_per_PPS_draw_design_based_stratified}}
 [0.49\textwidth]{\includegraphics[width=0.34\textwidth]{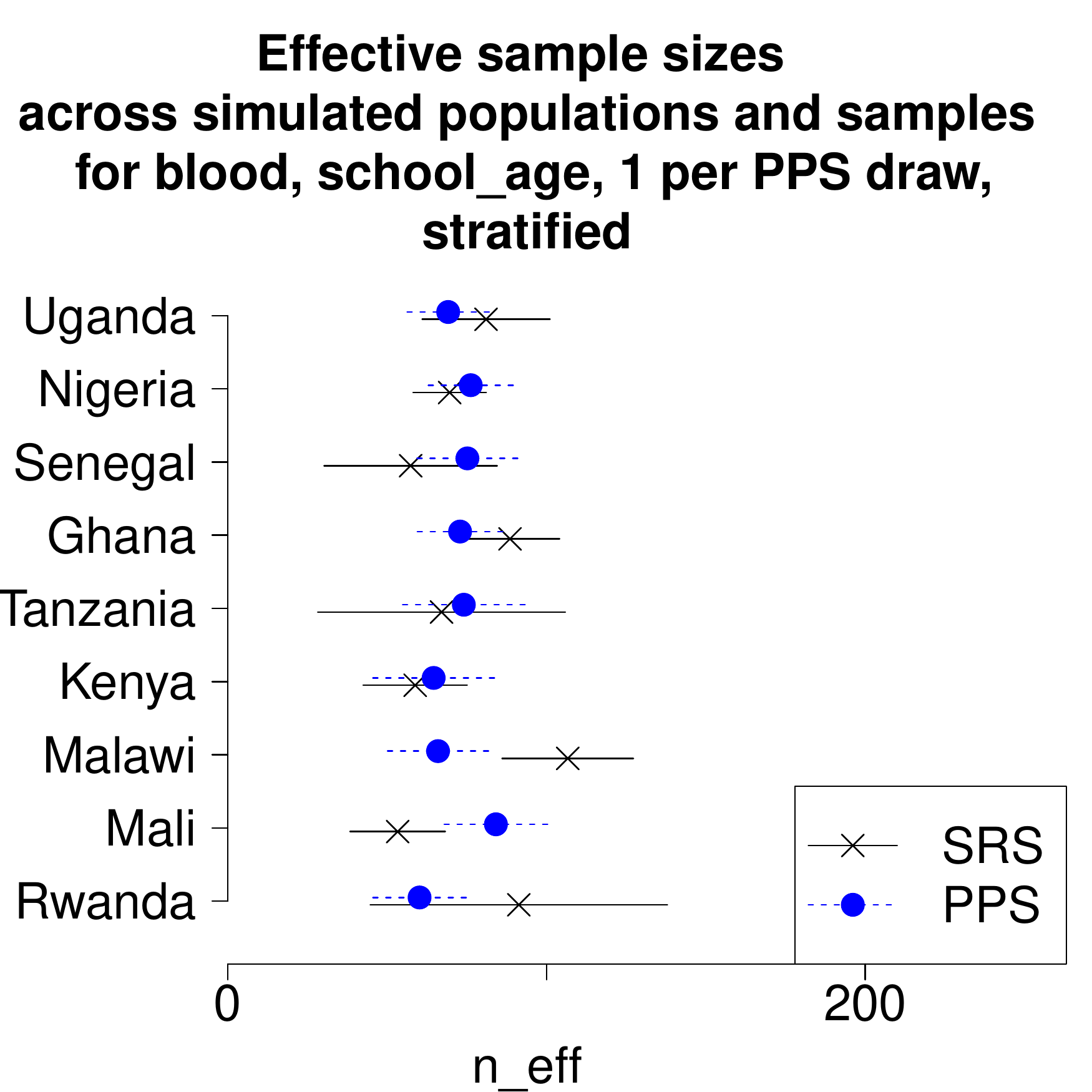}} 
\caption[]{Design-based blood (malaria and anemia) module results: under 5 and school-age children}\label{RDT_children_sampling_results_design_based_stratified}
\end{figure}

\begin{figure}[h!]
    \centering
\subcaptionbox{\label{deff_RDT_man_1_per_PPS_draw_design_based_stratified}}
 [0.49\textwidth]{\includegraphics[width=0.34\textwidth]{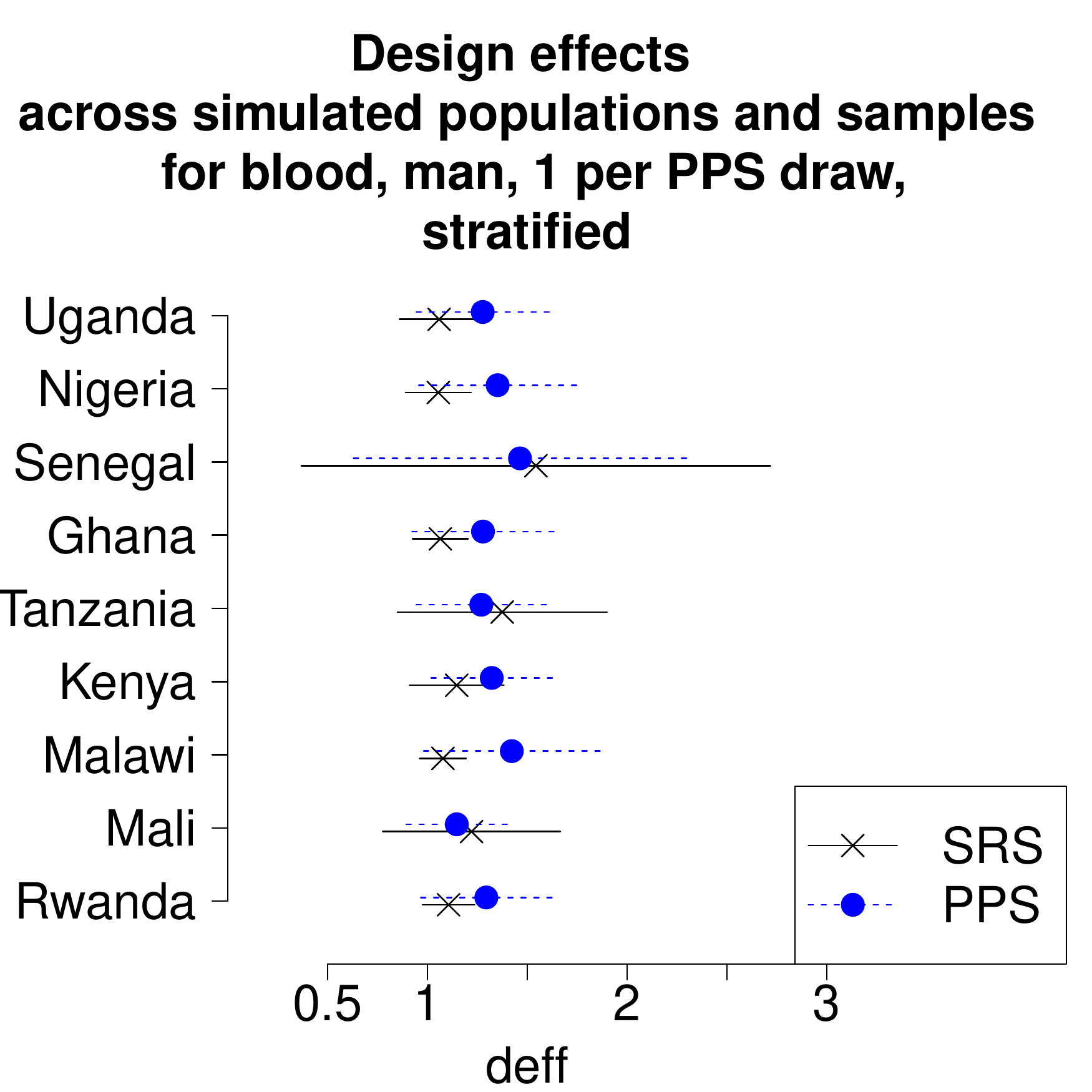}}
\subcaptionbox{\label{deff_RDT_woman_1_per_PPS_draw_design_based_stratified}}
 [0.49\textwidth]{\includegraphics[width=0.34\textwidth]{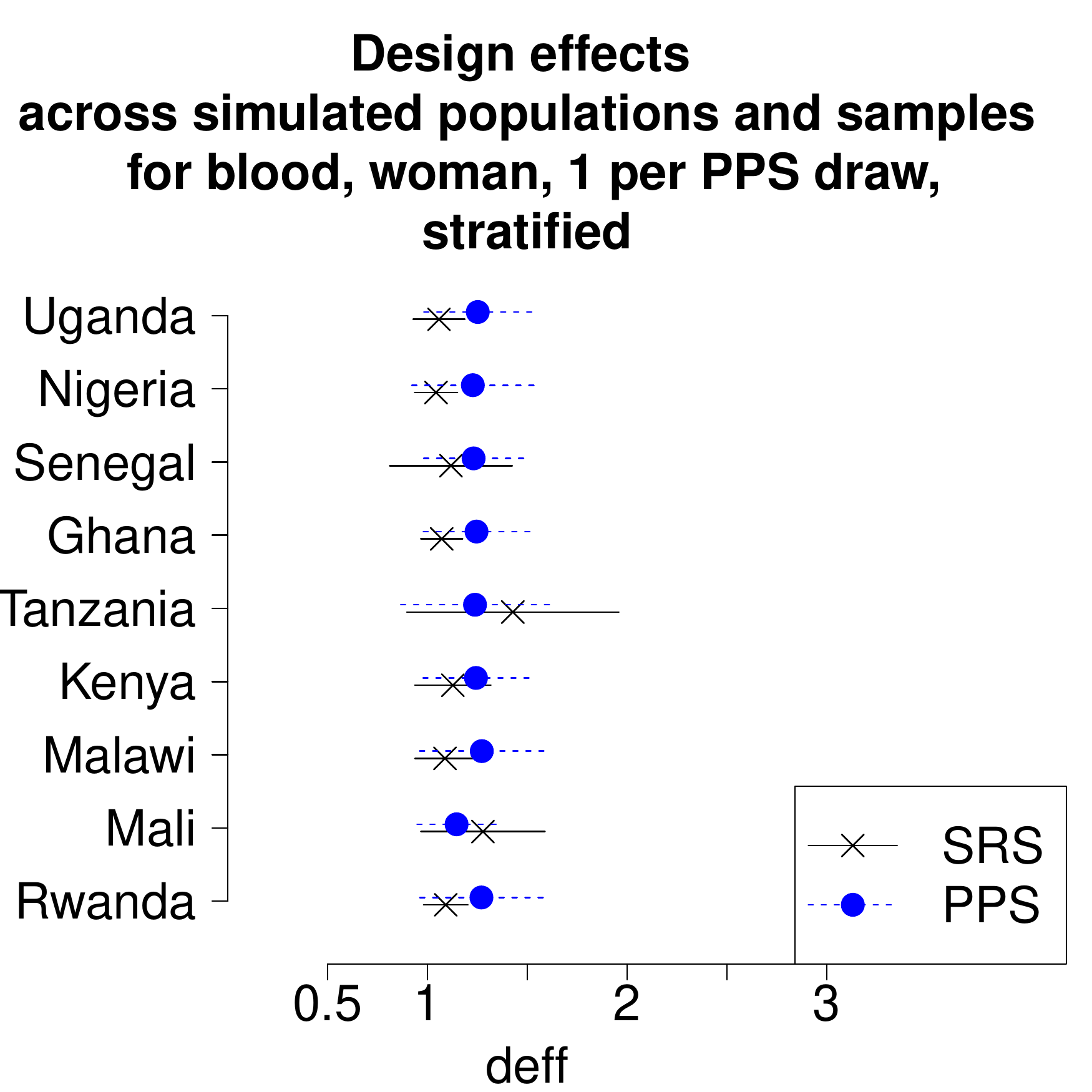}} \\
\subcaptionbox{\label{n_RDT_man_1_per_PPS_draw_stratified}}
 [0.49\textwidth]{\includegraphics[width=0.34\textwidth]{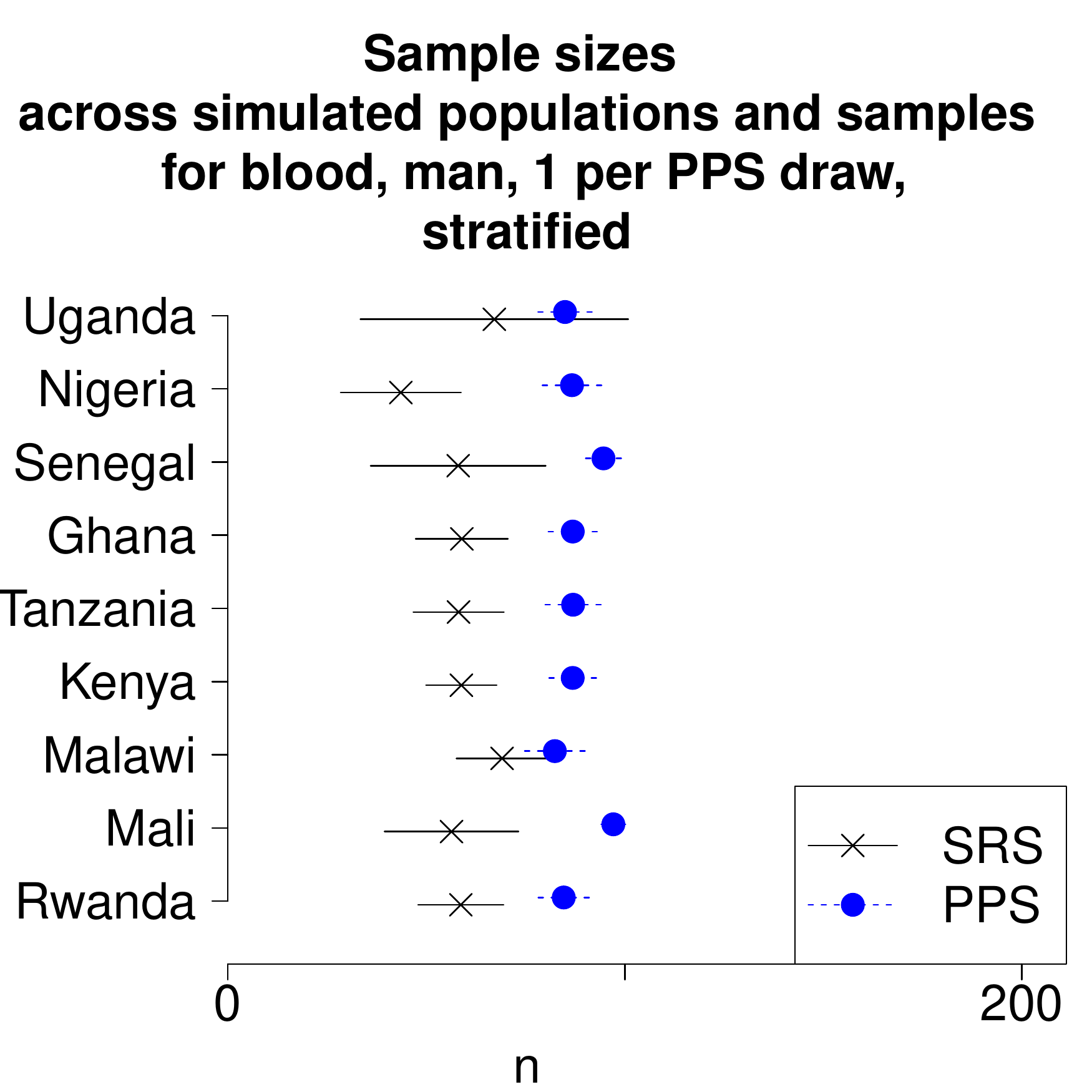}}
\subcaptionbox{\label{n_RDT_woman_1_per_PPS_draw_stratified}}
 [0.49\textwidth]{\includegraphics[width=0.34\textwidth]{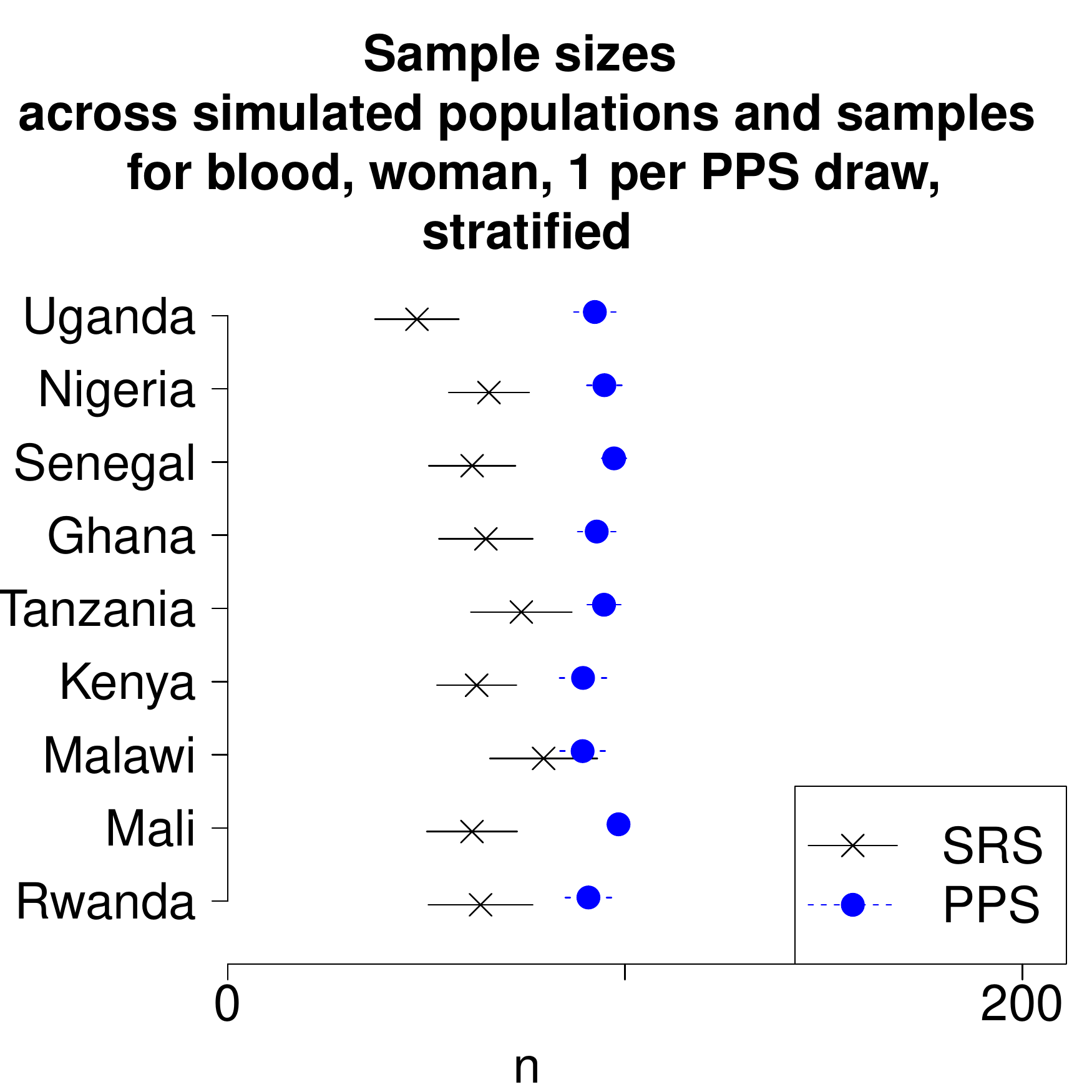}} \\
\subcaptionbox{\label{n_eff_RDT_man_1_per_PPS_draw_design_based_stratified}}
 [0.49\textwidth]{\includegraphics[width=0.34\textwidth]{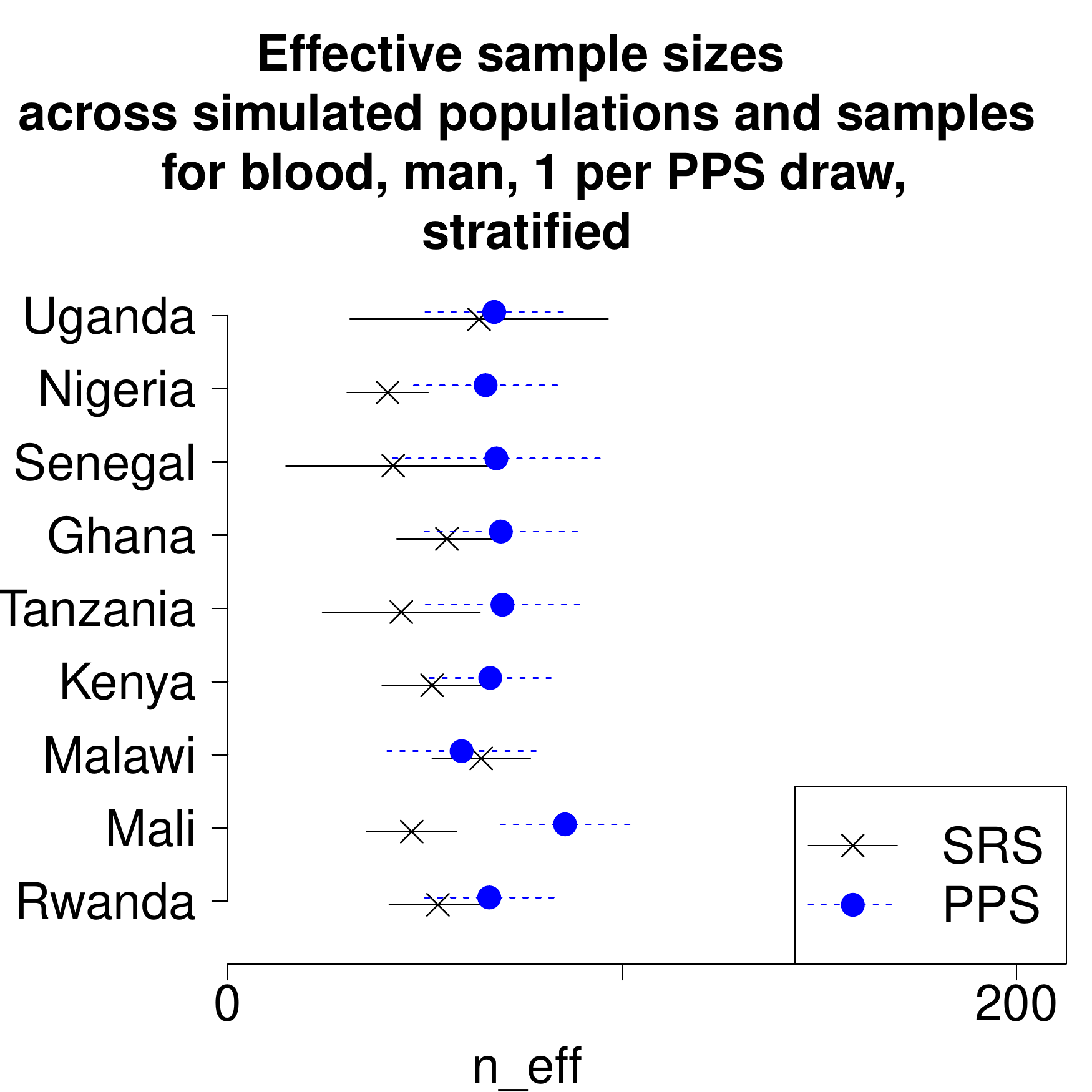}}
\subcaptionbox{\label{n_eff_RDT_woman_1_per_PPS_draw_design_based_stratified}}
 [0.49\textwidth]{\includegraphics[width=0.34\textwidth]{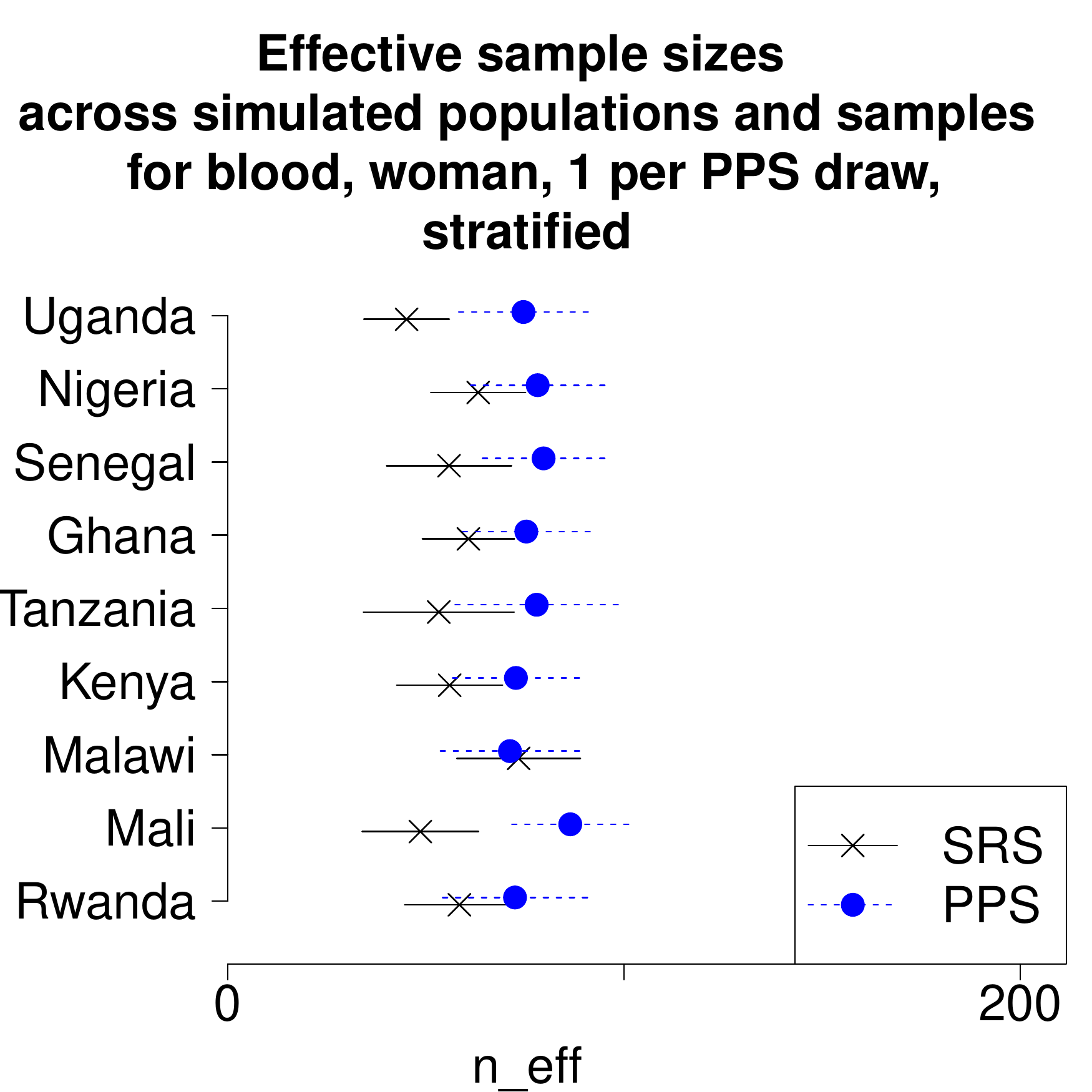}} 
\caption[]{Design-based blood (malaria and anemia) module results: men and women}\label{RDT_man_woman_sampling_results_design_based_stratified}
\end{figure}

\begin{figure}[h!]
    \centering
\subcaptionbox{\label{deff_anthro_under_5_1_per_PPS_draw_design_based_stratified}}
 [0.49\textwidth]{\includegraphics[width=0.34\textwidth]{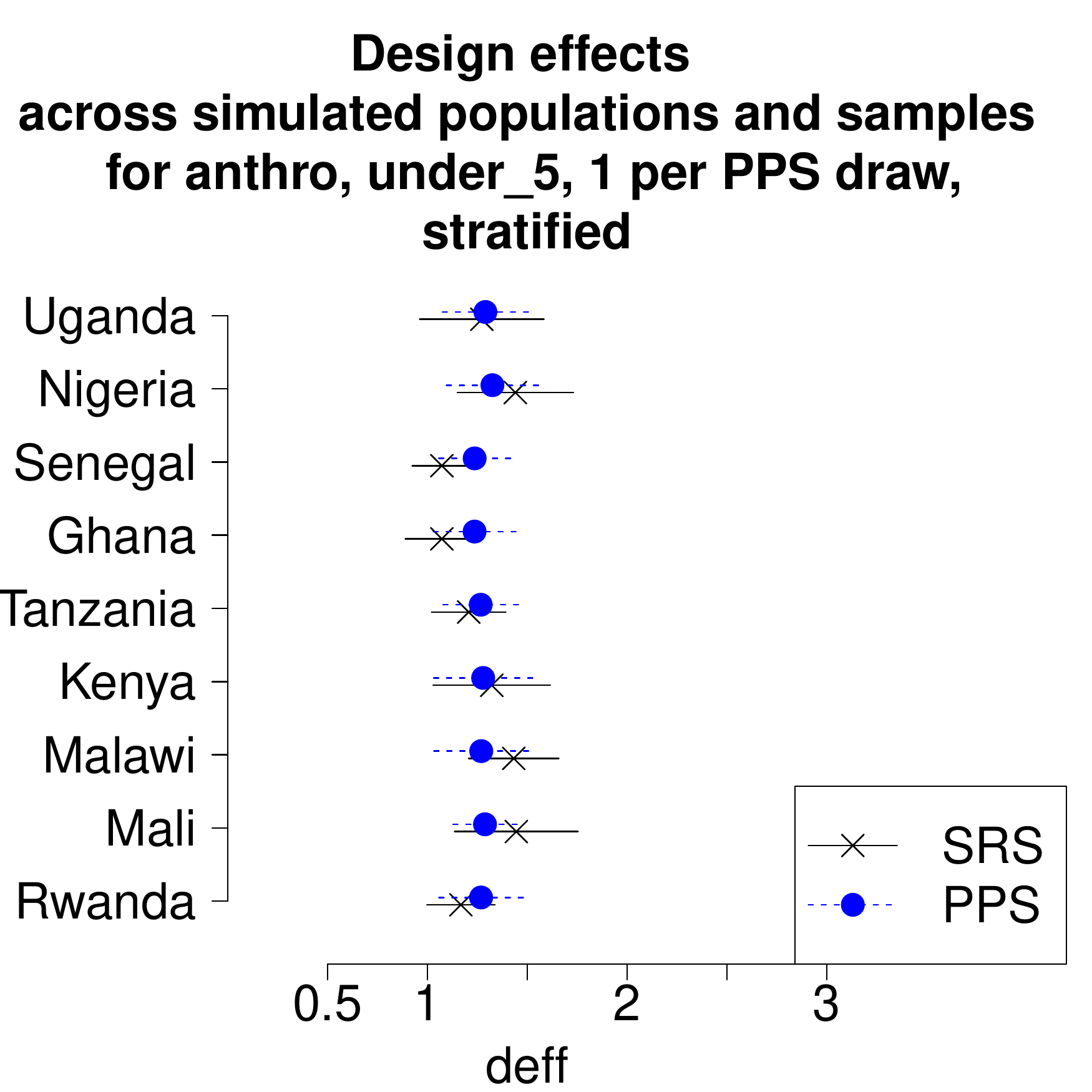}} \\
\subcaptionbox{\label{n_anthro_under_5_1_per_PPS_draw_stratified}}
 [0.49\textwidth]{\includegraphics[width=0.34\textwidth]{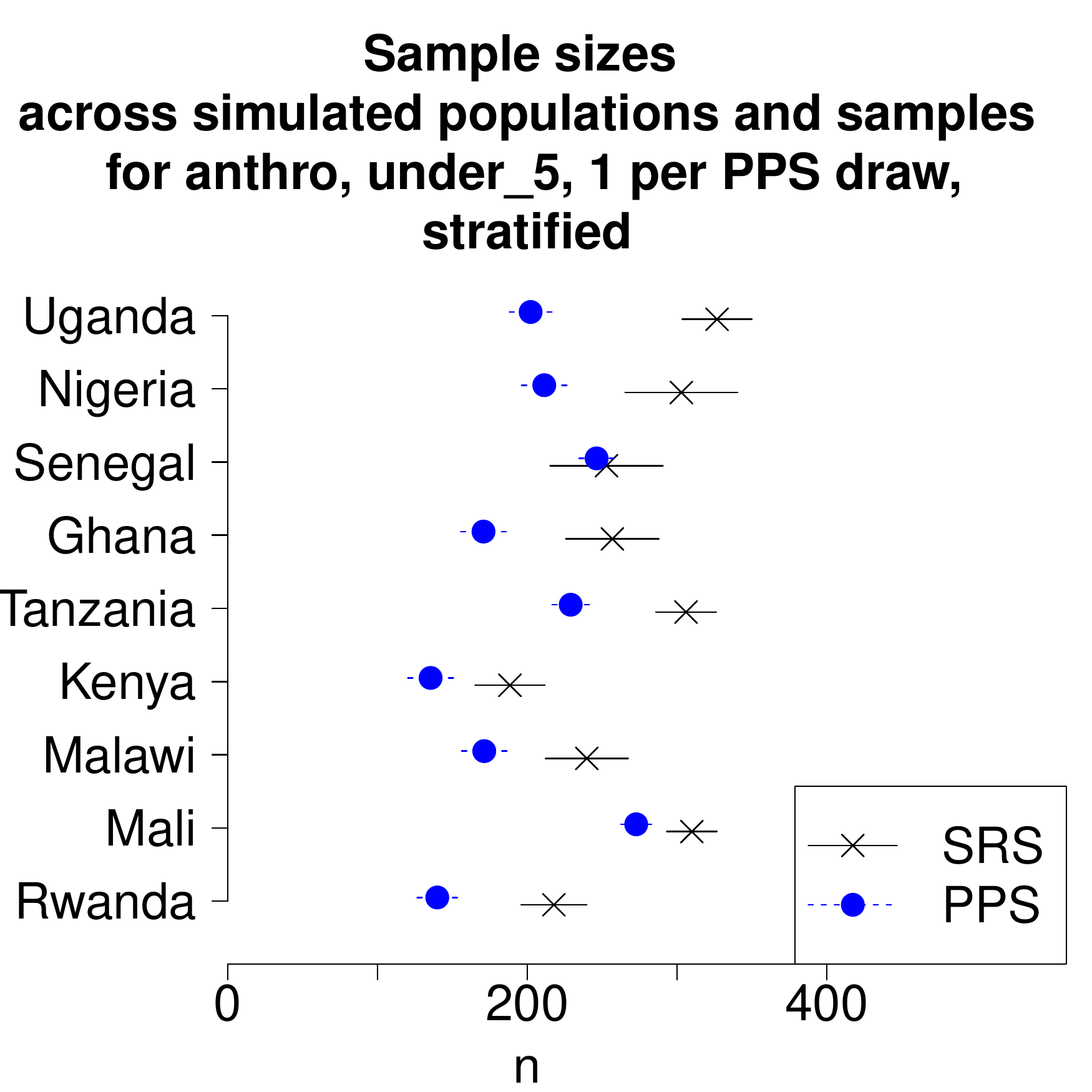}} \\
\subcaptionbox{\label{n_eff_anthro_under_5_1_per_PPS_draw_design_based_stratified}}
 [0.49\textwidth]{\includegraphics[width=0.34\textwidth]{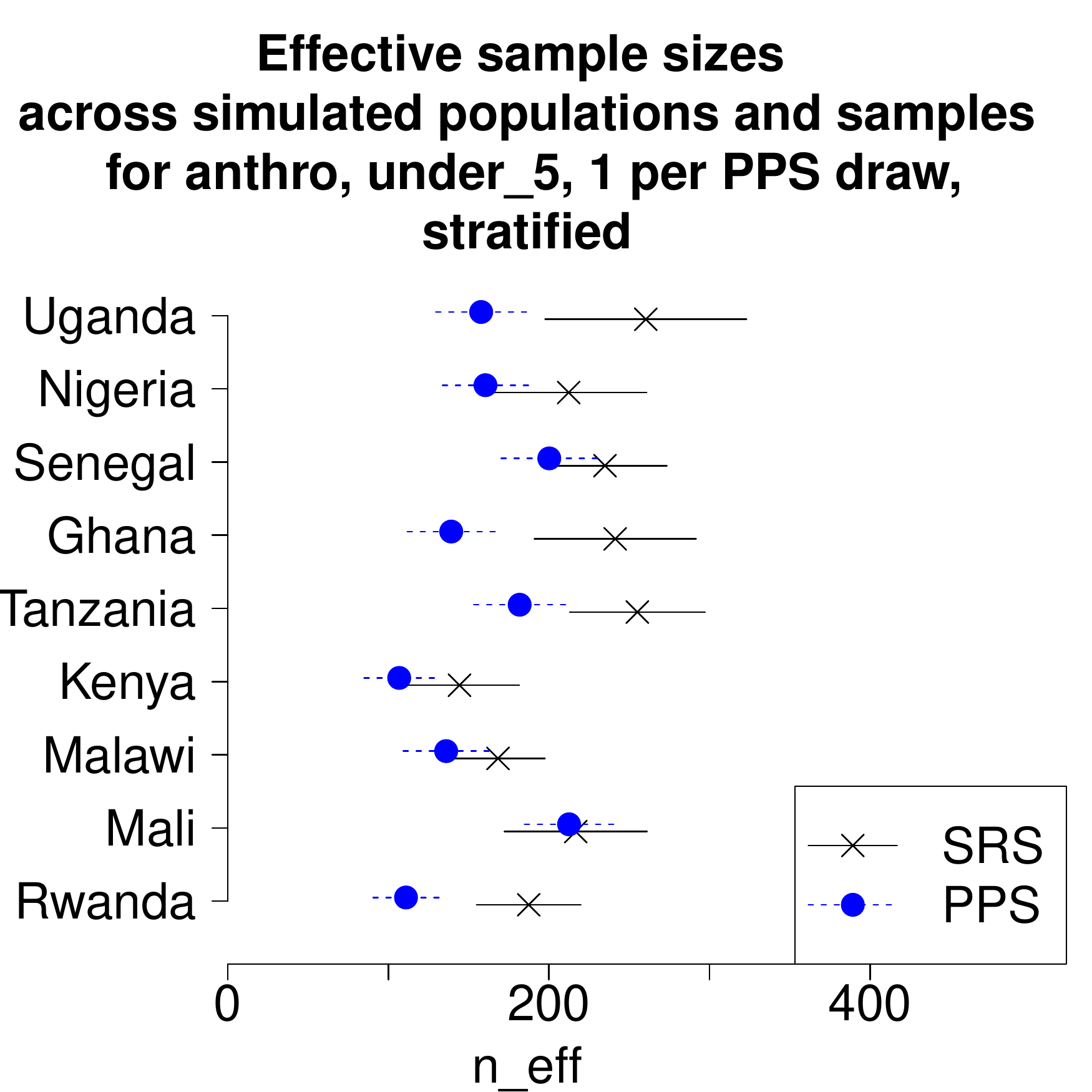}}
\caption[]{Design-based anthro module results}\label{anthro_sampling_results_design_based_stratified}
\end{figure}

\begin{figure}[h!]
\centering
\includegraphics[width = 8 cm]{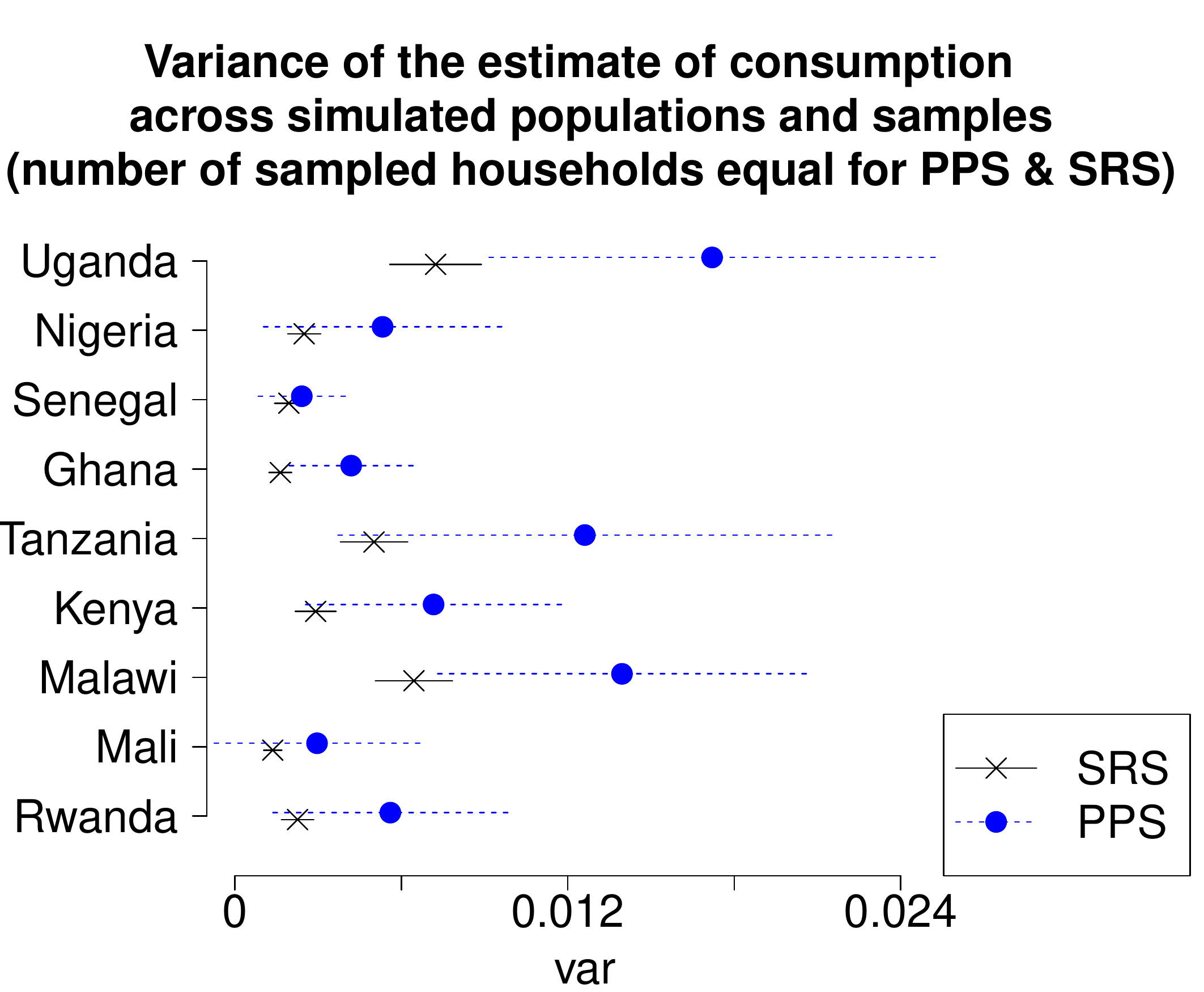}\caption{Design-based consumption module results}\label{var_hhsurvey_design_based_systematic}
\end{figure}

\clearpage
\newpage

\subsection{Model-based Results}

\begin{figure}[h!]
    \centering
\subcaptionbox{\label{deff_adult_man_1_per_PPS_draw_model_based_systematic}}
 [0.49\textwidth]{\includegraphics[width=0.34\textwidth]{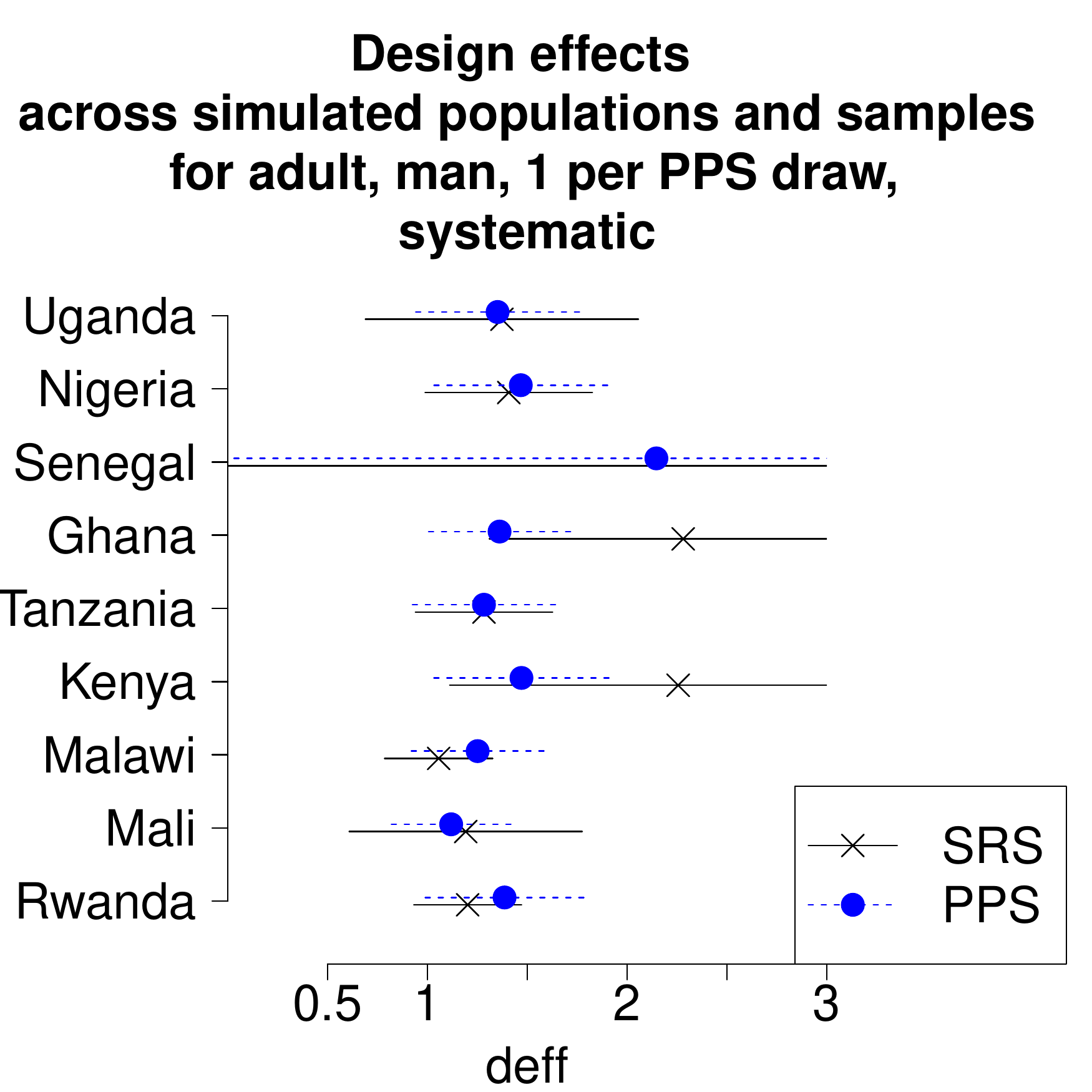}}
\subcaptionbox{\label{deff_adult_woman_1_per_PPS_draw_model_based_systematic}}
 [0.49\textwidth]{\includegraphics[width=0.34\textwidth]{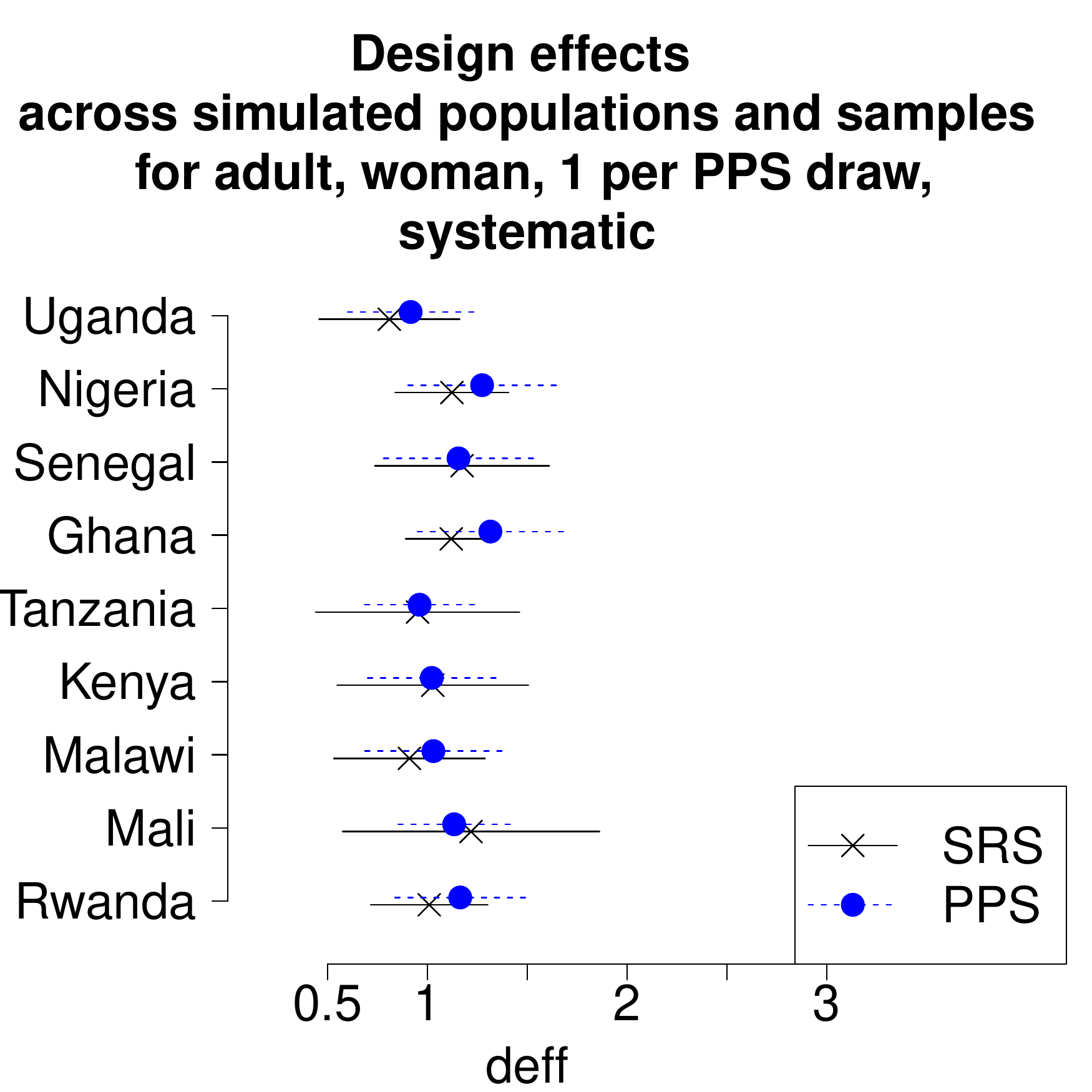}} \\
\subcaptionbox{\label{n_adult_man_1_per_PPS_draw_systematic}}
 [0.49\textwidth]{\includegraphics[width=0.34\textwidth]{n_adult_man_1_per_PPS_draw_systematic.pdf}}
\subcaptionbox{\label{n_adult_woman_1_per_PPS_draw_systematic}}
 [0.49\textwidth]{\includegraphics[width=0.34\textwidth]{n_adult_woman_1_per_PPS_draw_systematic.pdf}} \\
\subcaptionbox{\label{n_eff_adult_man_1_per_PPS_draw_model_based_systematic}}
 [0.49\textwidth]{\includegraphics[width=0.34\textwidth]{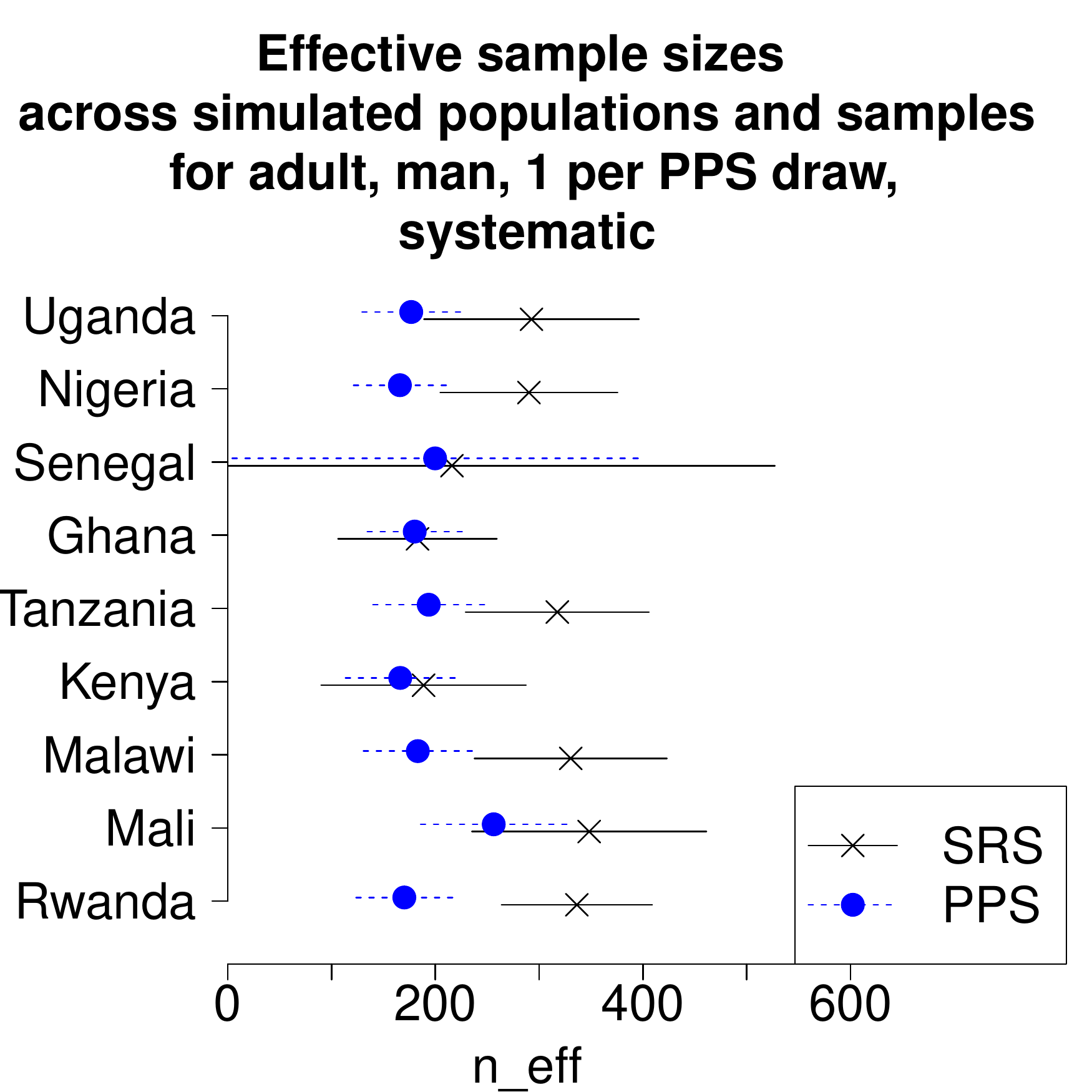}}
\subcaptionbox{\label{n_eff_adult_woman_1_per_PPS_draw_model_based_systematic}}
 [0.49\textwidth]{\includegraphics[width=0.34\textwidth]{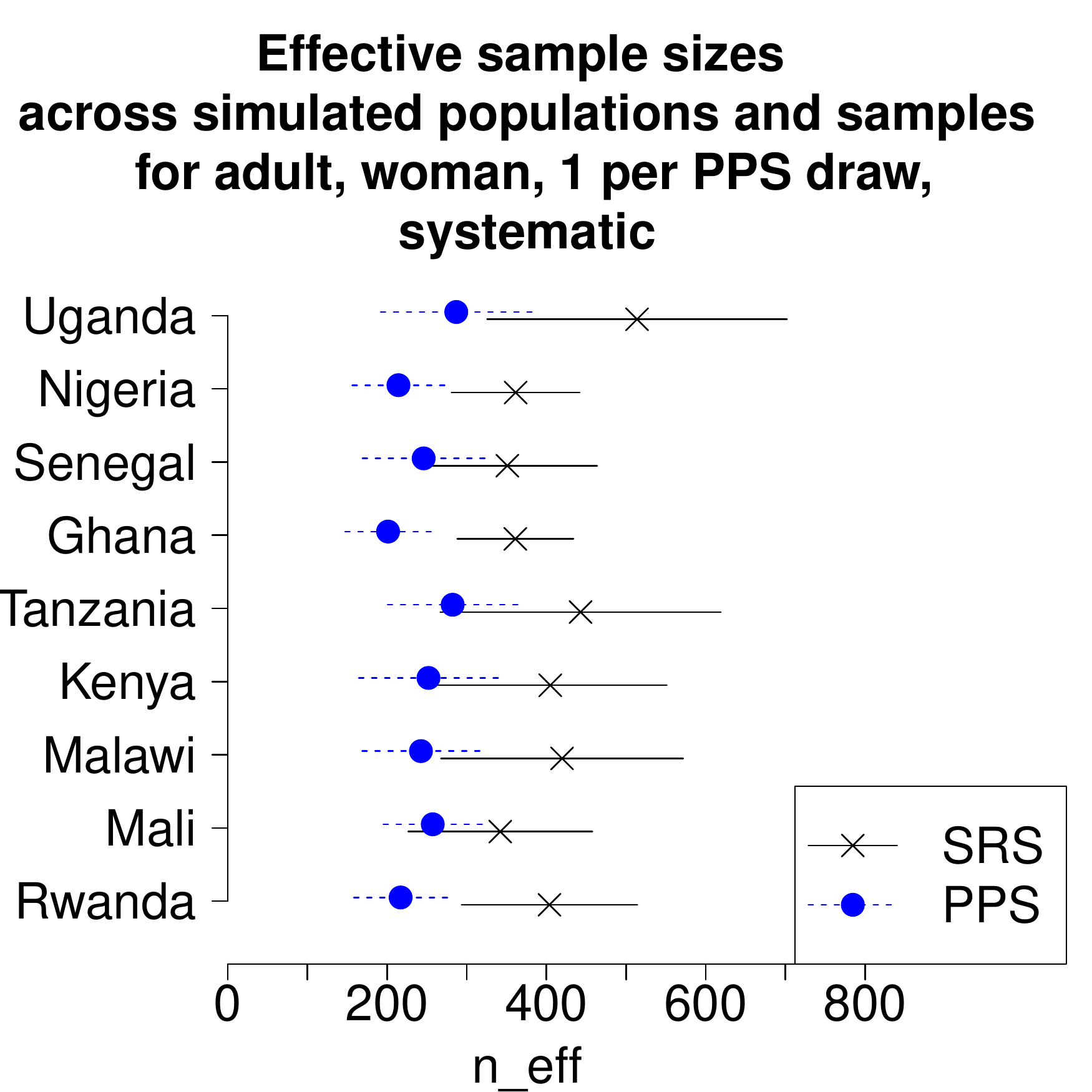}} 
\caption[]{Model-based adult module results}\label{adult_sampling_results_model_based_systematic}
\end{figure}

\begin{figure}[h!]
    \centering
\subcaptionbox{\label{deff_RDT_under_5_1_per_PPS_draw_model_based_systematic}}
 [0.49\textwidth]{\includegraphics[width=0.34\textwidth]{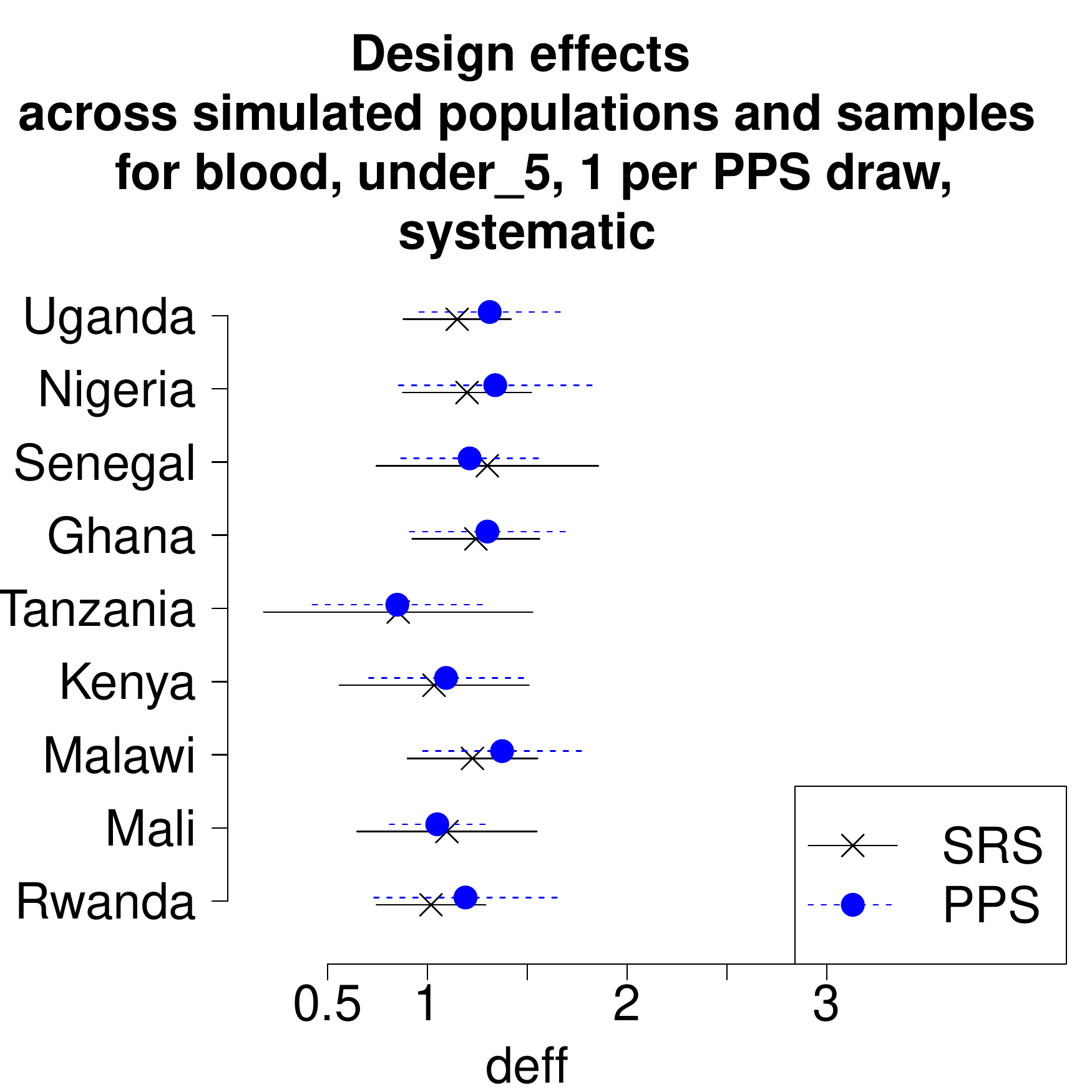}}
\subcaptionbox{\label{deff_RDT_school_age_1_per_PPS_draw_model_based_systematic}}
 [0.49\textwidth]{\includegraphics[width=0.34\textwidth]{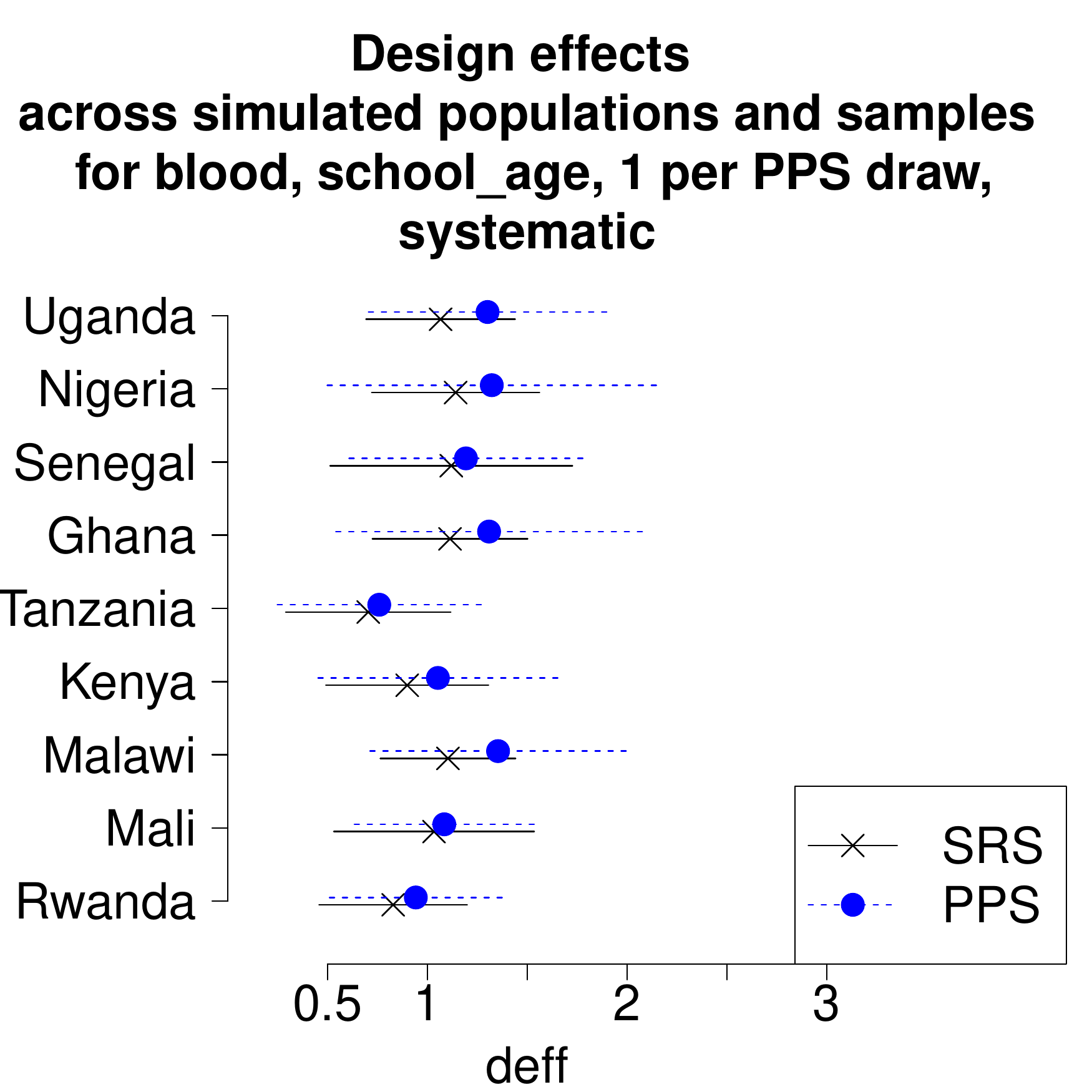}} \\
\subcaptionbox{\label{n_RDT_under_5_1_per_PPS_draw_systematic}}
 [0.49\textwidth]{\includegraphics[width=0.34\textwidth]{n_RDT_under_5_1_per_PPS_draw_systematic.pdf}}
\subcaptionbox{\label{n_RDT_school_age_1_per_PPS_draw_systematic}}
 [0.49\textwidth]{\includegraphics[width=0.34\textwidth]{n_RDT_school_age_1_per_PPS_draw_systematic.pdf}} \\
\subcaptionbox{\label{n_eff_RDT_under_5_1_per_PPS_draw_model_based_systematic}}
 [0.49\textwidth]{\includegraphics[width=0.34\textwidth]{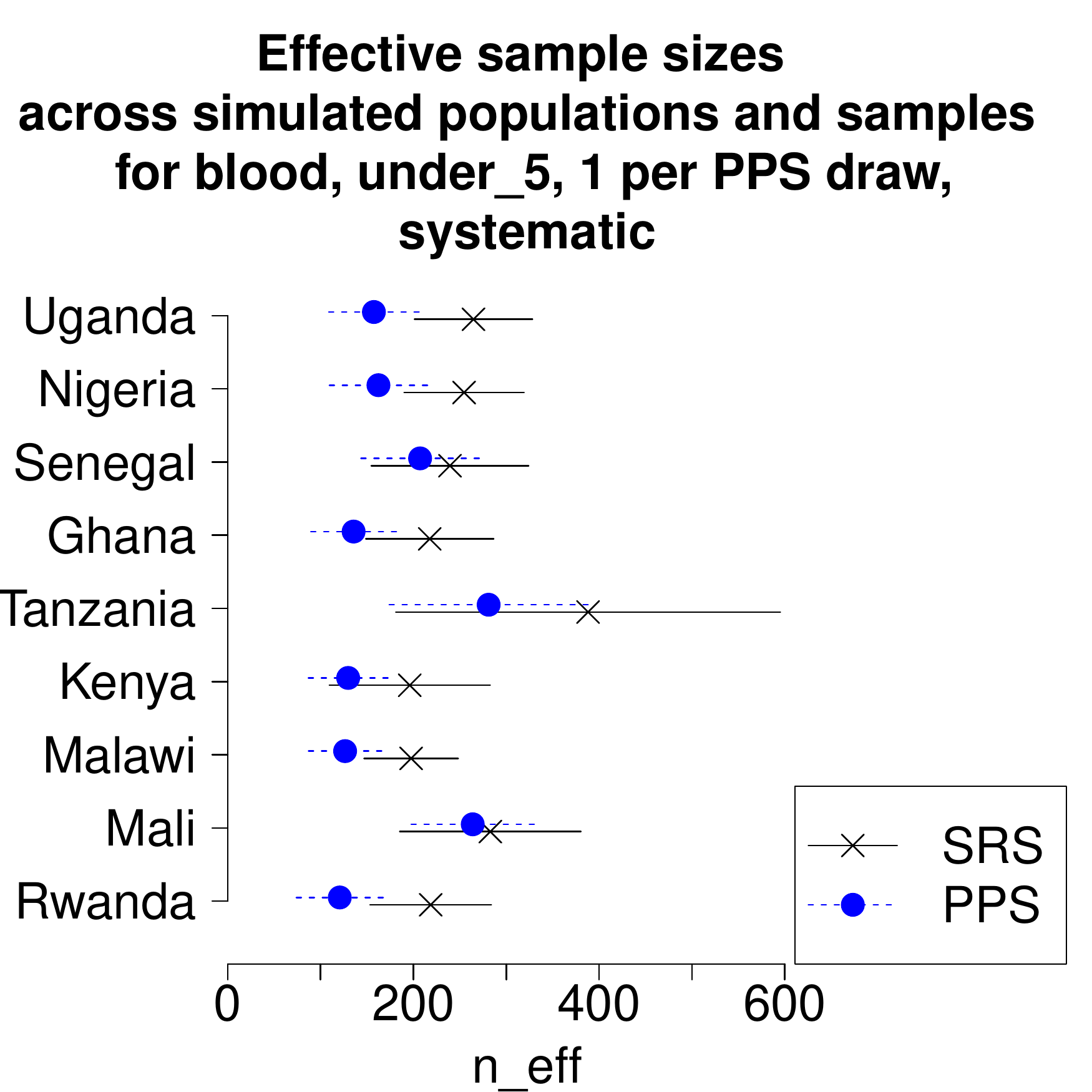}}
\subcaptionbox{\label{n_eff_RDT_school_age_1_per_PPS_draw_model_based_systematic}}
 [0.49\textwidth]{\includegraphics[width=0.34\textwidth]{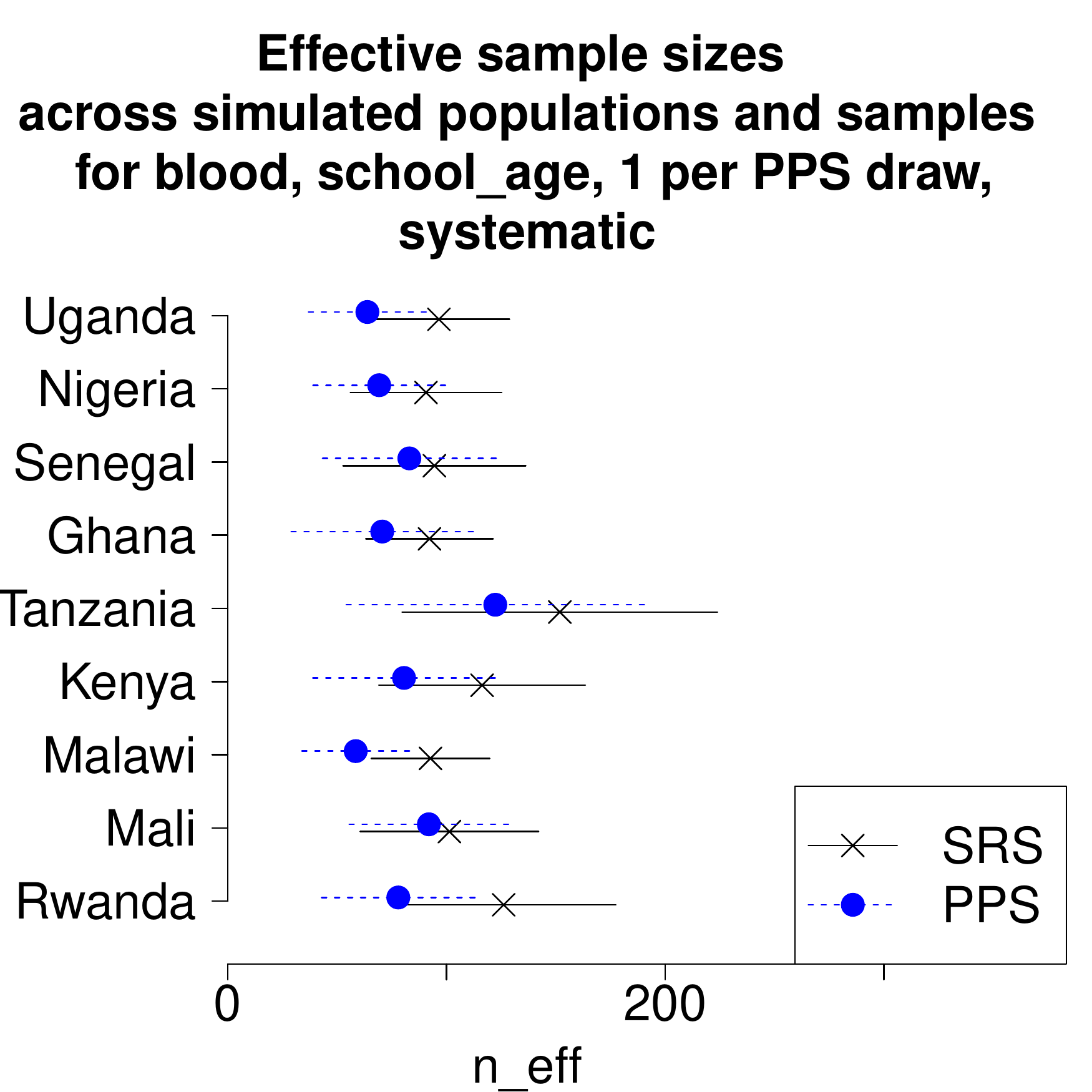}} 
\caption[]{Model-based blood (malaria and anemia) module results: under 5 and school-age children}\label{RDT_children_sampling_results_model_based_systematic}
\end{figure}

\begin{figure}[h!]
    \centering
\subcaptionbox{\label{deff_RDT_man_1_per_PPS_draw_model_based_systematic}}
 [0.49\textwidth]{\includegraphics[width=0.34\textwidth]{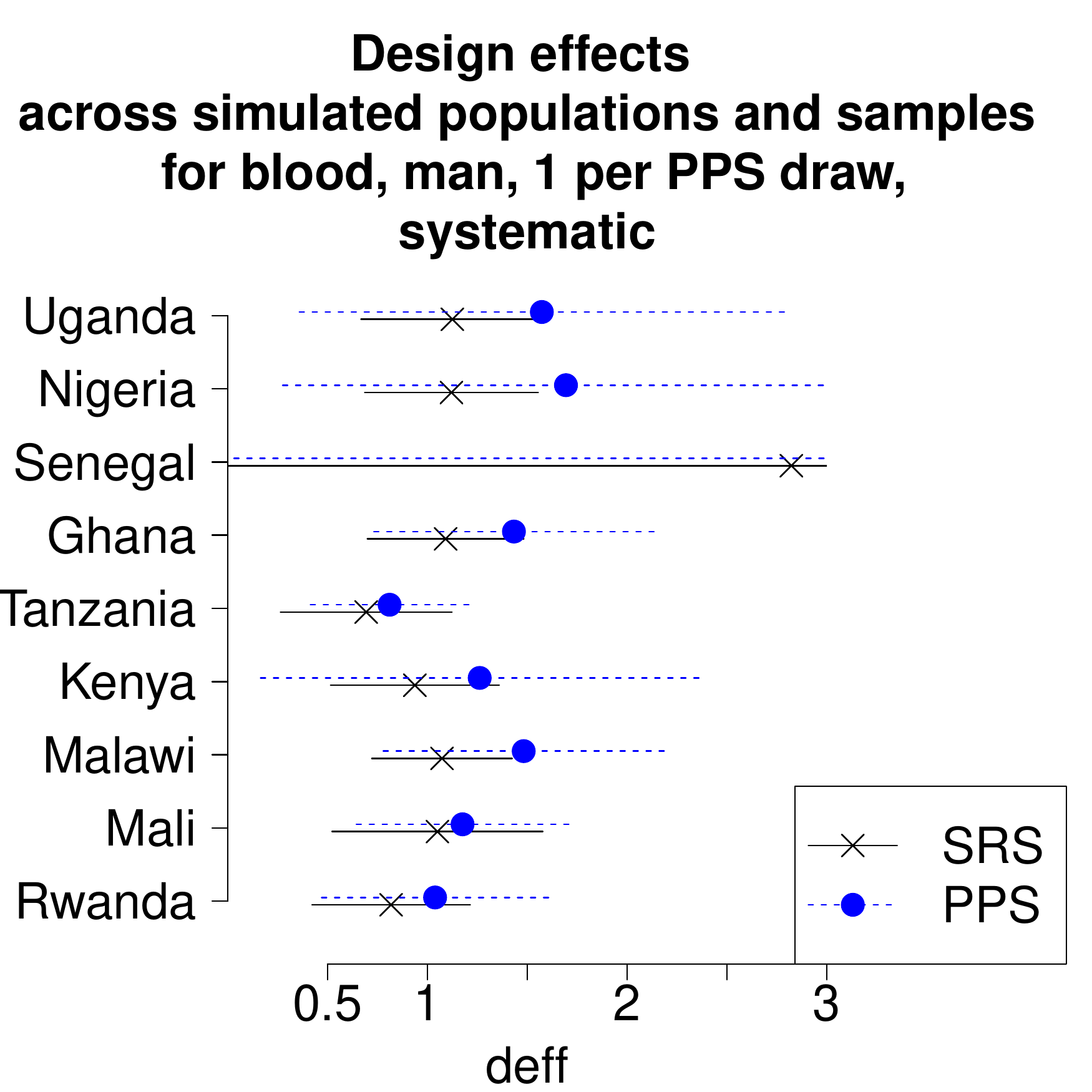}}
\subcaptionbox{\label{deff_RDT_woman_1_per_PPS_draw_model_based_systematic}}
 [0.49\textwidth]{\includegraphics[width=0.34\textwidth]{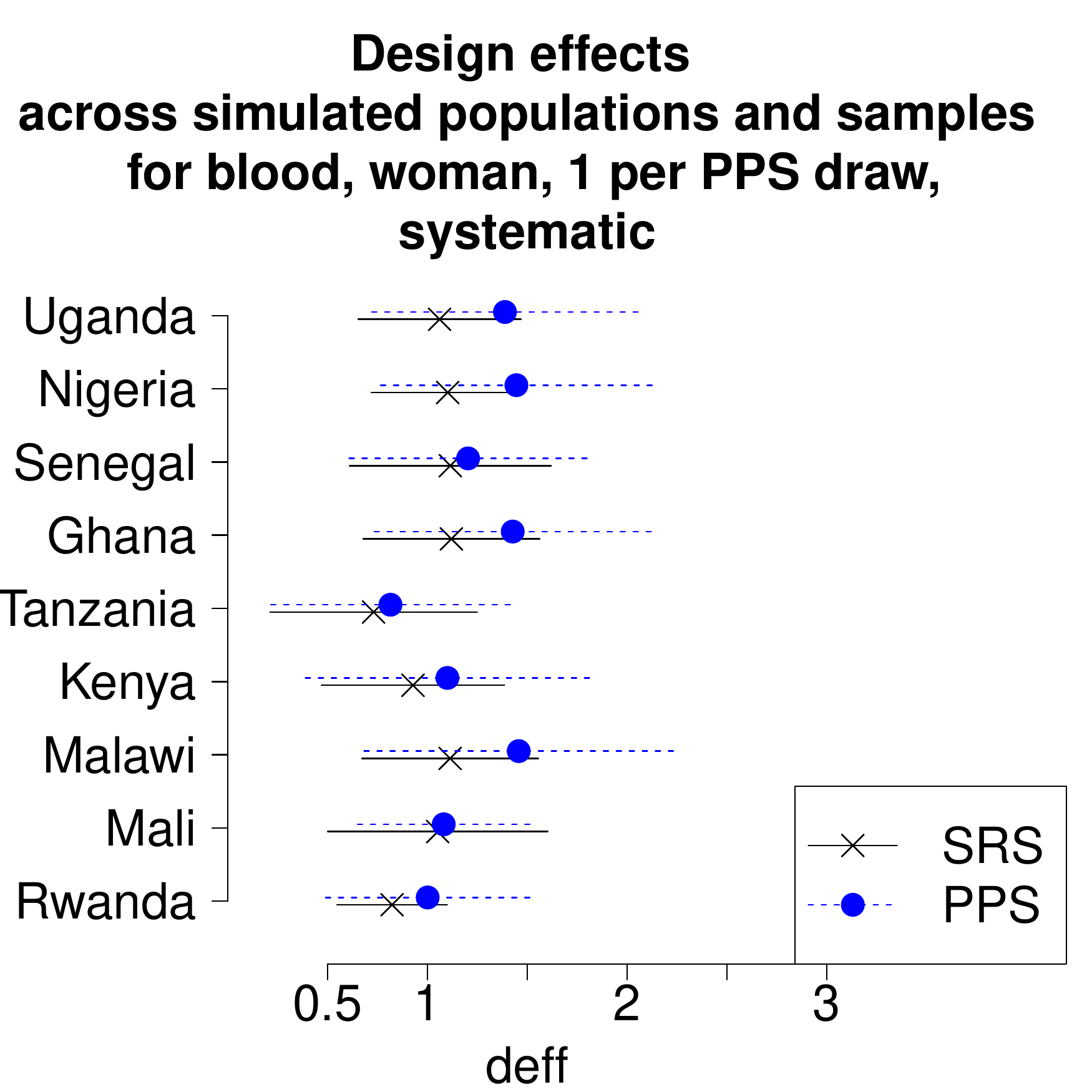}} \\
\subcaptionbox{\label{n_RDT_man_1_per_PPS_draw_systematic}}
 [0.49\textwidth]{\includegraphics[width=0.34\textwidth]{n_RDT_man_1_per_PPS_draw_systematic.pdf}}
\subcaptionbox{\label{n_RDT_woman_1_per_PPS_draw_systematic}}
 [0.49\textwidth]{\includegraphics[width=0.34\textwidth]{n_RDT_woman_1_per_PPS_draw_systematic.pdf}} \\
\subcaptionbox{\label{n_eff_RDT_man_1_per_PPS_draw_model_based_systematic}}
 [0.49\textwidth]{\includegraphics[width=0.34\textwidth]{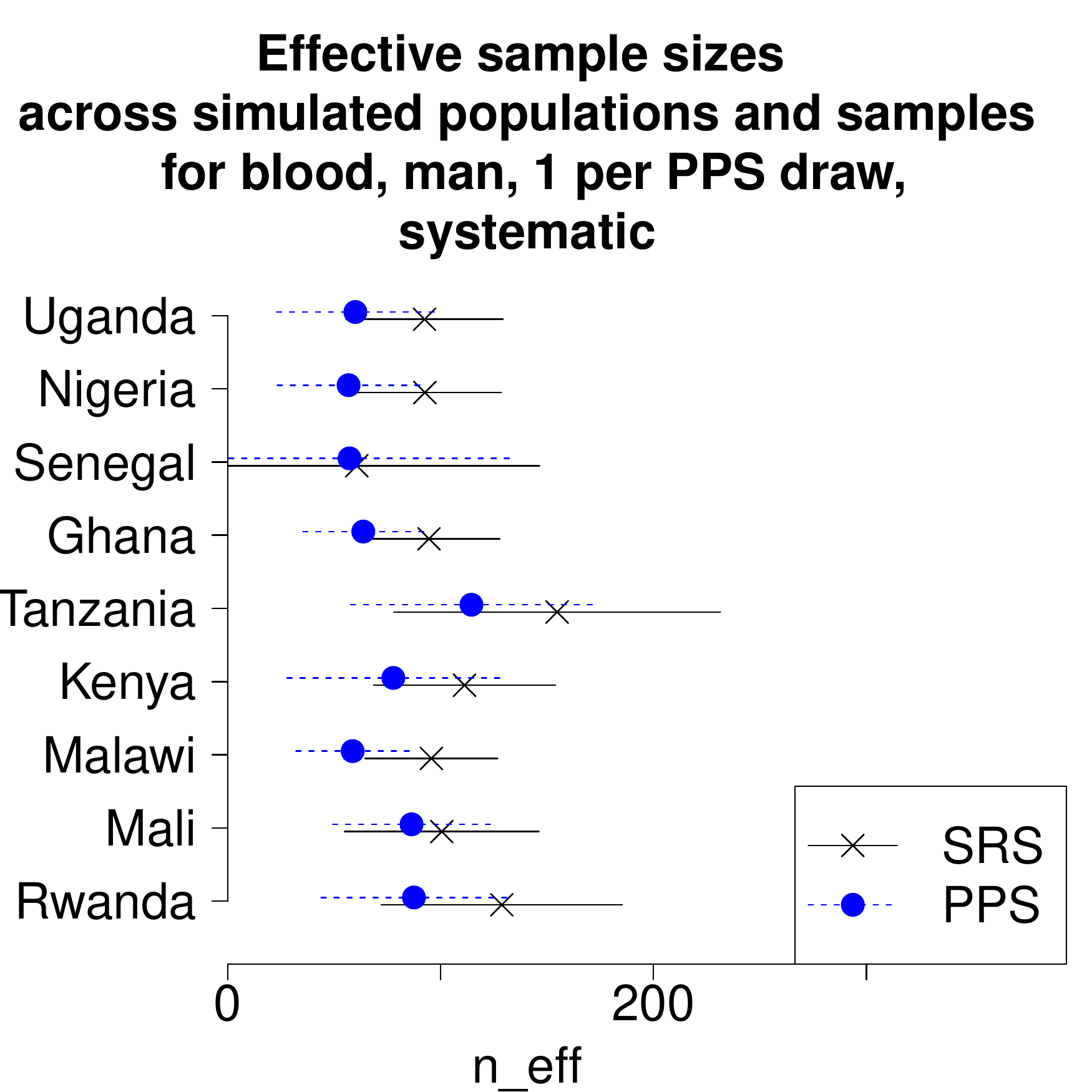}}
\subcaptionbox{\label{n_eff_RDT_woman_1_per_PPS_draw_model_based_systematic}}
 [0.49\textwidth]{\includegraphics[width=0.34\textwidth]{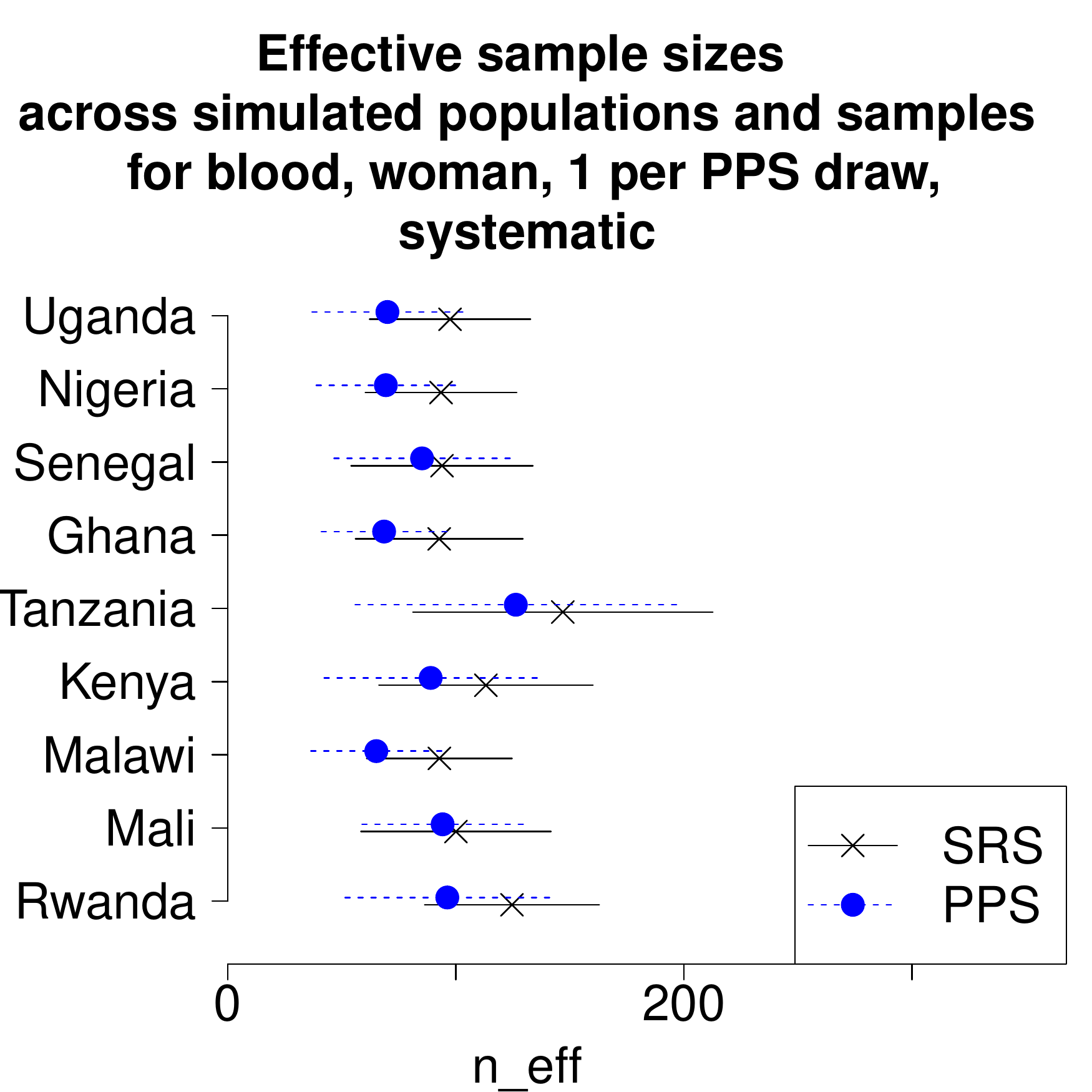}} 
\caption[]{Model-based blood (malaria and anemia) module results: men and women}\label{RDT_man_woman_sampling_results_model_based_systematic}
\end{figure}

\begin{figure}[h!]
    \centering
\subcaptionbox{\label{deff_anthro_under_5_1_per_PPS_draw_model_based_systematic}}
 [0.49\textwidth]{\includegraphics[width=0.34\textwidth]{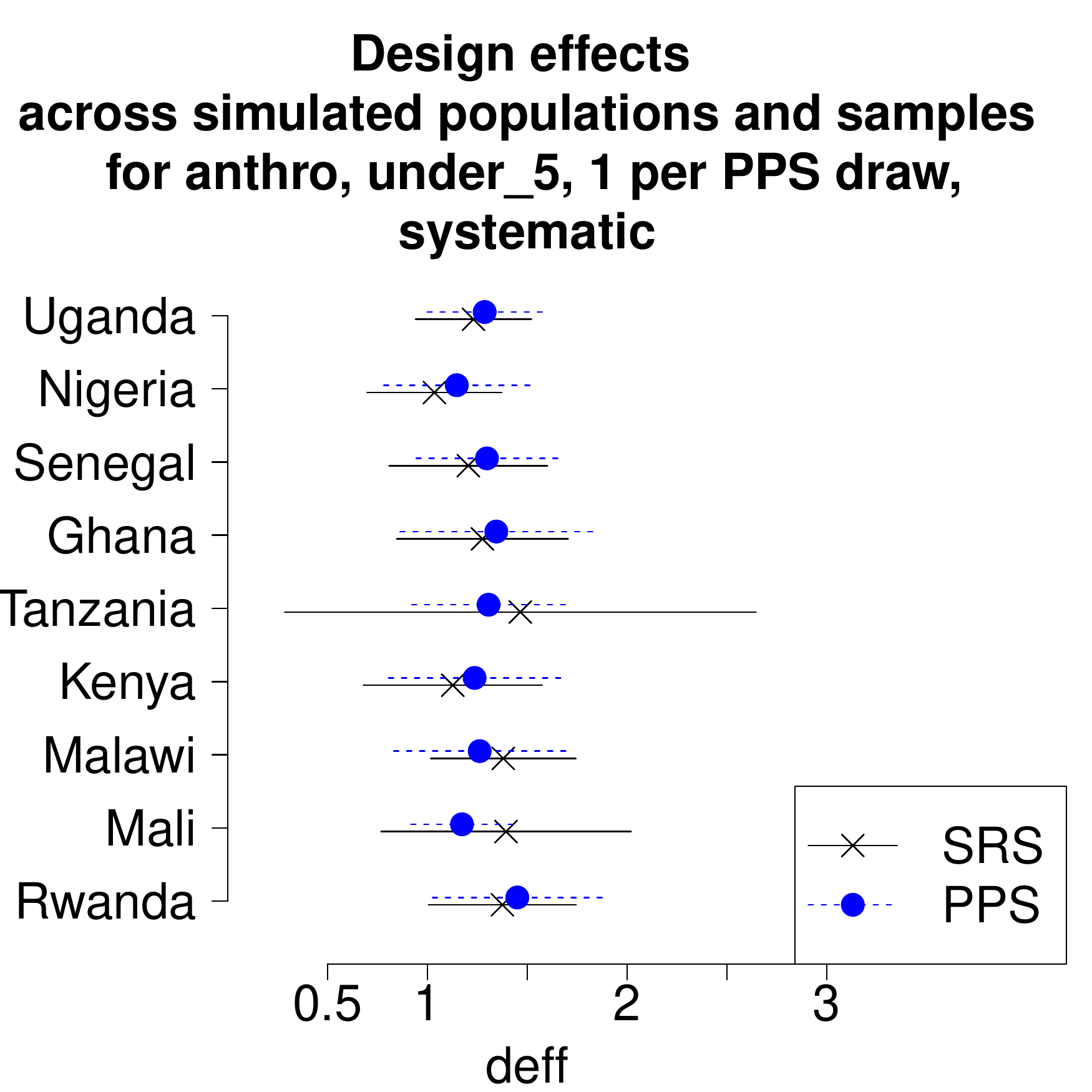}} \\
\subcaptionbox{\label{n_anthro_under_5_1_per_PPS_draw_systematic}}
 [0.49\textwidth]{\includegraphics[width=0.34\textwidth]{n_anthro_under_5_1_per_PPS_draw_systematic.pdf}} \\
\subcaptionbox{\label{n_eff_anthro_under_5_1_per_PPS_draw_model_based_systematic}}
 [0.49\textwidth]{\includegraphics[width=0.34\textwidth]{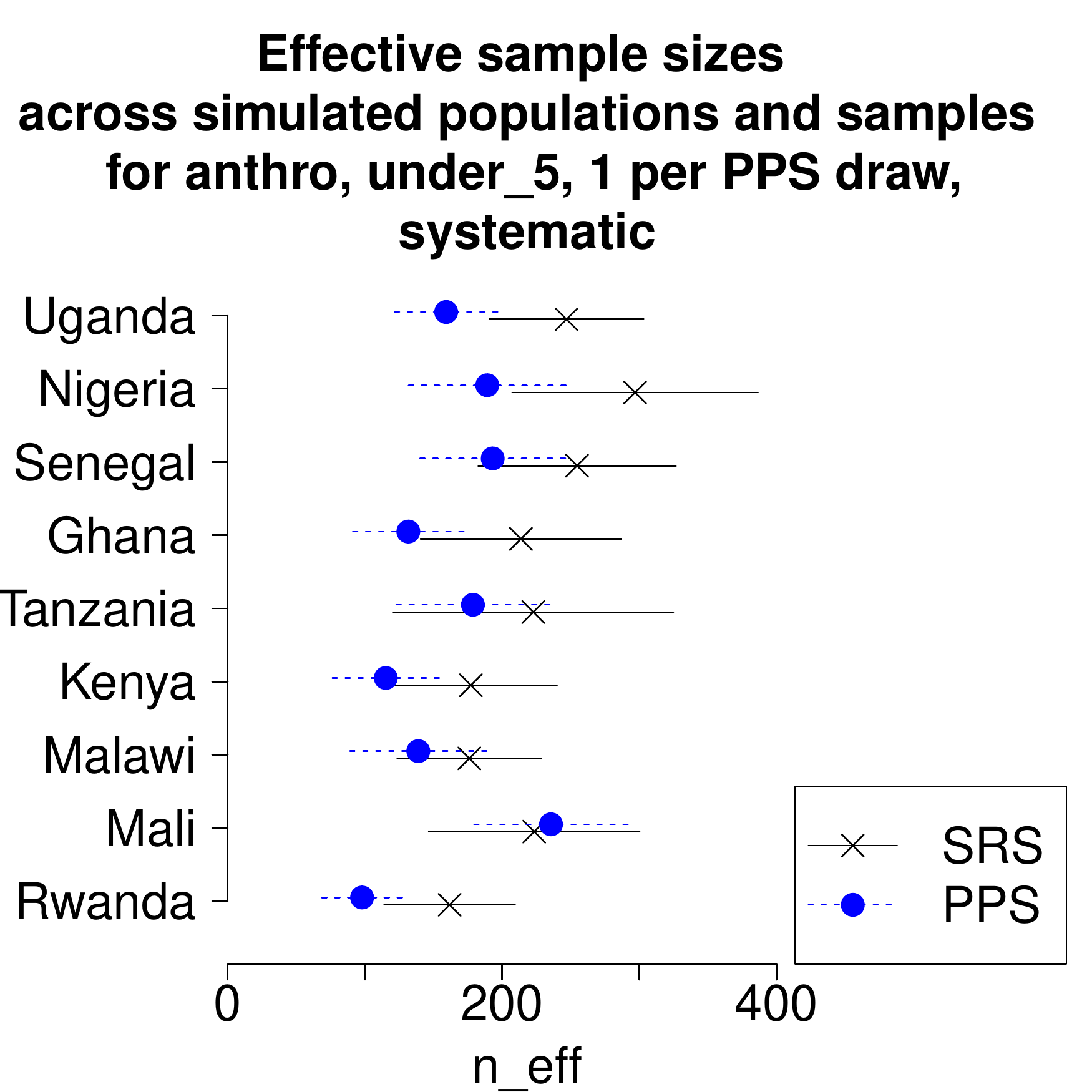}}
\caption[]{Model-based anthropometry module results}\label{anthro_sampling_results_model_based_systematic}
\end{figure}

\begin{figure}[h!]
    \centering
\subcaptionbox{\label{deff_adult_man_1_per_PPS_draw_model_based_stratified}}
 [0.49\textwidth]{\includegraphics[width=0.34\textwidth]{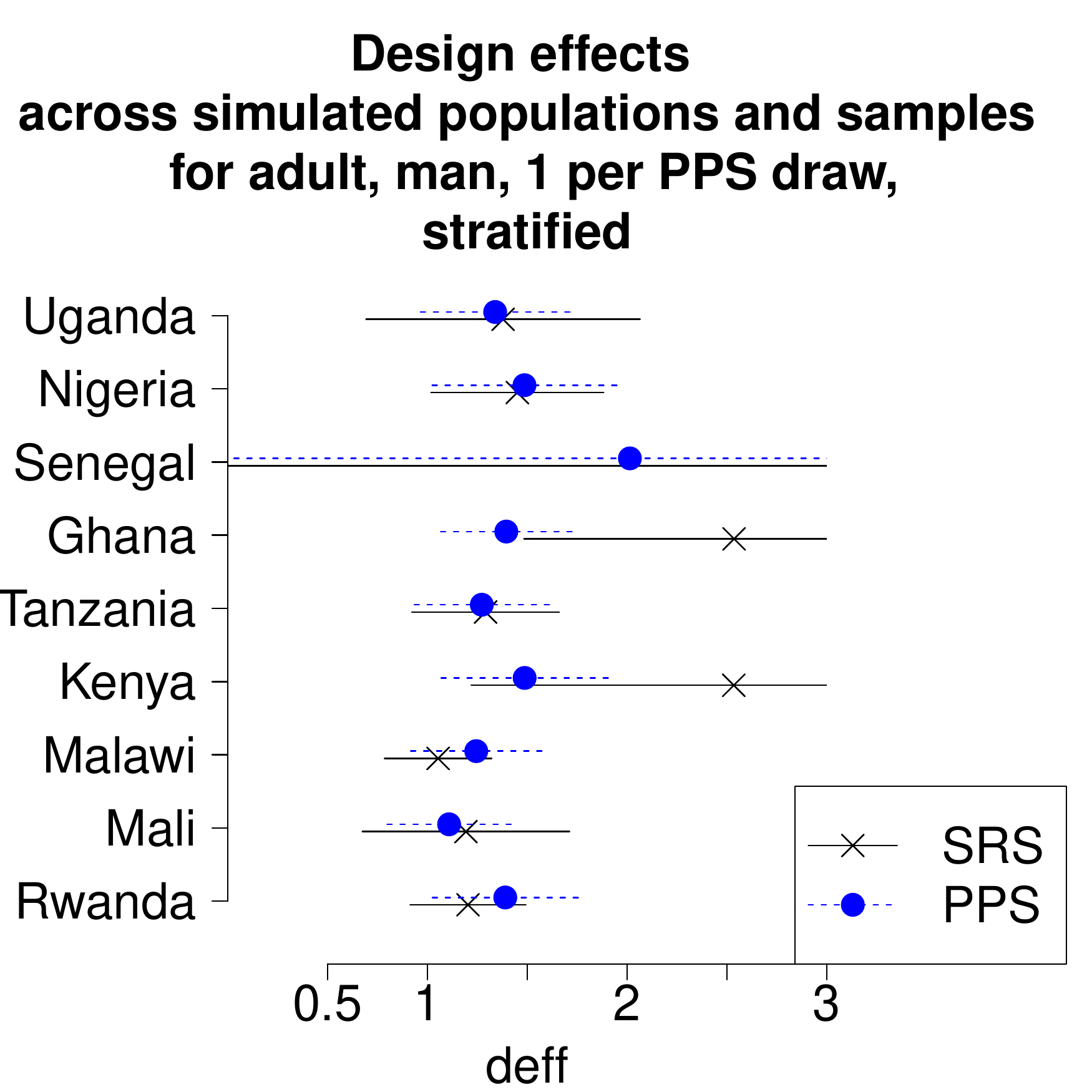}}
\subcaptionbox{\label{deff_adult_woman_1_per_PPS_draw_model_based_stratified}}
 [0.49\textwidth]{\includegraphics[width=0.34\textwidth]{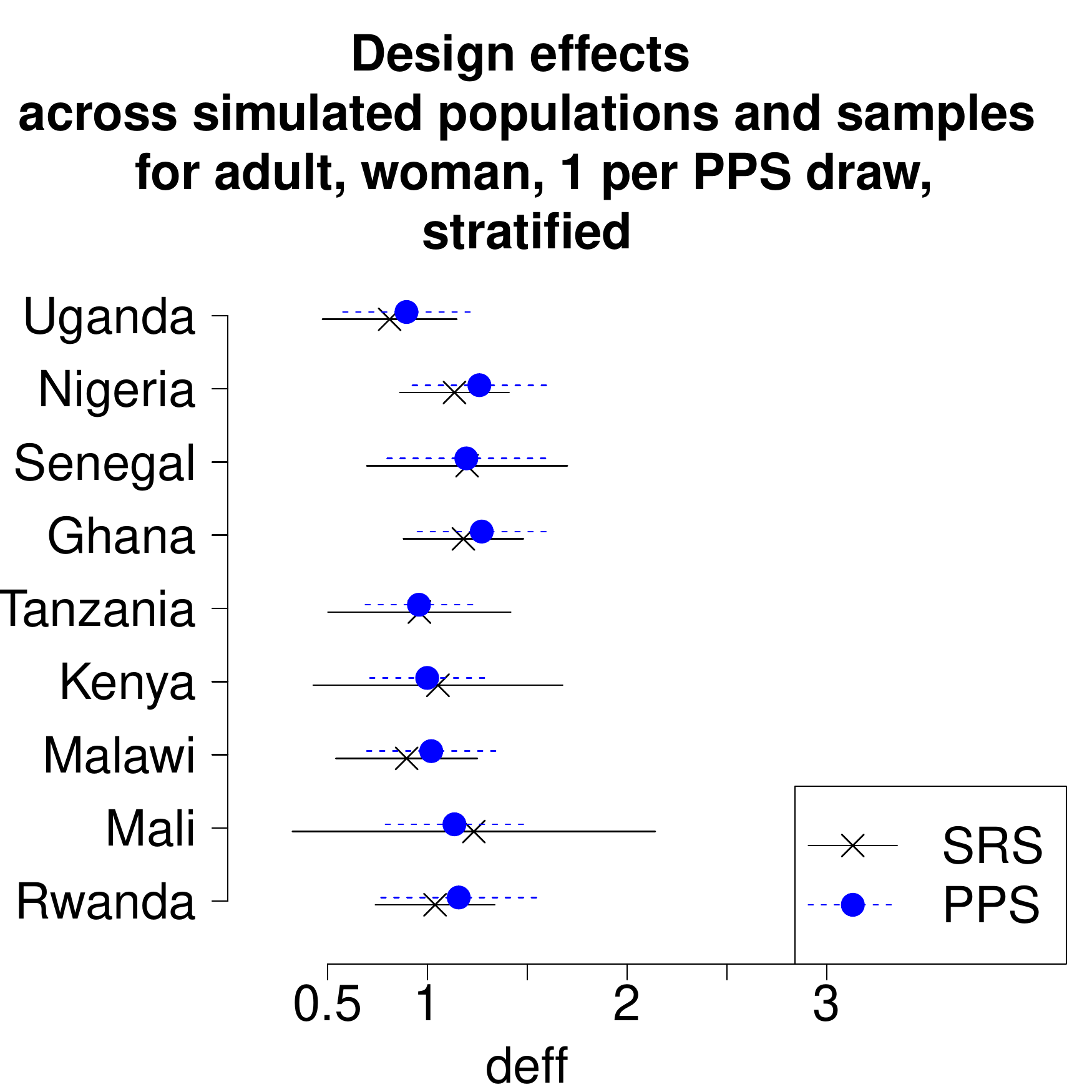}} \\
\subcaptionbox{\label{n_adult_man_1_per_PPS_draw_stratified}}
 [0.49\textwidth]{\includegraphics[width=0.34\textwidth]{n_adult_man_1_per_PPS_draw_stratified.pdf}}
\subcaptionbox{\label{n_adult_woman_1_per_PPS_draw_stratified}}
 [0.49\textwidth]{\includegraphics[width=0.34\textwidth]{n_adult_woman_1_per_PPS_draw_stratified.pdf}} \\
\subcaptionbox{\label{n_eff_adult_man_1_per_PPS_draw_model_based_stratified}}
 [0.49\textwidth]{\includegraphics[width=0.34\textwidth]{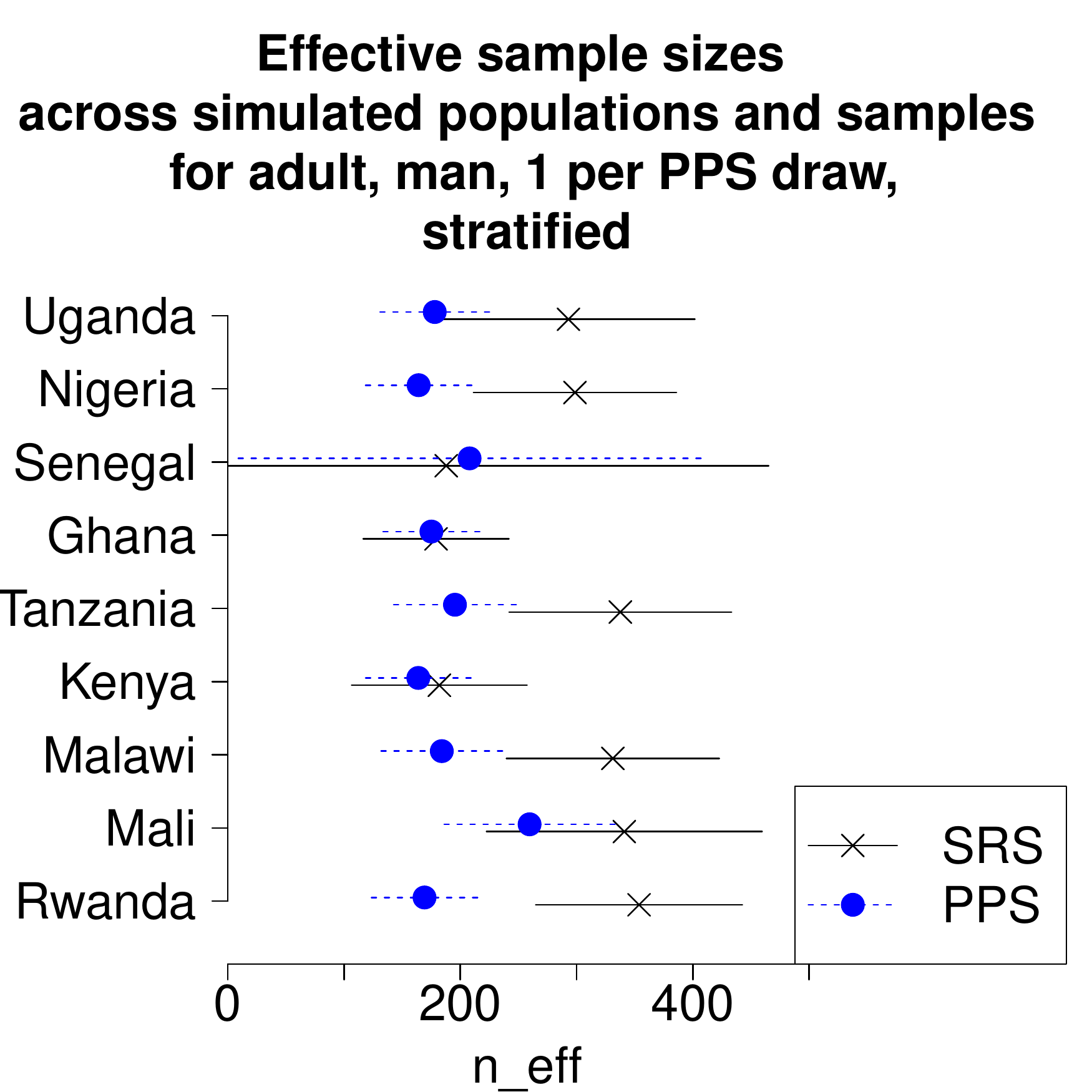}}
\subcaptionbox{\label{n_eff_adult_woman_1_per_PPS_draw_model_based_stratified}}
 [0.49\textwidth]{\includegraphics[width=0.34\textwidth]{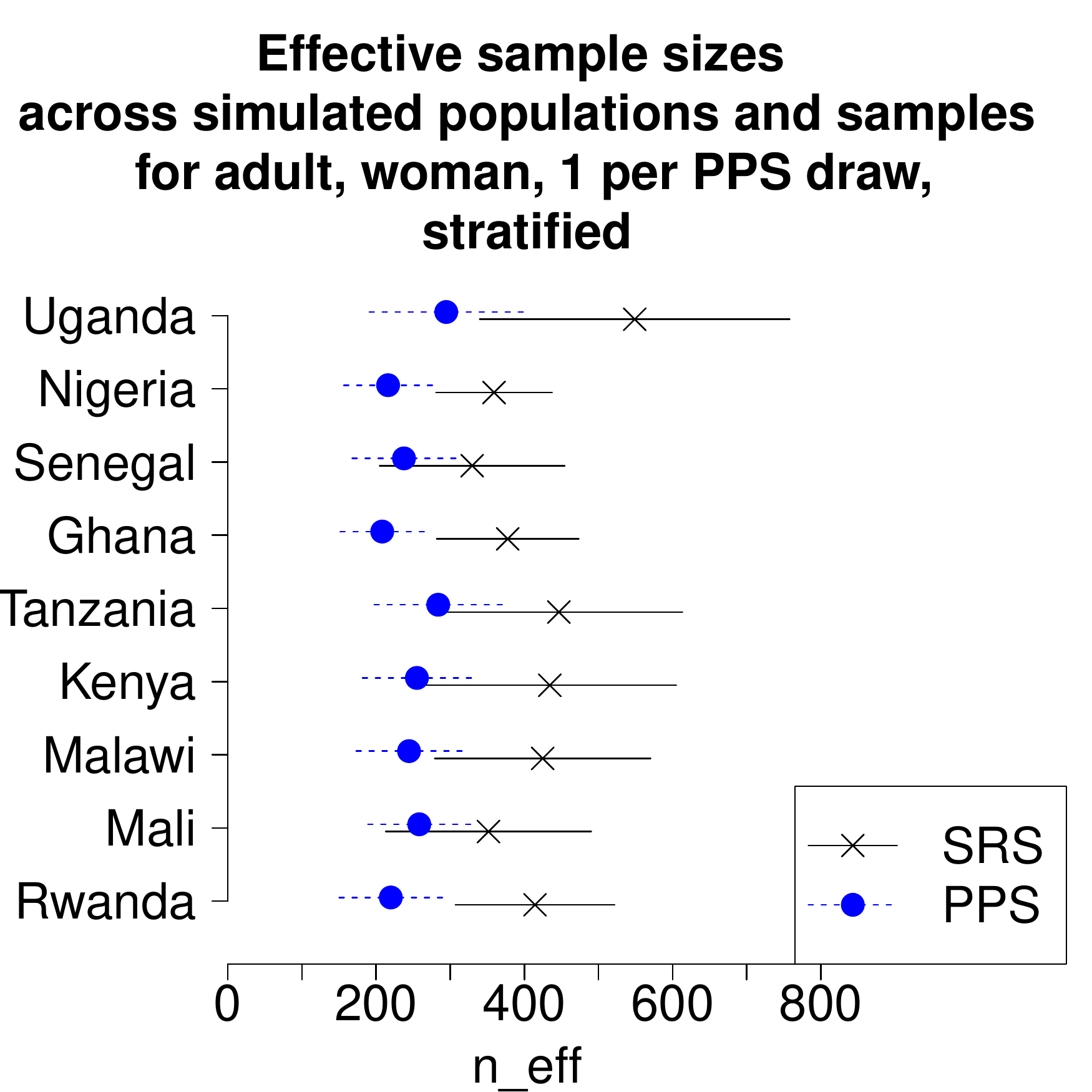}} 
\caption[]{Model-based adult module results}\label{adult_sampling_results_model_based_stratified}
\end{figure}

\begin{figure}[h!]
    \centering
\subcaptionbox{\label{deff_RDT_under_5_1_per_PPS_draw_model_based_stratified}}
 [0.49\textwidth]{\includegraphics[width=0.34\textwidth]{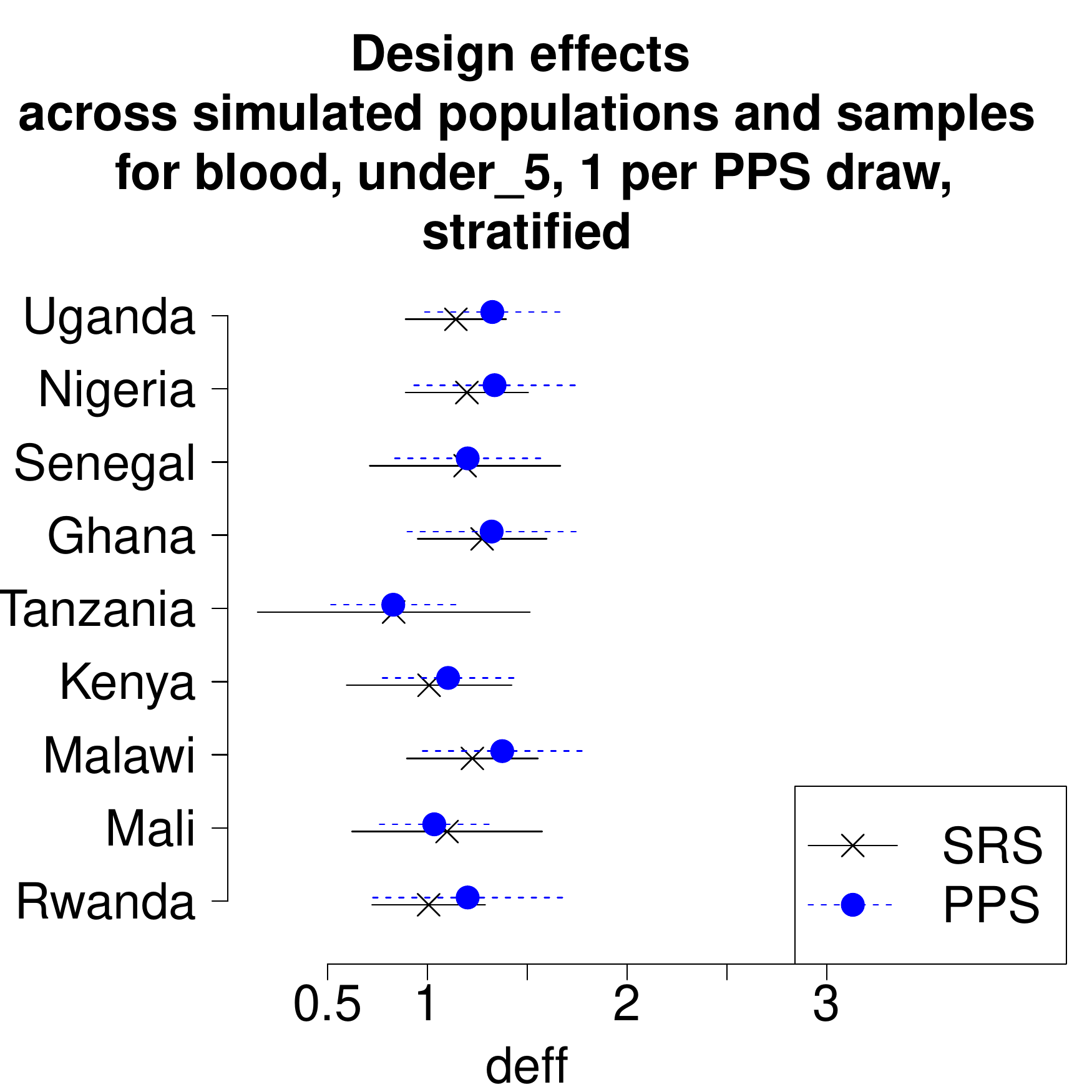}}
\subcaptionbox{\label{deff_RDT_school_age_1_per_PPS_draw_model_based_stratified}}
 [0.49\textwidth]{\includegraphics[width=0.34\textwidth]{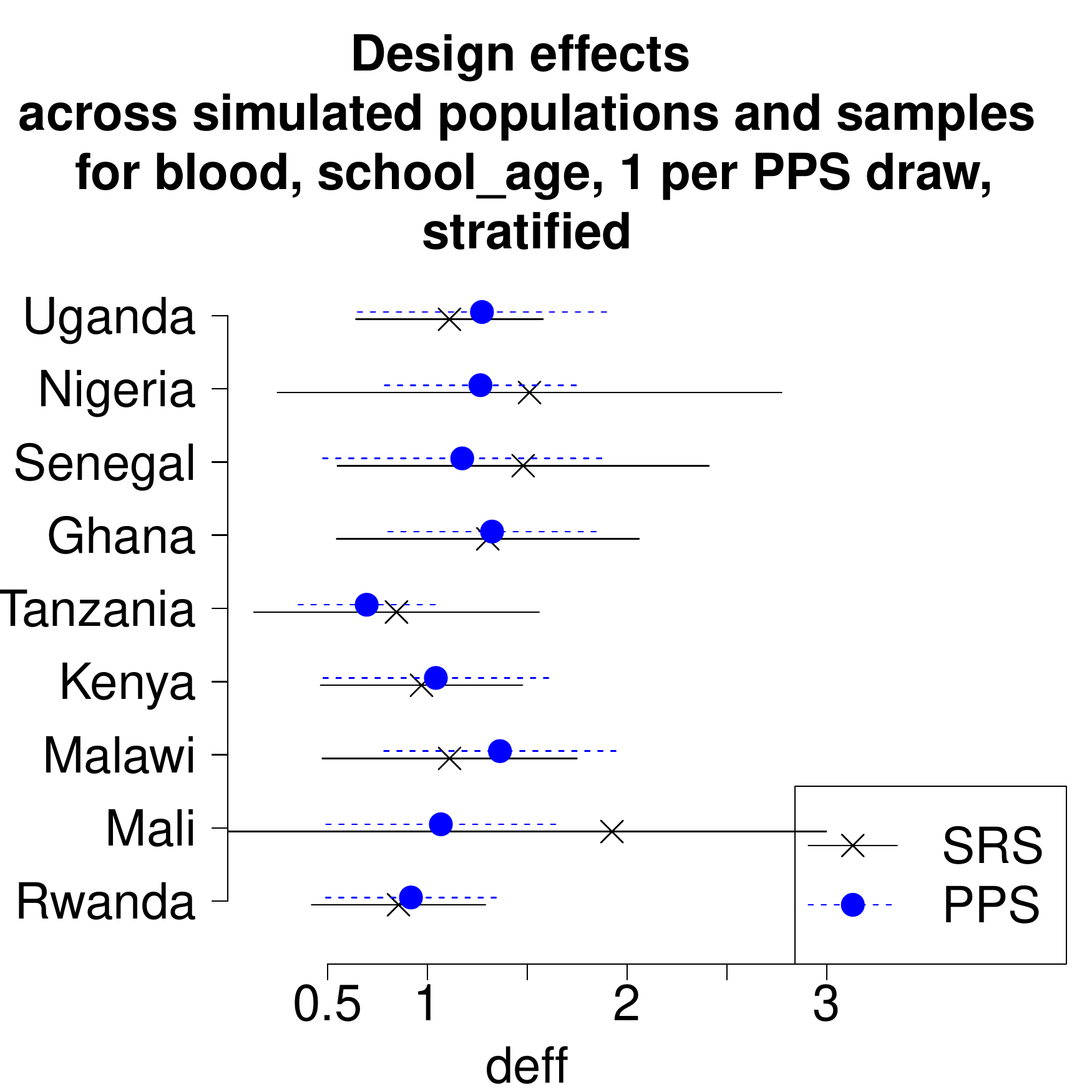}} \\
\subcaptionbox{\label{n_RDT_under_5_1_per_PPS_draw_stratified}}
 [0.49\textwidth]{\includegraphics[width=0.34\textwidth]{n_RDT_under_5_1_per_PPS_draw_stratified.pdf}}
\subcaptionbox{\label{n_RDT_school_age_1_per_PPS_draw_stratified}}
 [0.49\textwidth]{\includegraphics[width=0.34\textwidth]{n_RDT_school_age_1_per_PPS_draw_stratified.pdf}} \\
\subcaptionbox{\label{n_eff_RDT_under_5_1_per_PPS_draw_model_based_stratified}}
 [0.49\textwidth]{\includegraphics[width=0.34\textwidth]{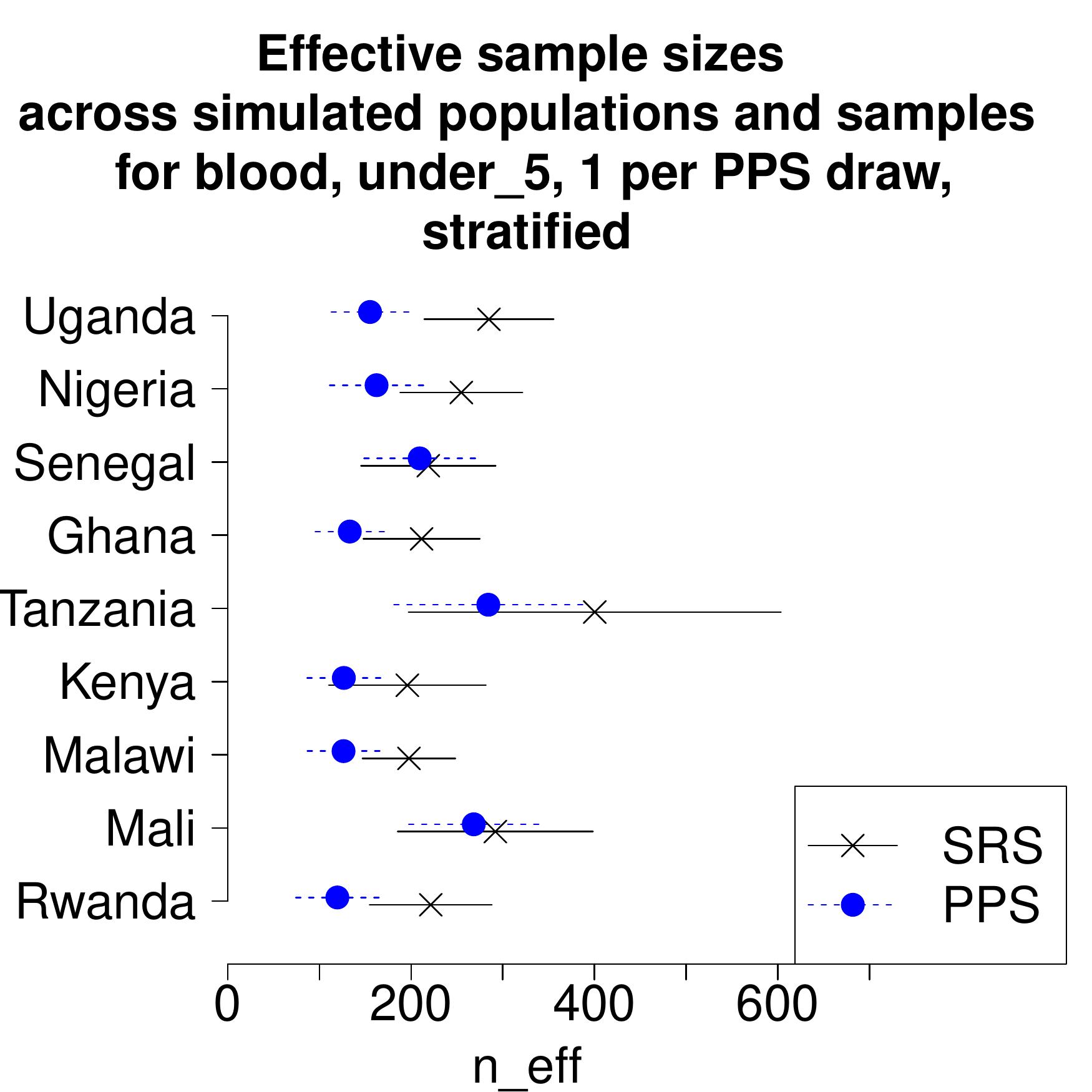}}
\subcaptionbox{\label{n_eff_RDT_school_age_1_per_PPS_draw_model_based_stratified}}
 [0.49\textwidth]{\includegraphics[width=0.34\textwidth]{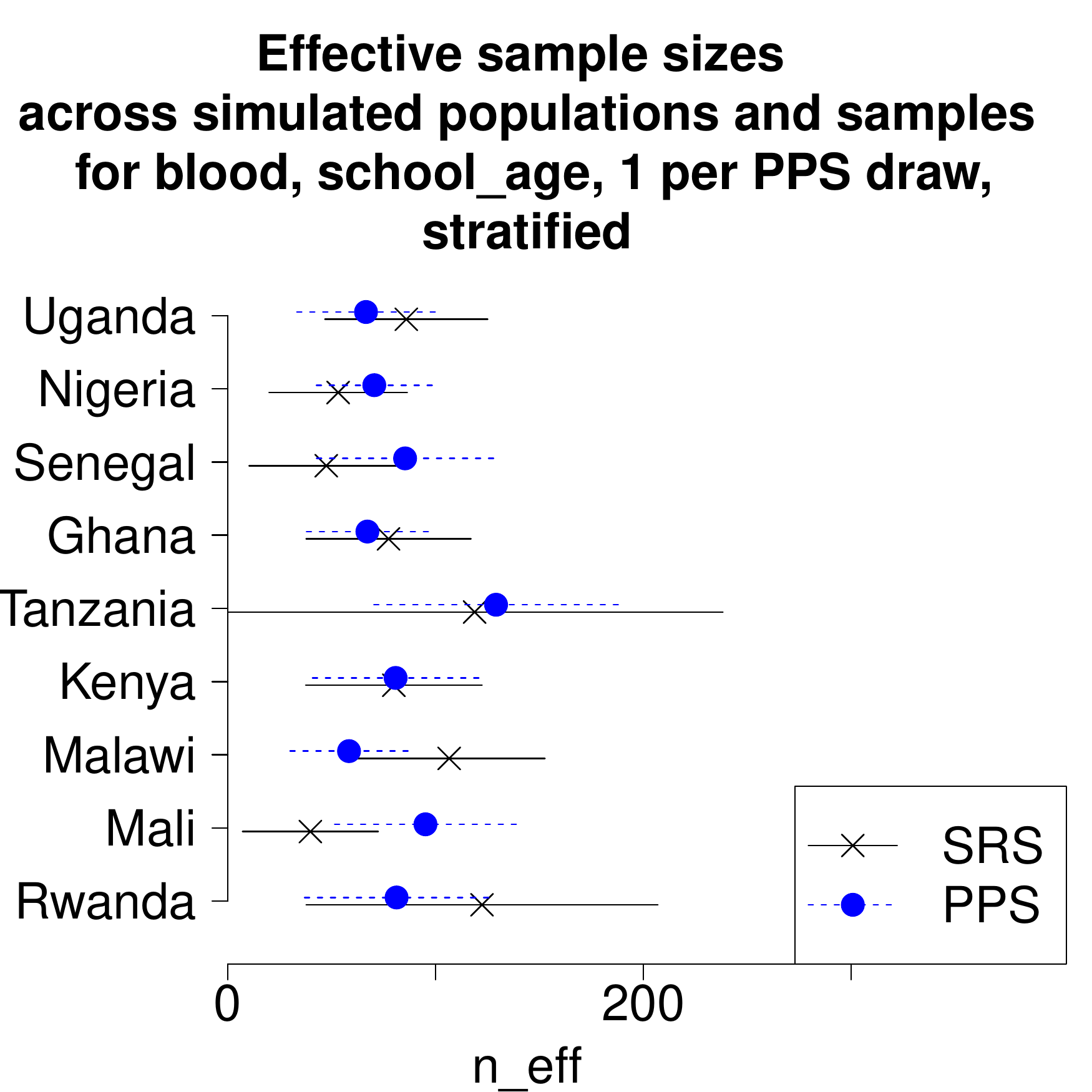}} 
\caption[]{Model-based blood (malaria and anemia) module results: under 5 and school-age children}\label{RDT_children_sampling_results_model_based_stratified}
\end{figure}

\begin{figure}[h!]
    \centering
\subcaptionbox{\label{deff_RDT_man_1_per_PPS_draw_model_based_stratified}}
 [0.49\textwidth]{\includegraphics[width=0.34\textwidth]{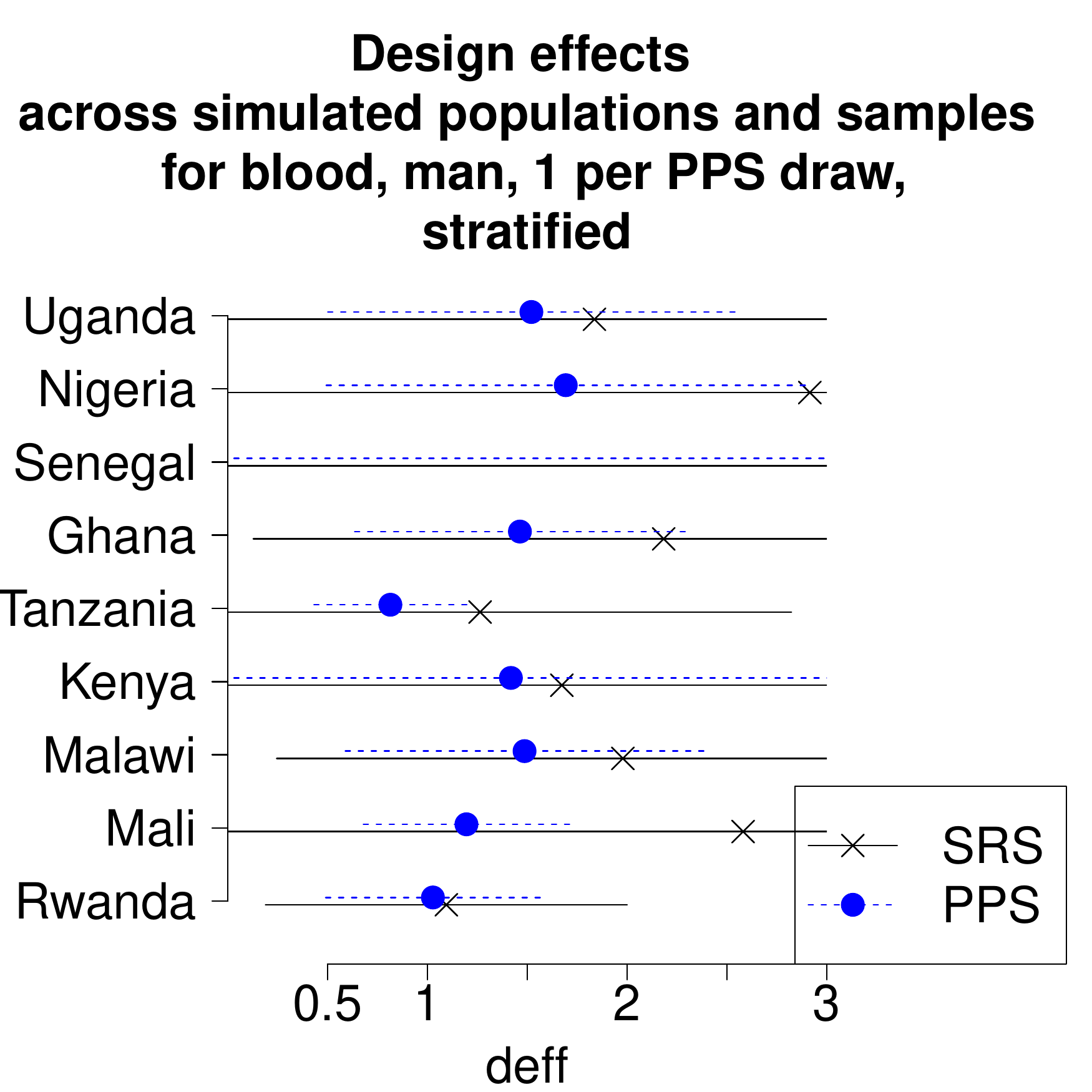}}
\subcaptionbox{\label{deff_RDT_woman_1_per_PPS_draw_model_based_stratified}}
 [0.49\textwidth]{\includegraphics[width=0.34\textwidth]{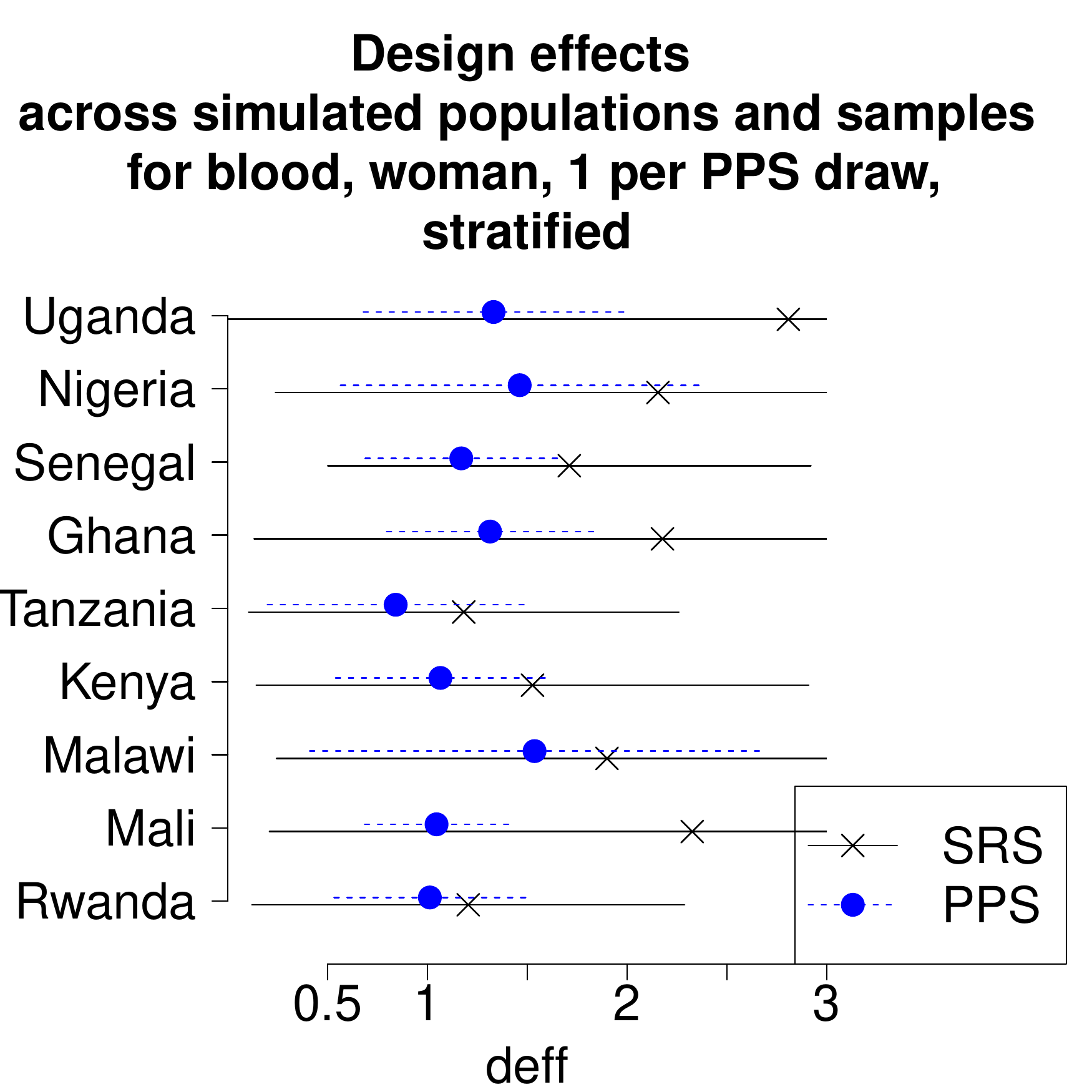}} \\
\subcaptionbox{\label{n_RDT_man_1_per_PPS_draw_stratified}}
 [0.49\textwidth]{\includegraphics[width=0.34\textwidth]{n_RDT_man_1_per_PPS_draw_stratified.pdf}}
\subcaptionbox{\label{n_RDT_woman_1_per_PPS_draw_stratified}}
 [0.49\textwidth]{\includegraphics[width=0.34\textwidth]{n_RDT_woman_1_per_PPS_draw_stratified.pdf}} \\
\subcaptionbox{\label{n_eff_RDT_man_1_per_PPS_draw_model_based_stratified}}
 [0.49\textwidth]{\includegraphics[width=0.34\textwidth]{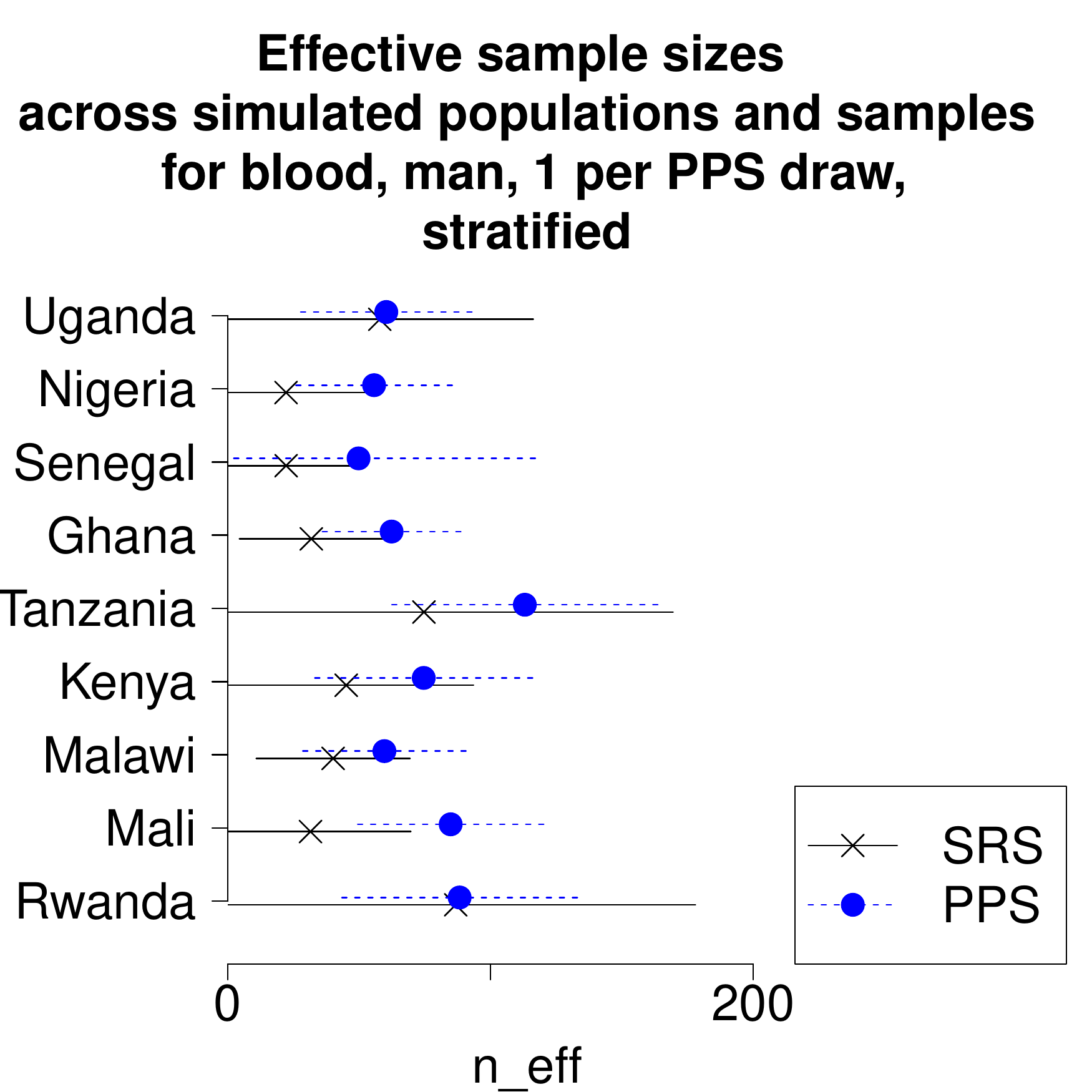}}
\subcaptionbox{\label{n_eff_RDT_woman_1_per_PPS_draw_model_based_stratified}}
 [0.49\textwidth]{\includegraphics[width=0.34\textwidth]{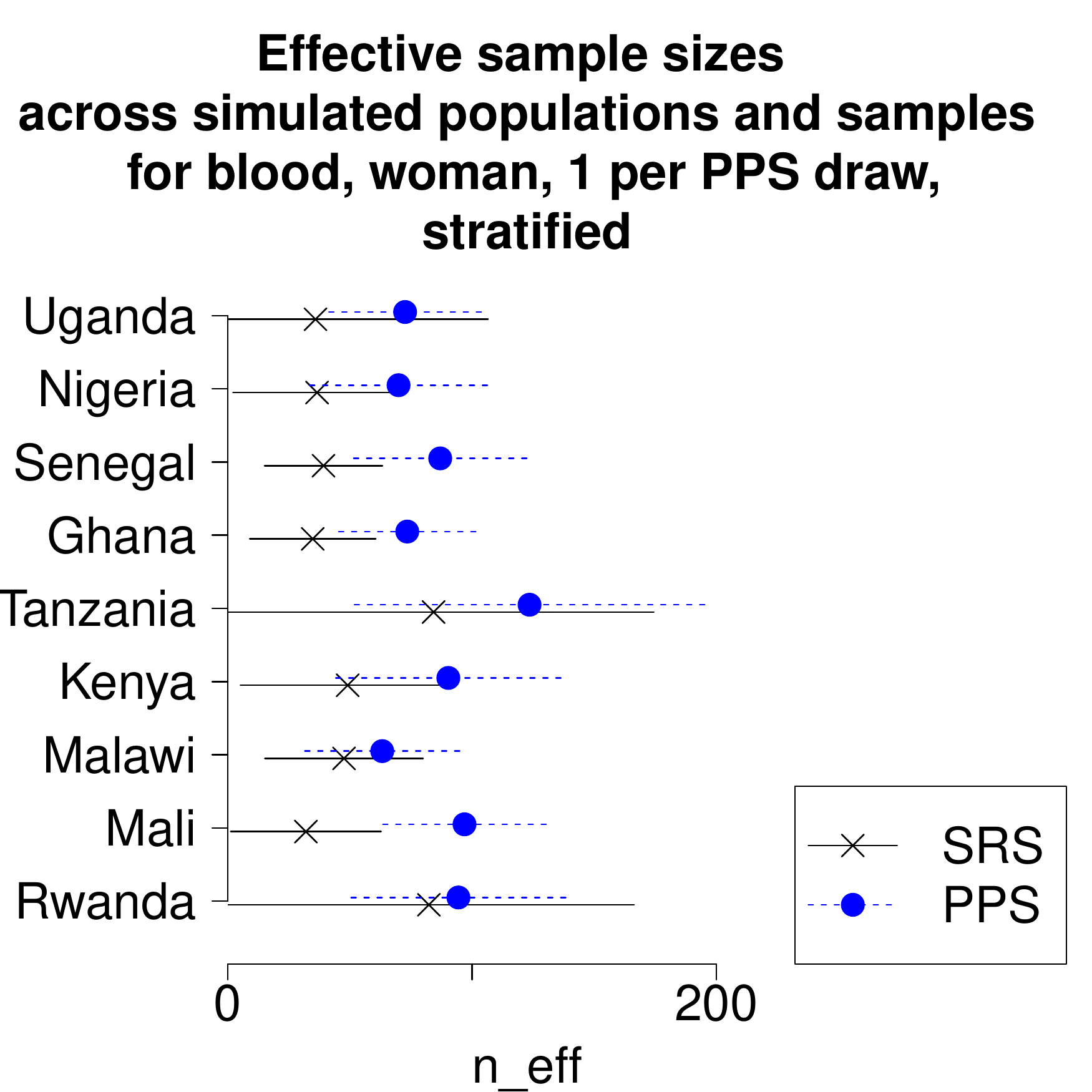}} 
\caption[]{Model-based blood (malaria and anemia) module results: men and women}\label{RDT_man_woman_sampling_results_model_based_stratified}
\end{figure}

\begin{figure}[h!]
    \centering
\subcaptionbox{\label{deff_anthro_under_5_1_per_PPS_draw_model_based_stratified}}
 [0.49\textwidth]{\includegraphics[width=0.34\textwidth]{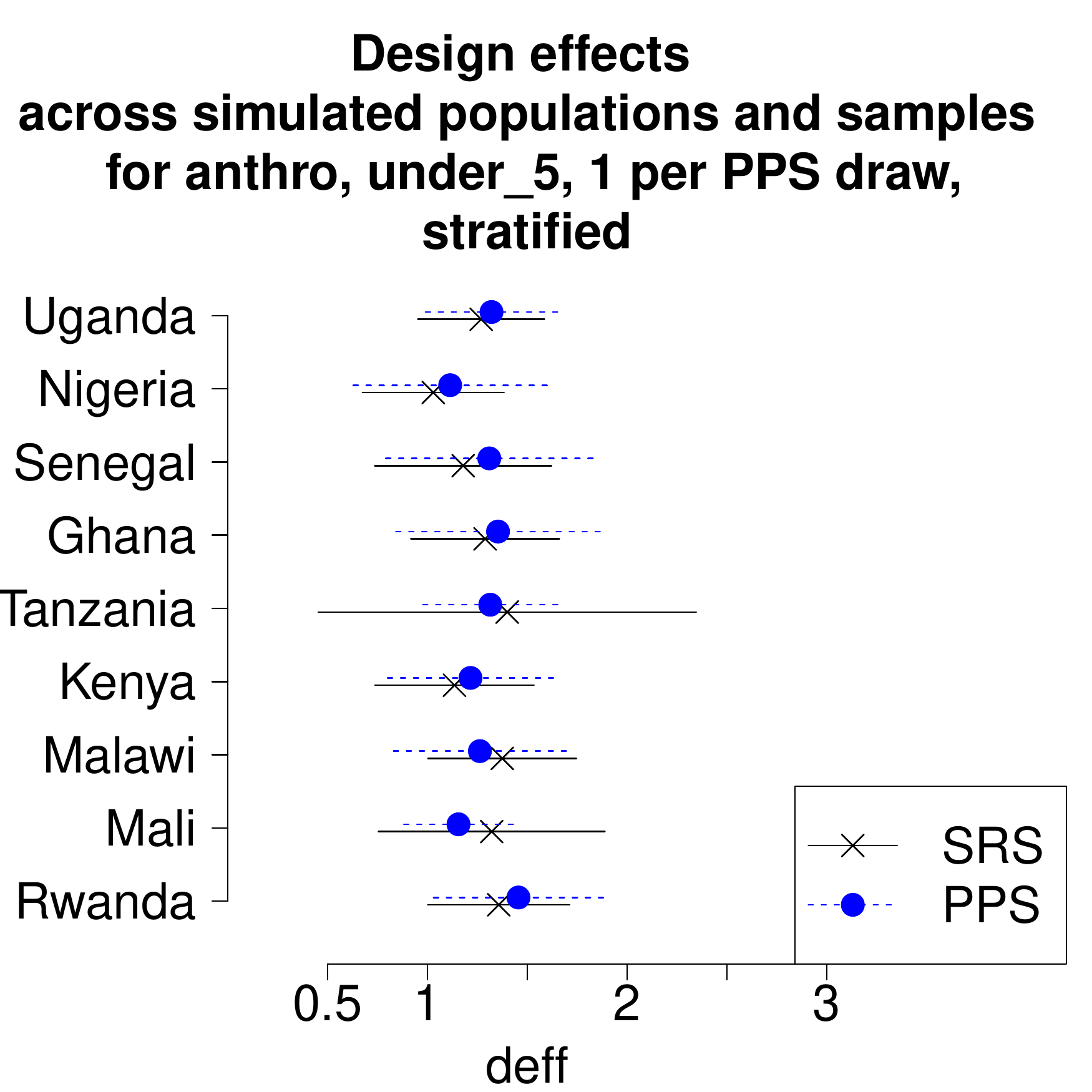}} \\
\subcaptionbox{\label{n_anthro_under_5_1_per_PPS_draw_stratified}}
 [0.49\textwidth]{\includegraphics[width=0.34\textwidth]{n_anthro_under_5_1_per_PPS_draw_stratified.pdf}} \\
\subcaptionbox{\label{n_eff_anthro_under_5_1_per_PPS_draw_model_based_stratified}}
 [0.49\textwidth]{\includegraphics[width=0.34\textwidth]{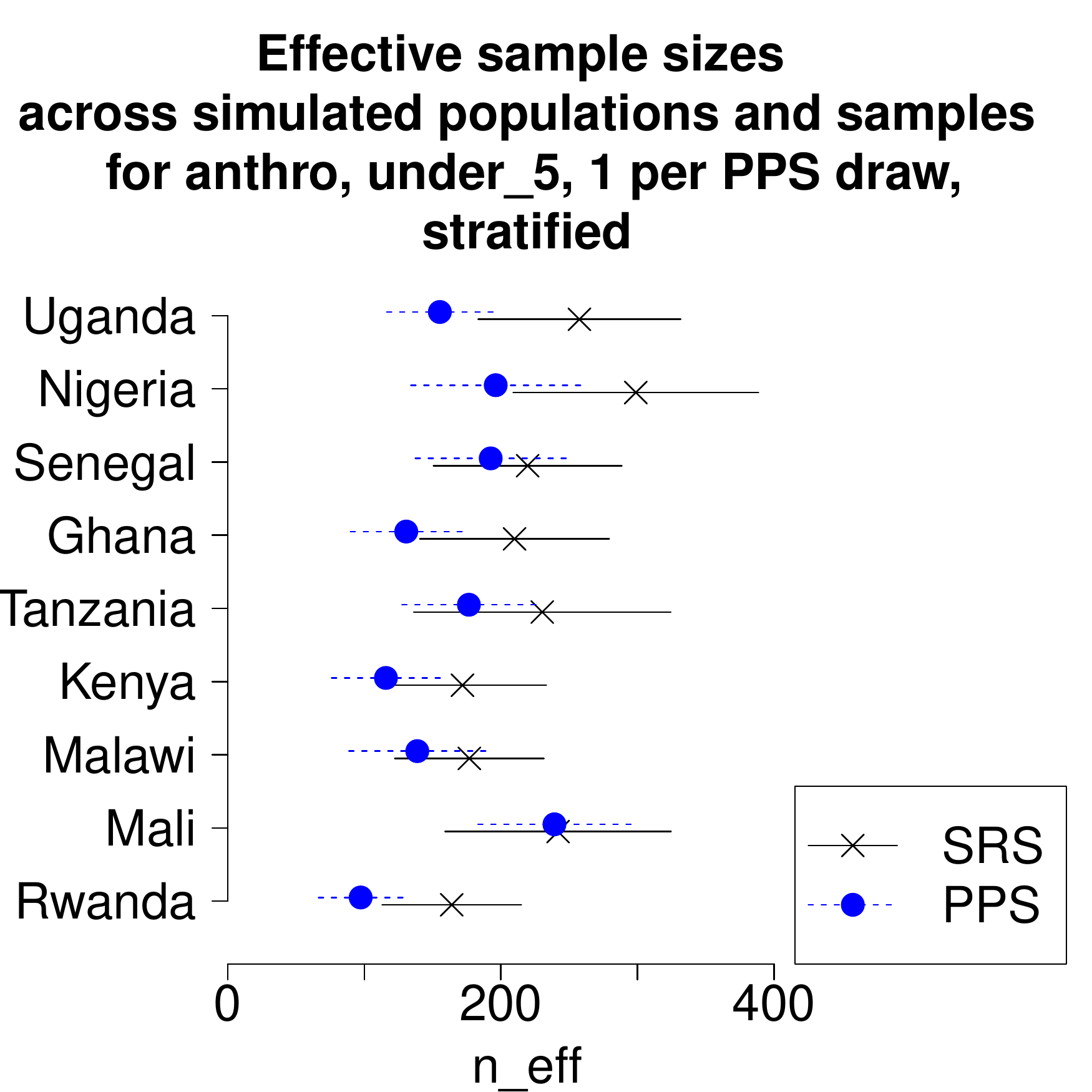}}
\caption[]{Model-based anthropometry module results}\label{anthro_sampling_results_model_based_stratified}
\end{figure}

\begin{figure}[h!]
\centering
\includegraphics[width = 8 cm]{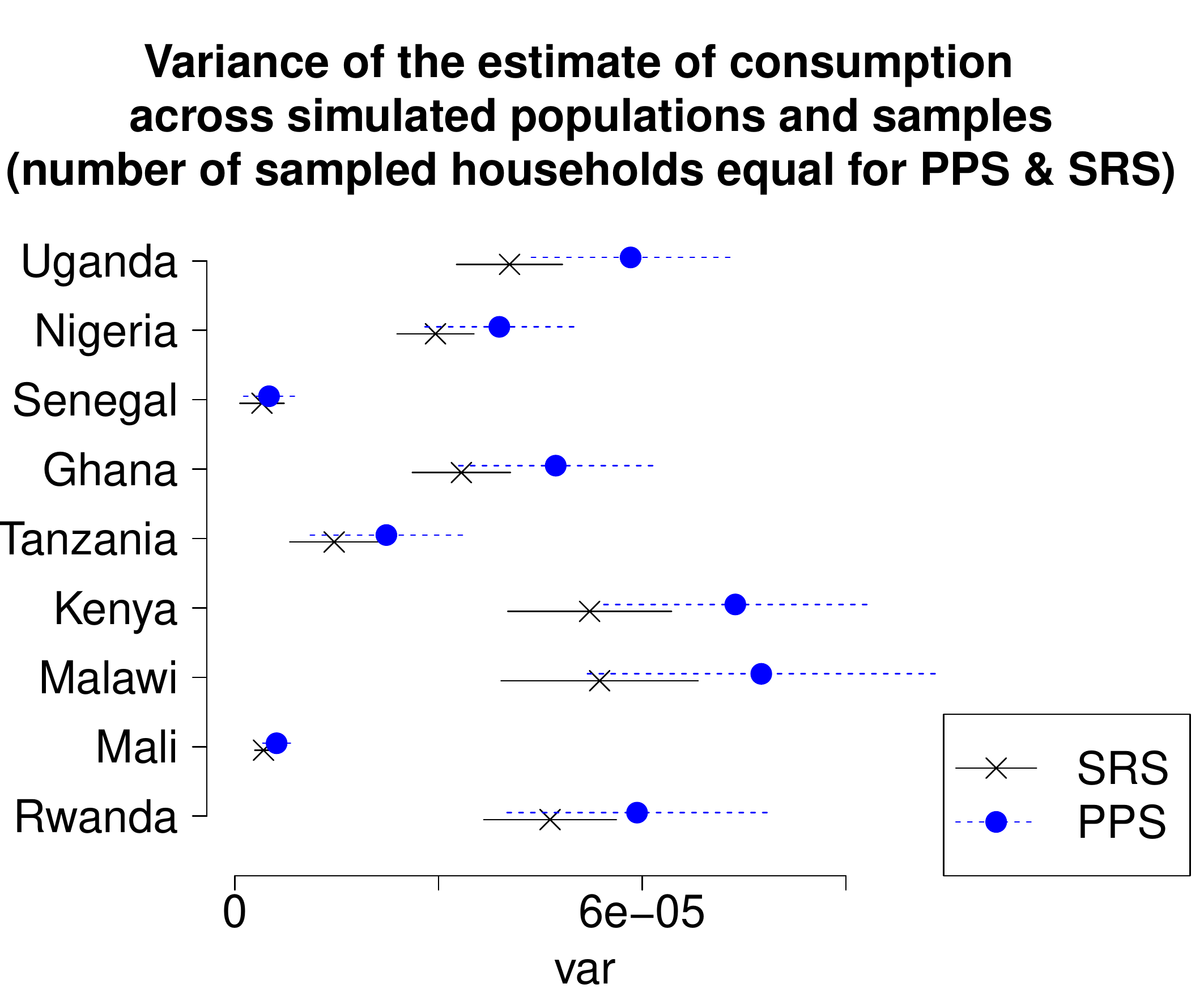}\caption{Model-based consumption module results}\label{var_hhsurvey_model_based_systematic}
\end{figure}

\clearpage
\newpage

\bibliographystyle{plainnat}
\bibliography{bibliography}

\end{document}